\documentclass[11pt]{article}
\def\withcolors{1}
\def\withnotes{0}
\def\withindex{0}
 \makeatletter{}\usepackage[T1]{fontenc}
\usepackage[utf8]{inputenc}

\usepackage{lmodern}
\usepackage{xspace}
\usepackage[normalem]{ulem}

\usepackage{amsfonts,amsmath,amssymb, amsthm, mathtools}
\usepackage[mathscr]{euscript} \usepackage{thm-restate}
\usepackage{dsfont} 
\usepackage{algorithmicx,algpseudocode,algorithm}

\usepackage[usenames,dvipsnames,table]{xcolor}

\usepackage{csquotes}

\usepackage{relsize}

\usepackage{multirow}
\usepackage{chngpage} 
\usepackage{mfirstuc}

\usepackage{tikz}
\usetikzlibrary{arrows}
\usetikzlibrary{calc,decorations.pathmorphing,patterns}

\ifnum\withindex=1
  \usepackage{makeidx}
  \usepackage{ifthen}
  \newcommand\indexed[2][]{\ifthenelse{\equal{#1}{}}{#2\index{#2}}{#2\index{#1}}}
  \makeindex \fi

\usepackage[colorlinks,citecolor=blue,bookmarks=true,linktocpage]{hyperref}
\usepackage{aliascnt}
\usepackage[numbered]{bookmark}

\usepackage{fullpage}

\usepackage{titling}

\usepackage[shortlabels]{enumitem}
  \setitemize{noitemsep,topsep=3pt,parsep=2pt,partopsep=2pt}   \setenumerate{itemsep=1pt,topsep=2pt,parsep=2pt,partopsep=2pt}
  \setdescription{itemsep=1pt}
  
\ifnum\withnotes=1
  \usepackage[colorinlistoftodos,textsize=scriptsize]{todonotes}
\fi

\usepackage{verbatim}

\usepackage{mleftright}  

\makeatletter
\@ifundefined{theorem}{    \theoremstyle{plain}   	\newtheorem{theorem}{Theorem}[section]
  	\newaliascnt{coro}{theorem}
  	  \newtheorem{corollary}[coro]{Corollary}
  	\aliascntresetthe{coro}
  	\newaliascnt{lem}{theorem}
  		\newtheorem{lemma}[lem]{Lemma}
  	\aliascntresetthe{lem}
  	\newaliascnt{clm}{theorem}
  		\newtheorem{claim}[clm]{Claim}
	\aliascntresetthe{clm}
	\newaliascnt{fact}{theorem}
 	 	\newtheorem{fact}[theorem]{Fact}
	\aliascntresetthe{fact}
  	
  \newaliascnt{prop}{theorem}
  		\newtheorem{proposition}[prop]{Proposition}
	\aliascntresetthe{prop}
	\newaliascnt{conj}{theorem}
  		
	\aliascntresetthe{conj}
  \theoremstyle{remark}

  \theoremstyle{definition}   	\newaliascnt{defn}{theorem}
 		 \newtheorem{definition}[defn]{Definition}
 	 \aliascntresetthe{defn}
}{}
\makeatother

\newenvironment{proofof}[1]{\begin{proof}[Proof of {#1}]}{\end{proof}}

\providecommand{\email}[1]{\href{mailto:#1}{\nolinkurl{#1}\xspace}}

\ifnum\withcolors=1
  \newcommand{\new}[1]{{\color{red} {#1}}}   \newcommand{\newer}[1]{{\color{blue} {#1}}}   \newcommand{\newest}[1]{{\color{orange} {#1}}}                \else
  \newcommand{\new}[1]{{{#1}}}
  \newcommand{\newer}[1]{{{#1}}}
  \newcommand{\newest}[1]{{{#1}}}

\fi

\ifnum\withnotes=1
\else
  \newcommand{\cnote}[1]{}
  \newcommand{\tnote}[1]{}
  \newcommand{\inote}[1]{}
  \newcommand{\rnote}[1]{}
  \newcommand{\todonote}[2][]{}
	\newcommand{\questionnote}[2][]{}
	\newcommand{\todonotedone}[2][]{}
	\newcommand{\todonoteinline}[2][]{}
  \newcommand{\marginnote}[1]{}
\fi
\newcommand{\ignore}[1]{\leavevmode\unskip}     
\newcommand{\eps}{\ensuremath{\varepsilon}\xspace}
 \newcommand{\Tester}{\ensuremath{\mathcal{T}}\xspace} \newcommand{\Learner}{\ensuremath{\mathcal{L}}\xspace} \newcommand{\property}{\ensuremath{\mathcal{P}}\xspace} \newcommand{\class}{\ensuremath{\mathcal{C}}\xspace} \newcommand{\eqdef}{\stackrel{\rm def}{=}}

\newcommand{\accept}{\textsf{ACCEPT}\xspace}
\newcommand{\fail}{\textsf{FAIL}\xspace}
\newcommand{\reject}{\textsf{REJECT}\xspace}
\newcommand{\opt}{{\textsc{opt}}\xspace}

\newcommand{\domain}{\ensuremath{\Omega}\xspace}  \newcommand{\yes}{{\sf{}yes}\xspace}
\newcommand{\no}{{\sf{}no}\xspace}

\newcommand{\littleO}[1]{{o\mleft( #1 \mright)}}
\newcommand{\bigO}[1]{{O\mleft( #1 \mright)}}
\newcommand{\bigTheta}[1]{{\Theta\mleft( #1 \mright)}}
\newcommand{\bigOmega}[1]{{\Omega\mleft( #1 \mright)}}
\newcommand{\tildeO}[1]{\tilde{O}\mleft( #1 \mright)}

\providecommand{\poly}{\operatorname*{poly}}

\newcommand{\setOfSuchThat}[2]{ \left\{\; #1 \;\colon\; #2\; \right\} } 			                                                                                           
\newcommand{\dtv}{\operatorname{d_{\rm TV}}}
\newcommand{\totalvardist}[2]{{\dtv\!\left({#1, #2}\right)}}

\newcommand{\kolmogorov}[2]{{\norm{#1 - #2}_{\rm Kol}}}

\newcommand{\dist}[2]{{\operatorname{dist}\!\left({#1, #2}\right)}}

\newcommand\restr[2]{{  \left.\kern-\nulldelimiterspace   #1   \vphantom{\big|}   \right|_{#2}   }}

\newcommand{\proba}{\Pr}
\newcommand{\probaOf}[1]{\proba\!\left[\, #1\, \right]}

\newcommand{\shortexpect}{\mathbb{E}}
\newcommand{\var}{\operatorname{Var}}

\newcommand{\uniform}{\ensuremath{\mathcal{U}}}

\newcommand{\binomial}[2]{\ensuremath{\operatorname{Bin}\!\left( #1, #2 \right)}}
\newcommand{\poisson}[1]{\ensuremath{\operatorname{Poisson}\!\left( #1 \right) }}

\newcommand{\norm}[1]{\lVert#1{\rVert}}
\newcommand{\normone}[1]{{\norm{#1}}_1}
\newcommand{\normtwo}[1]{{\norm{#1}}_2}
\newcommand{\norminf}[1]{{\norm{#1}}_\infty}
\newcommand{\abs}[1]{\left\lvert #1 \right\rvert}
\newcommand{\dabs}[1]{\lvert #1 \rvert}

\newcommand{\clg}[1]{\left\lceil #1 \right\rceil}

\newcommand{\R}{\ensuremath{\mathbb{R}}\xspace}

\newcommand{\N}{\ensuremath{\mathbb{N}}\xspace}

\newcommand{\SAMP}{{\sf SAMP}\xspace}

\newcommand{\pdfsamp}{dual\xspace}
\newcommand{\cdfsamp}{cumulative dual\xspace}
\newcommand{\Pdfsamp}{\expandafter\capitalisewords\expandafter{\pdfsamp}}
\newcommand{\Cdfsamp}{\expandafter\capitalisewords\expandafter{\cdfsamp}}

\newcommand{\lp}[1][1]{\ell_{#1}}

\newcommand{\D}{\ensuremath{D}}

\makeatletter

\newcommand{\Rom}[1]{\expandafter\@slowromancap\romannumeral #1@}

\makeatother

\makeatletter
\newcommand\ackname{Acknowledgements}
\if@titlepage
  \newenvironment{acknowledgements}{      \titlepage
      \null\vfil
      \@beginparpenalty\@lowpenalty
      \begin{center}        \bfseries \ackname
        \@endparpenalty\@M
      \end{center}}     {\par\vfil\null\endtitlepage}
\else
  
\fi
\makeatother

\newcommand{\classmon}{\ensuremath{\mathcal{M}}\xspace}
\newcommand{\classtmo}[1][t]{\ensuremath{\classmon_{#1}}\xspace}
\newcommand{\classuni}{\classtmo[1]}
\newcommand{\classlog}{\ensuremath{\mathcal{L}}\xspace}
\newcommand{\classcvx}{\ensuremath{\mathcal{K}^+}\xspace}
\newcommand{\classcve}{\ensuremath{\mathcal{K}^-}\xspace}
\newcommand{\classmhr}{\ensuremath{\mathcal{MHR}}\xspace}
\newcommand{\classpbd}{\ensuremath{\mathcal{PBD}}\xspace}
\newcommand{\classbin}{\ensuremath{\mathcal{BIN}}\xspace}
\newcommand{\classpoly}[1][t,d]{\ensuremath{\mathcal{P}_{#1}}\xspace}
\newcommand{\classhist}[1][t]{\ensuremath{\mathcal{H}_{#1}}\xspace}
\newcommand{\classksiirv}[1][k]{\ensuremath{#1}\text{-}\ensuremath{\mathcal{SIIRV}}\xspace}

\newcommand{\estimdist}[1][\class]{\textsc{ProjectionDist}_{#1}}

\newcounter{IRL}
\renewcommand{\theIRL}{\textsf{Succinctness}}\newcommand*{\inlineref}[1]{\refstepcounter{IRL}({\theIRL})\label{#1}}

\def\authornamecc{Cl\'ement L. Canonne}
\def\authorafficc{Columbia University. Email: \email{ccanonne@cs.columbia.edu}. Research supported by NSF CCF-1115703 and NSF CCF-1319788.}
\def\authornamedi{Ilias Diakonikolas}
\def\authoraffidi{University of Edinburgh. Email: \email{ilias.d@ed.ac.uk}. Research supported by EPSRC grant EP/L021749/1, a Marie Curie Career Integration Grant, and a SICSA grant. This work was performed in part while visiting CSAIL, MIT.}
\def\authornamegt{Themis Gouleakis}
\def\authoraffigt{CSAIL, MIT. Email: \email{tgoule@mit.edu}.}
\def\authornamerr{Ronitt Rubinfeld}
\def\authoraffirr{CSAIL, MIT and the Blavatnik School of Computer Science, Tel Aviv University. Email: \email{ronitt@csail.mit.edu}.}

\title{Testing Shape Restrictions of Discrete Distributions}

\author{
  \authornamecc\thanks{\authorafficc}
  \and \authornamedi\thanks{\authoraffidi}
  \and \authornamegt\thanks{\authoraffigt}
  \and \authornamerr\thanks{\authoraffirr}
}

\makeatletter   \hypersetup{
    pdftitle={\@title},
    pdfauthor={\authornamecc{}, \authornamedi{}, \authornamegt{} and \authornamerr}, 
    pdfsubject={Probability Distribution Testing},
    pdfkeywords={property testing} {distribution testing} {sublinear algorithms}
  }
\makeatother

\begin{document}

 \maketitle
 
\begin{abstract}
  \makeatletter{}We study the question of testing \emph{structured} properties (classes) of discrete distributions. Specifically, given 
sample access to an arbitrary distribution $\D$ over $[n]$ and a property $\property$, the goal is to distinguish 
between $\D\in\property$ and $\lp[1](\D,\property)>\eps$. 
We develop a general algorithm for this question, which applies to a large range of ``shape-constrained'' properties, 
including monotone, log-concave, $t$-modal, piecewise-polynomial, and Poisson Binomial distributions. Moreover, for all cases 
considered, our algorithm has near-optimal sample complexity with regard to the domain size and is computationally efficient.  
For most of these classes, we provide the first non-trivial tester in the literature. 
In addition, we also describe a generic method to prove lower bounds for this problem, and use it to show our upper bounds are nearly tight. 
Finally, we extend some of our techniques to tolerant testing, deriving nearly--tight upper and lower bounds for the corresponding questions.
 
\end{abstract}
\pagenumbering{gobble}\setcounter{page}{1}

\pagenumbering{arabic}
\section{Introduction}\label{sec:introduction}
\makeatletter{}Inferring information about the probability distribution that underlies a data sample is an essential question in Statistics, and one that has ramifications in every field of the natural sciences and quantitative research. In many situations, it is natural to assume that this data exhibits some simple structure because of known properties of the origin of the data, and in fact these assumptions are crucial in making the problem tractable. Such assumptions translate as constraints on the probability distribution -- e.g., it is supposed to be Gaussian, or to meet a smoothness or ``fat tail'' condition (see e.g.,~\cite{Mandelbrot:63:FatTail,Hougaard:86:StableDistribs,PhysRevLett:95}).

As a result, the problem of deciding whether a distribution possesses such a structural property has been widely investigated both in theory and practice, in the context of \emph{shape restricted inference}~\cite{BBBB:72,SS:01} and \emph{model selection}~\cite{MP:03}. Here, it is guaranteed or thought that the unknown distribution satisfies a shape constraint, such as having a monotone or log-concave probability density function~\cite{SN:99,BB:05,Wal:09,Diakonikolas:CRC}.
From a different perspective, a recent line of work in Theoretical Computer Science, originating from the \ignore{seminal} papers of Batu et al.~\cite{BFRSW:00,BFFKRW:01,GRexp:00} has also been tackling similar questions in the setting of property testing (see~\cite{Ron:08,Ron:10,Rubinfeld:12:Taming,Canonne:15:BlueData} for surveys on this field). This very active area has seen a spate of results and breakthroughs over the past decade, culminating in very efficient (both sample and time-wise) algorithms for a wide range of distribution testing problems~\cite{BDKR:05,GMV:06,Alon:2007,DDSV:13,CDVV:14,AD:15,DKN:15}. In many cases, this led to a tight characterization of the number of samples required for these tasks as well as the development of new tools and techniques, drawing connections to learning and information theory~\cite{ValiantValiant:10lb,VV:11:stoc,VV:14}.

In this paper, we focus on the following general property testing problem: given a class (property) of distributions \property and sample access to an \emph{arbitrary} distribution $\D$, one must distinguish between the case that \textsf{(a)} $\D\in\property$, versus \textsf{(b)} $\normone{\D-\D^\prime} >\eps$ for all $\D^\prime\in\property$ (i.e., $\D$ is either in the class, or far from it). While many of the previous works have focused on the testing of specific properties of distributions or obtained algorithms and lower bounds on a case-by-case basis, an emerging trend in distribution testing is to design general frameworks that can be applied to \emph{several} property testing problems~\cite{Valiant:11,VV:11:stoc, DKN:15, DKN:15:FOCS}. This direction, the testing analog of a similar movement in distribution learning~\cite{CDSS:13,CDSS:14:NIPS,CDSS:14,ADLS:15}, aims at abstracting the minimal assumptions that are shared by a large variety of problems, and giving algorithms that can be used for any of these problems. In this work, we make significant progress in this direction by providing a unified framework for the question of \ignore{\emph{testing membership in a class}} testing various properties of probability distributions. More specifically, we describe a generic technique to obtain upper bounds on the sample complexity of this question, which applies to a broad range of structured classes. Our technique yields sample near-optimal and computationally efficient testers for a wide range of distribution families. Conversely, we also develop a general approach to prove lower bounds on these sample complexities, and use it to derive tight or nearly tight bounds for many of these classes.

\paragraph{Related work.} Batu et al.~\cite{BKR:04} initiated the study of efficient property testers for monotonicity and obtained (nearly) matching upper and lower bounds for this problem; while~\cite{AD:15} later considered testing the class of Poisson Binomial Distributions, and settled the sample complexity of this problem (up to the precise dependence on $\eps$). Indyk, Levi, and Rubinfeld~\cite{ILR:12}, focusing on distributions that are piecewise constant on $t$ intervals (``$t$-histograms'') described a $\tilde{O}(\sqrt{tn}/\eps^5)$-sample algorithm for testing membership to this class. Another body of work by~\cite{BDKR:05},~\cite{BKR:04}, and~\cite{DDSV:13} shows how assumptions on the shape of the distributions can lead to significantly more efficient algorithms. They describe such improvements in the case of identity and closeness testing as well as for entropy estimation, under monotonicity or $k$-modality constraints. Specifically, Batu et al. show in~\cite{BKR:04} how to obtain a $O\big({\log^3 n/\eps^3}\big)$-sample tester for closeness in this setting, in stark contrast to the $\Omega\big({{n}^{2/3}}\big)$ general lower bound. Daskalakis et al.~\cite{DDSV:13} later gave ${O}(\sqrt{\log n})$ and ${O}({\log^{2/3} n})$-sample testing algorithms for testing respectively identity and closeness of monotone distributions, and obtained similar results for $k$-modal distributions. Finally, we briefly mention two related results, due respectively to~\cite{BDKR:05} and~\cite{DDS:12}. The first one states that for the task of getting a multiplicative \emph{estimate} of the entropy of a distribution, assuming monotonicity enables exponential savings in sample complexity -- $O\big({\log^6 n}\big)$, instead of $\bigOmega{n^c}$ for the general case. The second describes how to test if an unknown $k$-modal distribution is in fact monotone, using only $\bigO{k/\eps^2}$ samples. Note that the latter line of work differs from ours in that it \emph{presupposes} the distributions satisfy some structural property, and uses this knowledge to test something else about the distribution; while we are given \textit{a priori} arbitrary distributions, and must \emph{check} whether the structural property holds. Except for the properties of monotonicity and being a PBD, nothing was previously known on testing the shape restricted properties that we study. Independently and concurrently to this work, Acharya, Daskalakis, and Kamath obtained a sample near-optimal efficient algorithm for testing log-concavity.\footnotemark

\footnotetext{Following the communication of a preliminary version of this paper (February 2015), we were informed that~\cite{ADK:15} subsequently obtained near-optimal testers for some of the classes we consider. To the best of our knowledge, their work builds on ideas from~\cite{AD:15} and their techniques are orthogonal to ours.}

Moreover, for the specific problems of identity and closeness testing,\footnote{Recall that the identity testing problem asks, given the explicit description of a distribution $\D^\ast$ and sample access to an unknown distribution $\D$, to decide whether $\D$ is equal to $\D^\ast$ or far from it; while in closeness testing both distributions to compare are unknown.} recent results of~\cite{DKN:15,DKN:15:FOCS} describe a general algorithm which applies to a large range of shape or structural constraints, and yields optimal identity testers for classes of distributions that satisfy them. We observe that while the question they answer can be cast as a specialized instance of membership testing, our results are incomparable to theirs, both because of the distinction above (testing \emph{with} versus testing \emph{for} structure) and as the structural assumptions they rely on are fundamentally different from ours.

\subsection{Results and Techniques}

\noindent {\bf Upper Bounds.} A natural way to tackle our membership testing problem would be to first learn the unknown distribution $\D$ \emph{as if} it satisfied the property, before checking if the hypothesis obtained is indeed both close to the original distribution and to the property. Taking advantage of the purported structure, the first step could presumably be conducted with a small number of samples; things break down, however, in the second step. 
Indeed, most approximation results leading to the improved learning algorithms one would apply in the first stage only provide very weak guarantees, in the $\lp[1]$ sense. For this reason, they lack the robustness that would be required for the second part, where it becomes necessary to perform \emph{tolerant} testing between the hypothesis and $\D$ -- a task that would then entail a number of samples almost linear in the domain size. To overcome this difficulty, we need to move away from these global $\lp[1]$ closeness results and instead work with stronger requirements, this time in $\lp[2]$ norm. 

At the core of our approach is an idea of Batu et al.~\cite{BKR:04}, which show that monotone distributions can be well-approximated (in a certain technical sense) 
by piecewise constant densities on a suitable interval partition of the domain; and leverage this fact to reduce monotonicity testing to uniformity testing on each interval of this partition. 
While the argument of~\cite{BKR:04} is tailored specifically for the setting of monotonicity testing, we are able to abstract the key ingredients, and obtain a generic membership tester that applies to a wide range of distribution families. In more detail, we provide a testing algorithm which applies to any class of distributions which admits succinct approximate decompositions -- that is, each distribution in the class can be well-approximated (in a strong $\lp[2]$ sense) by piecewise constant densities on a small number of intervals (we hereafter refer to this approximation property, formally defined in~\autoref{def:struct:dec:split}, as \inlineref{label:struct:criterion}; and extend the notation to apply to any \emph{class} \class of distributions for which all $\D\in\class$ satisfy~\eqref{label:struct:criterion}). 
Crucially, the algorithm does not care about \emph{how} these decompositions can be obtained: 
for the purpose of testing these structural properties we only need to establish their \emph{existence}. Specific examples are given in the corollaries below.
Informally, our main algorithmic result, informally stated (see~\autoref{theo:main:testing:detailed} for a detailed formal statement), is as follows:

\begin{restatable}[Main Theorem]{theorem}{mainthmtestingalgo}\label{theo:main:testing}
There exists an algorithm \textsc{TestSplittable} which, given sampling access to an unknown distribution $\D$ over $[n]$ and parameter $\eps\in(0,1]$, can distinguish with probability $2/3$ between \textsf{(a)} $\D\in\property$ versus \textsf{(b)} $\lp[1](\D, \property) > \eps$, for \emph{any} property $\property$ that satisfies the above natural structural criterion \eqref{label:struct:criterion}. Moreover, for many such properties this algorithm is computationally efficient, and its sample complexity is optimal (up to logarithmic factors and the exact dependence on $\eps$).
\end{restatable} 

\noindent We then instantiate this result to obtain ``out-of-the-box'' \emph{computationally efficient} testers for several classes of distributions, by showing that they satisfy the premise of our theorem (the definition of these classes is given in~\autoref{ssec:class:definitions}):

\begin{corollary}\label{coro:main:testing}
The algorithm \textsc{TestSplittable} can test the classes of monotone, unimodal, log-concave, concave, convex, and monotone hazard rate (MHR) distributions, with $\tildeO{\sqrt{n}/\eps^{7/2}}$ samples.
\end{corollary}
\begin{corollary}\label{coro:main:testing:tmod}
The algorithm \textsc{TestSplittable} can test the class of $t$-modal distributions, with $\tildeO{\sqrt{tn}/\eps^{7/2}}$ samples. \end{corollary}

\begin{corollary}\label{coro:main:testing:piecewise}
The algorithm \textsc{TestSplittable} can test the classes of $t$-histograms and $t$-piecewise degree-$d$ distributions, with $\tildeO{\sqrt{tn}/\eps^3}$ and $\tildeO{\sqrt{t(d+1)n}/\eps^{7/2} + t(d+1)/\eps^3}$ samples respectively. \end{corollary}

\begin{corollary}\label{coro:main:testing:pbd}
The algorithm \textsc{TestSplittable} can test the classes of Binomial and Poisson Binomial Distributions, with $\tildeO{{n}^{1/4}/\eps^{7/2}}$ samples.
\end{corollary}

\makeatletter{}\newcommand{\pb}[2]{\parbox[c][][c]{#1}{\strut#2\strut}}
  \begin{table}[ht]\centering\small
    \begin{adjustwidth}{-.75in}{-.5in}\centering
  \begin{tabular}{@{}|l|c|c|@{}}\hline
    { \bf Class }& {\bf Upperbound} & \bf Lowerbound\\\hline
     {Monotone}  & {$\tildeO{\frac{\sqrt{n}}{\eps^6}}$ \cite{BKR:04}, $\tildeO{\frac{\sqrt{n}}{\eps^{7/2}}}$ (\autoref{coro:main:testing})} 
                 & {$\bigOmega{\frac{\sqrt{n}}{\eps^2}}$ \cite{BKR:04}, $\bigOmega{\frac{\sqrt{n}}{\eps^2}}$ (\autoref{coro:lb:sqrtn})} \\\hline
     {Unimodal}  & {$\tildeO{\frac{\sqrt{n}}{\eps^{7/2}}}$ (\autoref{coro:main:testing})}
                 & {$\bigOmega{\frac{\sqrt{n}}{\eps^2}}$ (\autoref{coro:lb:sqrtn})} \\\hline
     {$t$-modal}  & {$\tildeO{\frac{\sqrt{{t}n}}{\eps^{7/2}}}$ (\autoref{coro:main:testing:tmod})}
                 & {$\bigOmega{\frac{\sqrt{n}}{\eps^2}}$ (\autoref{coro:lb:sqrtn})} \\\hline
     \pb{40mm}{Log-concave, concave, convex}  & {$\tildeO{\frac{\sqrt{n}}{\eps^{7/2}}}$ (\autoref{coro:main:testing})}
                 & {$\bigOmega{\frac{\sqrt{n}}{\eps^2}}$ (\autoref{coro:lb:sqrtn})} \\\hline
     \pb{40mm}{Monotone Hazard Rate (MHR)}  & {$\tildeO{\frac{\sqrt{n}}{\eps^{7/2}}}$ (\autoref{coro:main:testing})}
                 & {$\bigOmega{\frac{\sqrt{n}}{\eps^2}}$ (\autoref{coro:lb:sqrtn})} \\\hline
     \pb{40mm}{Binomial, Poisson Binomial (PBD)}  & \pb{60mm}{\centering $\tildeO{\frac{{n}^{1/4}}{\eps^2} + \frac{1}{\eps^6}}$ \cite{AD:15},\\ $\tildeO{\frac{{n}^{1/4}}{\eps^{7/2}}}$ (\autoref{coro:main:testing:pbd})}
                 & {$\bigOmega{\frac{{n}^{1/4}}{\eps^2}}$~(\cite{AD:15},~\autoref{coro:lb:pbd})} \\\hline
     \pb{40mm}{$t$-histograms}  & { $\tildeO{\frac{\sqrt{tn}}{\eps^5}}$~\cite{ILR:12}, $\tildeO{\frac{\sqrt{tn}}{\eps^3}}$ (\autoref{coro:main:testing:piecewise}) }
                 & {$\bigOmega{\sqrt{tn}}$ for $t\leq \frac{1}{\eps}$~\cite{ILR:12}, $\bigOmega{\frac{\sqrt{n}}{\eps^2}}$ (\autoref{coro:lb:sqrtn})}  \\\hline
     \pb{40mm}{$t$-piecewise degree-$d$}  & { $\tildeO{\frac{\sqrt{t(d+1)n}}{\eps^{7/2}} + \frac{t(d+1)}{\eps^3} }$ (\autoref{coro:main:testing:piecewise}) }
                 & {$\bigOmega{\frac{\sqrt{n}}{\eps^2}}$ (\autoref{coro:lb:sqrtn})}  \\\hline
     \pb{40mm}{$k$-SIIRV}  & {}
                 & {$\bigOmega{ {k}^{1/2}{n}^{1/4} }$ (\autoref{coro:lb:ksiirv})} \\\hline
  \end{tabular}
  \end{adjustwidth}
\caption{\label{fig:table:results} Summary of results.}
  \end{table}

We remark that the aforementioned sample upper bounds are information-theoretically near-optimal in the domain size $n$ (up to logarithmic factors). See~\autoref{fig:table:results} and the following subsection for the corresponding lower bounds. We did not attempt to optimize the dependence on the parameter $\eps$,
though a more careful analysis can lead to such improvements.

We stress that prior to our work, no non-trivial testing bound was known for most of these classes~--~specifically, our nearly-tight bounds for
$t$-modal with $t>1$, log-concave, concave, convex, MHR, and piecewise polynomial distributions are new. Moreover, although a few of our applications 
were known in the literature (the $\tildeO{\sqrt{n}/\eps^6}$ upper and $\bigOmega{\sqrt{n}/\eps^2}$ lower bounds on testing monotonicity can be found in~\cite{BKR:04}, while the $\Theta\big({n^{1/4}}\big)$ sample complexity of testing PBDs was recently given\footnotemark{} in~\cite{AD:15}, and the task of testing $t$-histograms is considered in~\cite{ILR:12}), the crux here is that we are able to derive them in a \emph{unified} way, by applying the same generic algorithm to all these different distribution families. 
We note that our upper bound for $t$-histograms (\autoref{coro:main:testing:piecewise}) also improves on the previous $\tildeO{\sqrt{tn}/\eps^5}$-sample tester, as long as $t=\tildeO{n^{1/3}/\eps^2}$. In addition to its generality, our framework yields much cleaner and conceptually simpler proofs of the upper and lower bounds from~\cite{AD:15}.
\footnotetext{For the sample complexity of testing monotonicity, \cite{BKR:04} originally states an $\tildeO{{\sqrt{n}}/{\eps^4}}$ upper bound, but the proof seems to only result in an $\tildeO{{\sqrt{n}}/{\eps^6}}$ bound. Regarding the class of PBDs,~\cite{AD:15} obtain an ${n^{1/4}}\cdot\tilde{O}\big({1/\eps^2}\big) + \tilde{O}\big({1/\eps^6}\big)$ sample complexity, to be compared with our $\tilde{O}\big({n^{1/4}/\eps^{7/2}}) + \bigO{\log^4 n/\eps^4 }$ upper bound; as well as an $\Omega\big({n^{1/4}/\eps^2}\big)$ lower bound.}

\paragraph*{Lower Bounds.} To complement our upper bounds, we give a generic framework for proving lower bounds against testing classes of distributions. In more detail, we describe how to {\em reduce} -- under a mild assumption on the property \class{} -- the problem of testing \emph{membership to $\class$} (``does $\D\in\class$?'') to testing \emph{identity to $\D^\ast$} (``does $\D=\D^\ast$?''), for any explicit distribution $\D^\ast$ in $\class$. While these two problems need not in general be related,\footnote{As a simple example, consider the class $\class$ of \emph{all} distributions, for which testing membership is trivial.} we show that our reduction-based approach applies to a large number of natural properties, and obtain lower bounds that nearly match our upper bounds for all of them. 
Moreover, this lets us derive a 
simple proof of the lower bound of~\cite{AD:15} on testing the class of PBDs.
The reader is referred to \autoref{theo:main:testing:lb} for the formal statement of our reduction-based lower bound theorem.
In this section, we state the concrete corollaries we obtain for specific structured distribution families:

\begin{restatable}{corollary}{corolbsqrtn}\label{coro:lb:sqrtn}
  Testing log-concavity, convexity, concavity, MHR, unimodality, $t$-modality, $t$-histograms, and $t$-piecewise degree-$d$ distributions each require $\bigOmega{{\sqrt{n}}/{\eps^2}}$ samples (the last three for $t = o(\sqrt{n})$ and $t(d+1) = o(\sqrt{n})$, respectively), for any $\eps\geq 1/n^{O(1)}$.
\end{restatable}

\begin{restatable}{corollary}{corolbpbd}\label{coro:lb:pbd}
  Testing the classes of Binomial and Poisson Binomial Distributions each require $\bigOmega{{n^{1/4}}/{\eps^2}}$ samples, for any $\eps\geq 1/n^{O(1)}$.
\end{restatable}

\begin{restatable}{corollary}{corolbksiirv}\label{coro:lb:ksiirv}
  There exist absolute constants $c>0$ and $\eps_0 > 0$ such that testing the class of $k$-SIIRV distributions requires $\Omega\big( k^{1/2}n^{1/4} \big)$ samples, for any $k=\littleO{n^c}$ and $\eps \leq \eps_0$.\nolinebreak
\end{restatable}

\paragraph*{Tolerant Testing.} 
Using our techniques, we also establish nearly--tight upper and lower bounds on tolerant testing\footnotetext{\emph{Tolerant testing} of a property \property is defined as follows: given $0 \leq \eps_1 < \eps_2 \leq 1$, one must distinguish between \textsf{(a)} $\lp[1](\D,\property) \leq \eps_1$ and \textsf{(b)} $\lp[1](\D,\property) \geq \eps_2$. This turns out to be, in general, a much harder task than that of ``regular'' testing (where we take $\eps_1=0$).} for shape restrictions. 
Similarly, our upper and lower bounds are matching as a function of the domain size.
More specifically, we give a simple generic upper bound approach (namely, a learning followed by tolerant testing algorithm).
Our tolerant testing lower bounds follow the same reduction-based approach as in the non-tolerant case. 
In more detail, our results are as follows (see~\autoref{sec:lowerbounds} and~\autoref{sec:toltesting:ub}):

\begin{restatable}{corollary}{coromaintoltestingmlogm}\label{coro:main:tol:testing:mlogm}
Tolerant testing of log-concavity, convexity, concavity, MHR, unimodality, and $t$-modality can be performed with $O\big( \frac{1}{(\eps_2-\eps_1)^2}\frac{n}{\log n} \big)$ samples, for $\eps_2 \geq C \eps_1$ (where $C>2$ is an absolute constant).
\end{restatable}

\begin{restatable}{corollary}{coromaintoltestingpbd}\label{coro:main:tol:testing:pbd}
Tolerant testing of the classes of Binomial and Poisson Binomial Distributions can be performed with $O\big( \frac{1}{(\eps_2-\eps_1)^2}\frac{\sqrt{n\log({1}/{\eps_1})}}{\log n} \big)$ samples, for $\eps_2 \geq C \eps_1$ (where $C>2$ is an absolute constant).
\end{restatable}

\begin{restatable}{corollary}{corolbtolnlogn}\label{coro:tol:lb:nlogn}
  Tolerant testing of log-concavity, convexity, concavity, MHR, unimodality, and $t$-modality each require $\bigOmega{\frac{1}{(\eps_2-\eps_1)}\frac{n}{\log n}}$ samples (the latter for $t = o(n)$).
\end{restatable}

\begin{restatable}{corollary}{corolbtolpbd}\label{coro:lb:tol:pbd}
  Tolerant testing of the classes of Binomial and Poisson Binomial Distributions each require $\bigOmega{\frac{1}{(\eps_2-\eps_1)}\frac{\sqrt{n}}{\log n}}$ samples.
\end{restatable}

\paragraph*{On the scope of our results.} We point out that our main theorem is likely to apply to many other classes of structured distributions, 
due to the mild structural assumptions it requires. However, we did not attempt here to be comprehensive; but rather to illustrate the generality of our approach. Moreover, for all properties considered in this paper the generic upper and lower bounds we derive through our methods turn out to be optimal up to at most polylogarithmic factors (with regard to the support size). The reader is referred to~\autoref{fig:table:results} for a summary of our results and related work.

\subsection{Organization of the Paper}

We start by giving the necessary background and definitions in~\autoref{sec:preliminaries}, before turning to our main result, the proof of~\autoref{theo:main:testing} (our general testing algorithm) in~\autoref{sec:algorithm}. In~\autoref{sec:structural}, we establish the necessary structural theorems for each classes of distributions considered, enabling us to derive the upper bounds of~\autoref{fig:table:results}.~\autoref{sec:effectivesupport} introduces a slight modification of our algorithm which yields stronger testing results for classes of distributions with small effective support, and use it to derive~\autoref{coro:main:testing:pbd}, our upper bound for Poisson Binomial distributions. Second,~\autoref{sec:lowerbounds} contains the details of our lower bound methodology, and of its applications to the classes of~\autoref{fig:table:results}. Finally,~\autoref{sec:lowerbounds:tol} is concerned with the extension of this methodology to \emph{tolerant} testing, of which~\autoref{sec:toltesting:ub} describes a generic upper bound counterpart.
 
\section{Notation and Preliminaries}\label{sec:preliminaries}
  \makeatletter{}\subsection{Definitions}\label{ssec:class:definitions}

We give here the formal descriptions of the classes of distributions involved in this work. Recall that a distribution $\D$ over $[n]$ is \emph{monotone} (non-increasing) if its probability mass function (pmf) satisfies $\D(1) \geq \D(2) \geq \dots \D(n)$. A natural generalization of the class $\classmon$ of monotone distributions is the set of $t$-modal distributions, i.e. distributions whose pmf can go ``up and down'' or ``down and up'' up to $t$ times:\footnote{Note that this slightly deviates from the Statistics literature, where only the peaks are counted as modes (so that what is usually referred to as a bimodal distribution is, according to our definition, $3$-modal).}
\begin{definition}[$t$-modal]\label{def:tmodal}
  Fix any distribution $\D$ over $[n]$, and integer $t$. $\D$ is said to have $t$ \emph{modes} if there exists a sequence $i_0 < \dots < i_{t+1}$ such 
  that either $(-1)^j \D(i_j) < (-1)^j \D(i_{j+1})$ for all $0\leq j \leq t$, or $(-1)^j \D(i_j) > (-1)^j \D(i_{j+1})$ for all $0\leq j \leq t$. We call $\D$ \emph{$t$-modal} if it has at most $t$ modes, and write $\classtmo$ for the class of all $t$-modal distributions (omitting the dependence on $n$). The particular case of $t=1$ corresponds to the set $\classuni$ of \emph{unimodal} distributions.
\end{definition}

\begin{definition}[Log-Concave]\label{def:logconcave}
  A distribution $\D$ over $[n]$ is said to be \emph{log-concave} if it satisfies the following conditions: \textsf{(i)} for any $1 \leq i < j < k \leq n$ such that $\D(i)\D(k) > 0$, $\D(j) > 0$; and \textsf{(ii)} for all $1 < k < n$, $\D(k)^2 \geq \D(k-1)\D(k+1)$. We write $\classlog$ for the class of all log-concave distributions (omitting the dependence on $n$).
\end{definition}

\begin{definition}[Concave and Convex]\label{def:concave}
  A distribution $\D$ over $[n]$ is said to be \emph{concave} if it satisfies the following conditions: \textsf{(i)} for any $1 \leq i < j < k \leq n$ such that $\D(i)\D(k) > 0$, $\D(j) > 0$; and \textsf{(ii)} for all $1 < k < n$ such that $\D(k - 1)\D(k + 1)>0$, $2\D(k) \geq \D(k - 1)+\D(k + 1)$; it is \emph{convex} if the reverse inequality holds in \textsf{(ii)}. We write $\classcve$ (resp. $\classcvx$) for the class of all concave (resp. convex) distributions (omitting the dependence on $n$).
\end{definition}
It is not hard to see that convex and concave distributions are unimodal; moreover, every concave distribution is also log-concave, i.e. $\classcve\subseteq\classlog$. Note that in both \autoref{def:logconcave} and \autoref{def:concave}, condition \textsf{(i)} is equivalent to enforcing that the distribution be supported on an interval.

\begin{definition}[Monotone Hazard Rate]\label{def:mhr}
  A distribution $\D$ over $[n]$ is said to have \emph{monotone hazard rate} (MHR) if its \emph{hazard rate} $H(i)\eqdef \frac{\D(i)}{\sum_{j=i}^{n} \D(j)}$ is a non-decreasing function. We write $\classmhr$ for the class of all MHR distributions (omitting the dependence on $n$).
\end{definition}
It is known that every log-concave distribution is both unimodal and MHR (see e.g. \cite[Proposition 10]{An:96}), and that monotone distributions are MHR. Two other classes of distributions have elicited significant interest in the context of density estimation, that of \emph{histograms} (piecewise constant) and \emph{piecewise polynomial densities}:
\begin{definition}[Piecewise Polynomials~\cite{CDSS:14}]\label{def:piecewise}
  A distribution $\D$ over $[n]$ is said to be a \emph{$t$-piecewise degree-$d$ distribution} if there is a partition of $[n]$ into $t$ disjoint intervals $I_1,\dots,I_t$ such that $\D(i) = p_j(i)$ for all $i \in I_j$, where each $p_1,\dots p_t$ is a univariate polynomial of degree at most $d$. We write $\classpoly$ for the class of all $t$-piecewise degree-$d$ distributions (omitting the dependence on $n$). (We note that {$t$-piecewise degree-$0$ distributions} are also commonly referred to as \emph{$t$-histograms}, and write $\classhist$ for $\classpoly[t,0]$.)
\end{definition}

Finally, we recall the definition of the two following classes, which both extend the family of Binomial distributions $\classbin_n$: the first, by removing the need for each of the independent Bernoulli summands to share the same bias parameter.
\begin{definition}\label{def:pbd}
A random variable $X$ is said to follow a \emph{Poisson Binomial Distribution} (with parameter $n\in\N$) if it can be written as $X=\sum_{k=1}^n X_k$, where $X_1\dots,X_n$ are independent, non-necessarily identically distributed Bernoulli random variables. We denote by $\classpbd_n$ the class of all such Poisson Binomial Distributions.
\end{definition}
\noindent It is not hard to show that Poisson Binomial Distributions are in particular log-concave. One can generalize even further, by allowing each random variable of the summation to be integer-valued:
\begin{definition}\label{def:siirv}
Fix any $k\geq 0$. We say a random variable $X$ is a \emph{$k$-Sum of Independent Integer Random Variables ($k$-SIIRV)} with parameter $n\in\N$ if it can be written as $X=\sum_{j=1}^n X_j$, where $X_1\dots,X_n$ are independent, non-necessarily identically distributed random variables taking value in $\{0,1,\dots,k-1\}$. We denote by $\classksiirv_n$ the class of all such $k$-SIIRVs.
\end{definition}

\subsection{Tools from previous work}

We first restate a result of Batu et al. relating closeness to uniformity in $\lp[2]$ and $\lp[1]$ norms to ``overall flatness'' of the probability mass function, and which will be one of the ingredients of the proof of \autoref{theo:main:testing}:

\begin{lemma}[{\cite{BFRSW:00,BFFKRW:01}}]\label{lemma:small:l2:close:uniform:l1}
Let $\D$ be a distribution on a domain $S$. \textsf{(a)} If $\max_{i\in S} \D(i) \leq (1+\eps)\min_{i\in S} \D(i)$, then $\normtwo{\D}^2 \leq (1+\eps^2)/\abs{S}$. \textsf{(b)} If $\normtwo{\D}^2 \leq (1+\eps^2)/\abs{S}$, then $\normone{\D-\uniform_{S}} \leq \eps$.
\end{lemma}

\noindent To check condition \textsf{(b)} above we shall rely on the following, which one can derive from the techniques in \cite{DKN:15} and whose proof we defer to \autoref{app:l2:proof}:
\
\begin{restatable}[Adapted from {\cite[Theorem 11]{DKN:15}}]{lemma}{lemmaestimateltwoadd}\label{lemma:estimate:l2:add}
There exists an algorithm \textsc{Check-Small-$\lp[2]$} which, given parameters $\eps,\delta\in(0,1)$ and $c\cdot{\sqrt{\abs{I}}}/{\eps^2} \log(1/\delta)$ independent samples from a distribution $\D$ over $I$ (for some absolute constant $c>0$), outputs either \yes or \no, and satisfies the following.
  \begin{itemize}
    \item If $\normtwo{\D-\uniform_I} > {\eps}/{\sqrt{\abs{I}}}$, then the algorithm outputs \no with probability at least $1-\delta$;
    \item If $\normtwo{\D-\uniform_I} \leq {\eps}/{2\sqrt{\abs{I}}}$, then the algorithm outputs \yes with probability at least $1-\delta$.
  \end{itemize}
\end{restatable}

Finally, we will also rely on a classical result from Probability, the \emph{Dvoretzky--Kiefer--Wolfowitz} (DKW) inequality, restated below:
\begin{theorem}[\cite{DKW:56,Massart:90}]\label{theo:dkw:ineq}
Let $\D$ be a distribution over $[n]$. Given $m$ independent samples $x_1,\dots ,x_m$ from $\D$, define the empirical distribution $\hat{\D}$ as follows:
\[
\hat{\D}(i) \eqdef \frac{\abs{ \setOfSuchThat{j\in[m]}{x_j=i} } }{m}, \quad i\in[n].
\]
Then, for all $\eps > 0$, $\probaOf{ \kolmogorov{\D}{\hat{\D}} > \eps } \leq 2e^{-2m\eps^2}$, where $\kolmogorov{\cdot}{\cdot}$ denotes the Kolmogorov distance (i.e., the $\lp[\infty]$ distance between cumulative distribution functions).
\end{theorem} 
\noindent In particular, this implies that $\bigO{1/\eps^2}$ samples suffice to learn a distribution up to $\eps$ in Kolmogorov distance.
 
\section{The General Algorithm}\label{sec:algorithm}
  \makeatletter{}In this section, we obtain our main result, restated below:
\mainthmtestingalgo*

\paragraph*{Intuition.} Before diving into the proof of this theorem, we first provide a high-level description of the argument. The algorithm proceeds in 3 stages: the first, the \emph{decomposition step}, attempts to recursively construct a partition of the domain in a small number of intervals, with a very strong guarantee. If the decomposition succeeds, then the unknown distribution $\D$ will be close (in $\lp[1]$ distance) to its ``flattening'' on the partition; while if it fails (too many intervals have to be created), this serves as evidence that $\D$ does not belong to the class and we can reject. The second stage, the \emph{approximation step}, then learns this flattening of the distribution -- which can be done with few samples since by construction we do not have many intervals. The last stage is purely computational, the \emph{projection step}: where we verify that the flattening we have learned is indeed close to the class \class. If all three stages succeed, then by the triangle inequality it must be the case that $\D$ is close to \class; and by the structural assumption on the class, if $\D\in\class$ then it will admit succinct enough partitions, and all three stages will go through.\medskip

\noindent Turning to the proof, we start by defining formally the ``structural criterion'' we shall rely on, before describing the algorithm at the heart of our result in~\autoref{ssec:main:algorithm}. (We note that a modification of this algorithm will be described in~\autoref{sec:effectivesupport}, and will allow us to derive~\autoref{coro:main:testing:pbd}.)

\begin{definition}[Decompositions]\label{def:struct:dec:split}
Let $\gamma > 0$ and $L=L(\gamma,n)\geq 1$.  A class of distributions \class on $[n]$ is said to be \emph{$(\gamma,L)$-decomposable} if for every $\D\in\class$ there exists $\ell \leq L$ and a partition $\mathcal{I}(\gamma,\D)=(I_1,\dots,I_\ell)$ of \new{the interval $[1,n]$} such that, for all $j\in[\ell]$, one of the following holds:
\begin{enumerate}[(i)]
  \item\label{def:struct:item:light}  $\D(I_j) \leq \frac{\gamma}{L}$; or 
  \item\label{def:struct:item:flat} $\displaystyle \max_{i \in I_j} \D(i)\leq  (1+\gamma)\cdot \min_{i \in I_j} \D(i)$.
\end{enumerate}
Further, if $\mathcal{I}(\gamma,\D)$ is \emph{dyadic} (i.e., each $I_k$ is of the form $[j\cdot 2^i+1,(j+1)\cdot 2^i]$ for some integers $i,j$, corresponding to the leaves of a recursive bisection of $[n]$), then \class is said to be \emph{$(\gamma,L)$-splittable}.
\end{definition}

\begin{lemma}\label{lemma:decomposable:splittable}
If \class is $(\gamma,L)$-decomposable, then it is $(\gamma, \bigO{L\log n})$-splittable.
\end{lemma}
\begin{proof}
We will begin by proving a claim that for every partition $\mathcal{I}=\{I_1,I_2,...I_L\}$ of the interval $[1,n]$ into $L$ intervals, there exists a refinement of that partition which consists of at most $L\cdot \log n$ dyadic intervals. So, it suffices to prove that every interval $[a,b]\subseteq [1,n]$, can be partitioned in at most $\bigO{\log n}$ dyadic intervals. Indeed, let $\ell$ be the largest integer such that $2^\ell\leq \frac{b-a}{2}$ and let $m$ be the smallest integer such that $m\cdot 2^\ell\geq a$. If follows that $m\cdot 2^\ell\leq a+\frac{b-a}{2}=\frac{a+b}{2}$ and $(m+1)\cdot 2^\ell\leq b$. So, the interval $I=[m\cdot 2^\ell+1,(m+1)\cdot 2^\ell]$ is fully contained in $[a,b]$ and has size at least $\frac{b-a}{4}$. 

We will also use the fact that, for every $\ell^\prime\leq \ell$,
\begin{equation}\label{eq:lem:dec:split:smaller:powers}
m\cdot 2^\ell=m\cdot 2^{\ell-\ell^\prime}\cdot 2^{\ell^\prime}=m^\prime \cdot 2^{\ell^\prime}
\end{equation} 

Now consider the following procedure: Starting from right (resp. left) side of the interval $I$, we add the largest interval which is adjacent to it and fully contained in $[a,b]$ and recurse until we cover the whole interval $[(m+1)\cdot 2^\ell+1,b]$ (resp. $[a,m\cdot 2^\ell]$). Clearly, at the end of this procedure, the whole interval $[a,b]$ is covered by dyadic intervals. It remains to show that the procedure takes $\bigO{\log n}$ steps. Indeed, using~\autoref{eq:lem:dec:split:smaller:powers}, we can see that at least half of the remaining left or right interval is covered in each step (except maybe for the first 2 steps where it is at least a quarter). Thus, the procedure will take at most $2\log n +2=\bigO{\log n}$ steps in total. From the above, we can see that each of the $L$ intervals of the partition $\mathcal{I}$ can be covered with $\bigO{\log n}$ dyadic intervals, which completes the proof of the claim. 

In order to complete the proof of the lemma, notice that the two conditions in~\autoref{def:struct:dec:split} are closed under taking subsets. 
\end{proof}

\subsection{The algorithm}\label{ssec:main:algorithm}

\autoref{theo:main:testing}, and with it~\autoref{coro:main:testing} and~\autoref{coro:main:testing:tmod} will follow from the theorem below, combined with the structural theorems from~\autoref{sec:structural}:
\begin{theorem}\label{theo:main:testing:detailed}
  Let $\class$ be a class of distributions over $[n]$ for which the following holds. 
    \begin{enumerate}
      \item $\class$ is $(\gamma,L(\gamma,n))$-splittable;
      \item there exists a procedure $\estimdist$ which, given as input a parameter $\alpha\in(0,1)$ and the explicit description of a distribution $\D$ over $[n]$, 
      returns \yes if the distance $\lp[1](\D,\class)$ to $\class$ is at most $\alpha/10$, and \no if $\lp[1](\D,\class) \geq 9\alpha/10$ (and either \yes or \no otherwise).
    \end{enumerate}
  Then, the algorithm \textsc{TestSplittable} (\autoref{algo:test:splittable}) is a $\bigO{ \max\mleft({\sqrt{nL}}\log n/{\eps^3}, L/{\eps^2}\mright) }$-sample tester for $\class$, for $L=L(\eps,n)$. (Moreover, if $\estimdist$ is computationally efficient, then so is \textsc{TestSplittable}.)
\end{theorem}

\begin{algorithm}[ht]
  \algblock[block]{Start}{Start}
  \algblockdefx[]{Start}{End}    [1]{\textsc{#1}}    [1][]{\textsc{#1}}
  \begin{algorithmic}[1]
    \Require Domain $I$ (interval), sample access to $\D$ over $I$; subroutine $\estimdist$
    \renewcommand{\algorithmicrequire}{\textbf{Input:}}
    \Require Parameters $\eps$ and function $L_\class(\cdot,\cdot)$.
    \Start{Setting Up}
      \State Define $\gamma\eqdef \frac{\eps}{80}$, $L\eqdef L_\class(\gamma,\abs{I})$, $\kappa\eqdef\frac{\eps}{160L}$, $\delta\eqdef\frac{1}{10L}$; and $c>0$ be as in~\autoref{lemma:estimate:l2:add}. 
      \State Set $m\eqdef C\cdot\max\left( \frac{1}{\kappa}, \frac{\sqrt{L\abs{I}}}{\eps^3} \right)\cdot\log \abs{I} = \tildeO{ \frac{\sqrt{L\abs{I}}}{\eps^3} + \frac{L}{\eps} }$ \Comment{{\small $C$ is an absolute constant.}}
    \State\label{step:get:samples} Obtain a sequence $\textbf{s}$ of $m$ independent  samples from $\D$.
    \;\;\Comment{{\small  For any $J\subseteq I$, let $m_{J}$ be the number of samples falling in $J$. }}
    \End
    \Start{Decomposition}
      \While{ $m_I \geq \max\left( c\cdot\frac{\sqrt{\abs{I}}}{\eps^2}\log\frac{1}{\delta}, \kappa m \right)$ and at most $L$ splits have been performed } \label{step:recursion} 
        \State Run \textsc{Check-Small-$\lp[2]$} (from~\autoref{lemma:estimate:l2:add}) with parameters $\frac{\eps}{40}$ and $\delta$, using the samples of $\textbf{s}$ belonging to $I$.
        \If{ \textsc{Check-Small-$\lp[2]$} outputs \no }
          \State Bisect $I$, and recurse on both halves (using the same samples).
        \EndIf
      \EndWhile
      \If{more than $L$ splits have been performed} \label{step:check:succinctness}
        \State \Return \reject
      \Else
        \State Let $\mathcal{I}\eqdef (I_1,\dots,I_\ell)$ be the partition of $[n]$ from the leaves of the recursion. \Comment{{\small $\ell \leq L$.}}
      \EndIf
    \End
    \Start{Approximation}
      \State\label{step:learn:flattening} Learn the flattening $\Phi(\D,\mathcal{I})$ of $\D$ to $\lp[1]$ error $\frac{\eps}{20}$ (with probability $1/10$), using $\bigO{\ell/\eps^2}$ new samples. Let $\tilde{\D}$ be the resulting hypothesis. \Comment{{\small $\tilde{\D}$ is a $\ell$-histogram.}}
    \End
    \Start{Offline Check}
      \State \Return \accept if and only if $\estimdist(\eps, \tilde{\D})$ returns \yes. \Comment{{\small No sample needed.}}
    \End
  \end{algorithmic}
  \caption{\label{algo:test:splittable}\sc TestSplittable}
\end{algorithm}

\subsection{Proof of~\autoref{theo:main:testing:detailed}}
We now give the proof of our main result (\autoref{theo:main:testing:detailed}), first analyzing the sample complexity of~\autoref{algo:test:splittable} before arguing its correctness. For the latter, we will need the following simple lemma from~\cite{ILR:12}, restated below:
\begin{fact}[{\cite[Fact 1]{ILR:12}}]\label{fact:ilr12}
Let $\D$ be a distribution over $[n]$, and $\delta \in (0,1]$. Given $m\geq C\cdot\frac{\log\frac{n}{\delta}}{\eta}$ independent samples from $\D$ (for some absolute constant $C>0$), with probability at least $1-\delta$ we have that, for every interval $I\subseteq[n]$:
\begin{enumerate}[(i)]
  \item\label{fact:ilr12:1} if $\D(I) \geq \frac{\eta}{4}$, then $\frac{\D(I)}{2} \leq \frac{m_I}{m} \leq \frac{3\D(I)}{2}$;
  \item\label{fact:ilr12:2} if $\frac{m_I}{m} \geq \frac{\eta}{2}$, then $\D(I) > \frac{\eta}{4}$;
  \item\label{fact:ilr12:3} if $\frac{m_I}{m} < \frac{\eta}{2}$, then $\D(I) < \eta$;
\end{enumerate}
where $m_I \eqdef \abs{ \setOfSuchThat{j\in[m]}{x_j \in I } }$ is the number of the samples falling into $I$.
\end{fact}

\subsection{Sample complexity.}
The sample complexity is immediate, and comes from Steps~\ref{step:get:samples} and \ref{step:learn:flattening}. The total number of samples is
\[
  m+\bigO{\frac{\ell}{\eps^2}} = \bigO{ \frac{\sqrt{\abs{I}\cdot L}}{\eps^3} \log \abs{I} + \frac{L}{\eps}\log \abs{I} + \frac{L}{\eps^2} }
   = \bigO{ \frac{\sqrt{\abs{I}\cdot L}}{\eps^3}\log \abs{I} + \frac{L}{\eps^2} }\;.
\]
\subsection{Correctness.}
Say an interval $I$ considered during the execution of the ``Decomposition'' step is \emph{heavy} if $m_I$ is big enough on Step~\ref{step:recursion}, and \emph{light} otherwise; and let $\mathscr{H}$ and $\mathscr{L}$ denote the sets of heavy and light intervals respectively. By choice of $m$ and a union bound over all $\abs{I}^2$ possible intervals, we can assume on one hand that with probability at least $9/10$ the guarantees of~\autoref{fact:ilr12} hold simultaneously for all intervals considered. We hereafter condition on this event.\medskip 

We first argue that if the algorithm does not reject in Step~\ref{step:check:succinctness}, then with probability at least $9/10$ we have $\normone{\D-\Phi(\D,\mathcal{I})} \leq \eps/20$. Indeed, we can write
\begin{align*}
  \normone{\D-\Phi(\D,\mathcal{I})} &= 
            \sum_{k\colon I_k \in \mathscr{L} } \D(I_k)\cdot \normone{\D_{I_k} - \uniform_{I_k}} 
          + \sum_{k\colon I_k \in \mathscr{H}} \D(I_k)\cdot \normone{\D_{I_k} - \uniform_{I_k}} \\
          &\leq 2\sum_{k\colon I_k \in \mathscr{L} } \D(I_k) + \sum_{k\colon I_k \in \mathscr{H} } \D(I_k)\cdot \normone{\D_{I_k} - \uniform_{I_k}}\;.
\end{align*}
Let us bound the two terms separately.
    \begin{itemize}
      \item If $I^\prime \in \mathscr{H}$, then by our choice of threshold we can apply~\autoref{lemma:estimate:l2:add} with $\delta=\frac{1}{10L}$; conditioning on all of the (at most $L$) events happening, which overall fails with probability at most $1/10$ by a union bound, we get
        \[
            \normtwo{\D_{I^\prime}}^2 = \normtwo{\D_{I^\prime}-\uniform_{I^\prime}}^2 + \frac{1}{\abs{I^\prime}} 
            \leq \left( 1+\frac{\eps^2}{1600} \right) \frac{1}{\abs{I^\prime}}
        \]
        as \textsc{Check-Small-$\lp[2]$} returned \yes; and by~\autoref{lemma:small:l2:close:uniform:l1} this implies $\normone{\D_{I^\prime} - \uniform_{I^\prime}}\leq \eps/40$.
      \item 
      If $I^\prime \in \mathscr{L}$, then we claim that 
      $\D(I^\prime) \leq \max( \kappa, 2c\cdot\frac{\sqrt{\abs{I^\prime}}}{m\eps^2}\log\frac{1}{\delta} )$. Clearly, this is true if $\D(I^\prime) \leq \kappa$, so it only remains to show that $\D(I^\prime) \leq 2c\cdot\frac{\sqrt{\abs{I^\prime}}}{m\eps^2}\log\frac{1}{\delta}$. But this follows from~\autoref{fact:ilr12} \ref{fact:ilr12:1}, as if we had 
$\D(I^\prime) > 2c\cdot\frac{\sqrt{\abs{I^\prime}}}{m\eps^2}\log\frac{1}{\delta}$ then $m_{I^\prime}$ would have been big enough, and $I^\prime\notin \mathscr{L}$. Overall,
\[
\sum_{I^\prime \in \mathscr{L} } \D(I^\prime) 
\leq \sum_{I^\prime \in \mathscr{L} } \left( \kappa + 2c\cdot\frac{\sqrt{\abs{I^\prime}}}{m\eps^2}\log\frac{1}{\delta} \right)
\leq L\kappa + 2\sum_{I^\prime \in \mathscr{L} } c\cdot\frac{\sqrt{\abs{I^\prime}}}{m\eps^2}\log\frac{1}{\delta}
\leq \frac{\eps}{160}\left(1+ \sum_{I^\prime \in \mathscr{L} } \sqrt{\frac{\abs{I^\prime}}{\abs{I}L}}\right)
\leq \frac{\eps}{80}
\]
for a sufficiently big choice of constant $C>0$ in the definition of $m$; where we first used that $\abs{\mathscr{L}} \leq L$, and then that $\sum_{I^\prime \in \mathscr{L} } \sqrt{\frac{\abs{I^\prime}}{\abs{I}}}\leq \sqrt{L}$ by Jensen's inequality.

    \end{itemize}
Putting it together, this yields
\begin{align*}
  \normone{\D-\Phi(\D,\mathcal{I})} 
  &\leq 2\cdot \frac{\eps}{80} + \frac{\eps}{40} \sum_{ I^\prime \in \mathscr{H} } \D(I_k) \leq \eps/40+\eps/40 = \eps/20.
\end{align*}

\begin{description}
  \item[Soundness.] By contrapositive, we argue that if the test returns \accept, then (with probability at least $2/3$) $\D$ is $\eps$-close to \class. Indeed, conditioning on $\tilde{\D}$ being $\eps/20$-close to $\Phi(\D,\mathcal{I})$, we get by the triangle inequality that 
  \begin{align*}
      \normone{\D-\class} &\leq \normone{\D-\Phi(\D,\mathcal{I})} + \normone{\Phi(\D,\mathcal{I}) - \tilde{\D}} + \dist{\tilde{\D}}{\class} \\
      &\leq \frac{\eps}{20} + \frac{\eps}{20}+\frac{9\eps}{10} = \eps.
  \end{align*}
  Overall, this happens except with probability at most $1/10+1/10+1/10 < 1/3$.
  
  \item[Completeness.] Assume $\D\in\class$. Then the choice of of $\gamma$ and $L$ ensures the existence of a good dyadic partition $\mathcal{I}(\gamma,\D)$ in the sense of~\autoref{def:struct:dec:split}. For any $I$ in this partition for which~\ref{def:struct:item:light} holds ($\D(I) \leq \frac{\gamma}{L} < \frac{\kappa}{2}$), $I$ will have $\frac{m_I}{m} < \kappa$ and be kept as a ``light leaf'' (this by contrapositive of~\autoref{fact:ilr12} \ref{fact:ilr12:2}). For the other ones, \ref{def:struct:item:flat} holds: let $I$ be one of these (at most $L$) intervals.
    \begin{itemize}
      \item If $m_I$ is too small on  Step~\ref{step:recursion}, then $I$ is kept as ``light leaf.''
      \item Otherwise, then by our choice of constants we can use~\autoref{lemma:small:l2:close:uniform:l1} and apply~\autoref{lemma:estimate:l2:add} with $\delta=\frac{1}{10L}$; conditioning on all of the (at most $L$) events happening, which overall fails with probability at most $1/10$ by a union bound, \textsc{Check-Small-$\lp[2]$} will output \yes, as
        \[
            \normtwo{\D_I-\uniform_I}^2 = \normtwo{\D_I}^2 - \frac{1}{\abs{I}} \leq \left( 1+\frac{\eps^2}{6400} \right) \frac{1}{\abs{I}} - \frac{1}{\abs{I}}
             = \frac{\eps^2}{6400\abs{I}}
        \]
        and $I$ is kept as ``flat leaf.''
    \end{itemize}
  Therefore, as $\mathcal{I}(\gamma,\D)$ is dyadic the \textsc{Decomposition} stage is guaranteed to stop within at most $L$ splits (in the worst case, it goes on until $\mathcal{I}(\gamma,\D)$ is considered, at which point it succeeds).\footnotemark{} Thus Step~\ref{step:check:succinctness} passes, and the algorithm reaches the \textsc{Approximation} stage. By the foregoing discussion, this implies $\Phi(\D,\mathcal{I})$ is $\eps/20$-close to $\D$ (and hence to \class); $\tilde{\D}$ is then (except with probability at most $1/10$) $(\frac{\eps}{20}+\frac{\eps}{20}=\frac{\eps}{10})$-close to \class, and the algorithm returns \accept.
\end{description}
\footnotetext{In more detail, we want to argue that if $\D$ is in the class, then a decomposition with
at most $L$ pieces is found by the algorithm. Since there \emph{is} a dyadic
decomposition with at most $L$ pieces (namely, $\mathcal{I}(\gamma,\D)=(I_1,\dots,I_t)$), it suffices to
argue that the algorithm will never split one of the $I_j$'s (as every
single $I_j$ will eventually be considered by the recursive binary splitting,
unless the algorithm stopped recursing in this ``path'' before  even considering $I_j$,
which is even better). But this is the case by the above argument, which ensures each such $I_j$ will
be recognized as satisfying one of the two conditions for ``good
decomposition'' (being either close to uniform in $\lp[2]$, or having very little mass).}
 
\section{Structural Theorems}\label{sec:structural}
  \makeatletter{}In this section, we show that a wide range of  natural distribution families 
are succinctly decomposable, and provide efficient projection algorithms
for each class. 

\subsection{Existence of Structural Decompositions} \label{ssec:struct-existence}

\begin{theorem}[Monotonicity]\label{theo:structural:monotone}
For all $\gamma > 0$, the class $\classmon$ of monotone distributions on $[n]$ is $(\gamma,L)$-splittable for $L \eqdef \bigO{\frac{\log^2 n}{\gamma}}$.
\end{theorem}
\noindent Note that this proof can already be found in \cite[Theorem 10]{BKR:04}, interwoven with the analysis of their algorithm. For the sake of being self-contained, we reproduce the structural part of their argument, removing its algorithmic aspects:
\begin{proofof}{\autoref{theo:structural:monotone}}
We define the $\mathcal{I}$ recursively as follows: $\mathcal{I}^{(0)}=([1,n])$, and for $j \geq 0$ the partition $\mathcal{I}^{(j+1)}$ is obtained from $\mathcal{I}^{(j)}=(I_1^{(j)},\dots,I_{\ell_j}^{(j)})$ by going over the $I^{(j)}_i=[a^{(j)}_i, b^{(j)}_i]$ in order, and:
\begin{enumerate}[(a)]
  \item\label{theo:struct:proof:item:light} if $\D(I^{(j)}_i)\leq \frac{\gamma}{L}$, then $I^{(j)}_i$ is added as element of $\mathcal{I}^{(j+1)}$ (``marked as leaf'');
  \item\label{theo:struct:proof:item:flat} else, if $\D(b^{(j)}_i) \leq (1+\gamma)\D(a^{(j)}_i)$, then $I^{(j)}_i$ is added as element of $\mathcal{I}^{(j+1)}$ (``marked as leaf'');
  \item otherwise, bisect $I^{(j)}$ in $I^{(j)}_{\rm L}$, $I^{(j)}_{\rm R}$ (with $\abs{I^{(j)}_{\rm L}}=\clg{\abs{I^{(j)}}/2}$) and add both $I^{(j)}_{\rm L}$ and $I^{(j)}_{\rm R}$ as elements of $\mathcal{I}^{(j+1)}$.
\end{enumerate}
and repeat until convergence (that is, whenever the last item is not applied for any of the intervals). Clearly, this process is well-defined, and will eventually terminate (as $(\ell_j)_j$ is a non-decreasing sequence of natural numbers, upper bounded by $n$). Let $\mathcal{I}=(I_1,\dots,I_{\ell})$ (with $I_i=[a_i,a_{i+1})$) be its outcome, so that the $I_i$'s are consecutive intervals all satisfying either \ref{theo:struct:proof:item:light} or \ref{theo:struct:proof:item:flat}. As \ref{theo:struct:proof:item:flat} clearly implies \ref{def:struct:item:flat}, we only need to show that $\ell \leq L$; for this purpose, we shall leverage as in \cite{BKR:04} the fact that $\D$ is monotone to bound the number of recursion steps.

The recursion above defines a complete binary tree (with the leaves being the intervals satisfying \ref{theo:struct:proof:item:light} or \ref{theo:struct:proof:item:flat}, and the internal nodes the other ones). Let $t$ be the number of recursion steps the process goes through before converging to $\mathcal{I}$ (height of the tree); as mentioned above, we have $t\leq \log n$ (as we start with an interval of size $n$, and the length is halved at each step.). Observe further that if at any point an interval $I^{(j)}_i=[a^{(j)}_i, b^{(j)}_i]$ has $\D(a^{(j)}_i) \leq \frac{\gamma}{nL}$, then it immediately (as well as all the $I^{(j)}_k$'s for $k\geq i$ by monotonicity) satisfies \ref{theo:struct:proof:item:light} and is no longer split (``becomes a leaf''). So at any $j \leq t$, the number of intervals $i_j$ for which neither \ref{theo:struct:proof:item:light} nor \ref{theo:struct:proof:item:flat} holds must satisfy
\[
  1 \geq \D(a^{(j)}_1) > (1+\gamma)\D(a^{(j)}_2) > (1+\gamma)^2\D(a^{(j)}_3) > \dots > (1+\gamma)^{i_j-1}\D(a^{(j)}_{i_j}) \geq (1+\gamma)^{i_j-1}\frac{\gamma}{nL}
\]
where $a_k$ denotes the beginning of the $k$-th interval (again we use monotonicity to argue that the extrema were reached at the ends of each interval), so that $i_j \leq 1+\frac{\log\frac{nL}{\gamma}}{\log(1+\gamma)}$. In particular, the total number of internal nodes is then 
\[
\sum_{i=1}^t i_j \leq t\cdot\left(1+\frac{\log\frac{nL}{\gamma}}{\log(1+\gamma)}\right) = (1+\littleO{1})\frac{\log^2 n}{\log(1+\gamma)} \leq L\;.\]
This implies the same bound on the number of leaves $\ell$.
\end{proofof}

\begin{corollary}[Unimodality]\label{theo:structural:unimodal}
For all $\gamma > 0$, the class $\classuni$ of unimodal distributions on $[n]$ is $(\gamma,L)$-decomposable for $L \eqdef \bigO{\frac{\log^2 n}{\gamma}}$.
\end{corollary}
\begin{proof}
For any $\D\in\classuni$, $[n]$ can be partitioned in two intervals $I$, $J$ such that $\D_{I}$, $\D_J$ are either monotone non-increasing or non-decreasing. Applying \autoref{theo:structural:monotone} to $\D_{I}$ and $\D_J$ and taking the union of both partitions yields a (no longer necessarily dyadic) partition of $[n]$.
\end{proof}
\noindent The same argument yields an analogue statement for $t$-modal distributions:
\begin{corollary}[$t$-modality]\label{theo:structural:tmodal}
For any $t\geq 1$ and all $\gamma > 0$, the class $\classtmo$ of $t$-modal distributions on $[n]$ is $(\gamma,L)$-decomposable for $L \eqdef \bigO{\frac{t\log^2 n}{\gamma}}$.
\end{corollary}

\begin{corollary}[Log-concavity, concavity and convexity]\label{theo:structural:logconcave}
For all $\gamma > 0$, the classes $\classlog$, $\classcve$ and $\classcvx$ of log-concave, concave and convex distributions on $[n]$ are $(\gamma,L)$-decomposable for $L \eqdef \bigO{\frac{\log^2 n}{\gamma}}$.
\end{corollary}
\begin{proof}
This is directly implied by \autoref{theo:structural:unimodal}, recalling that log-concave, concave and convex distributions are unimodal.
\end{proof}

\begin{restatable}[Monotone Hazard Rate]{theorem}{theostructuralmhr}\label{theo:structural:mhr}
For all $\gamma > 0$, the class $\classmhr$ of MHR distributions on $[n]$ is $(\gamma,L)$-decomposable for $L \eqdef \bigO{\frac{\log n}{\gamma}}$.
\end{restatable}
\begin{proof}
This follows from adapting the proof of \cite{CDSS:13}, which establishes that every MHR distribution can be approximated in $\lp[1]$ distance by a $\bigO{\log(n/\eps)/\eps}$-histogram. For completeness, we reproduce their argument, suitably modified to our purposes, in \autoref{app:structural:proofs}.
\end{proof}

\begin{theorem}[Piecewise Polynomials]\label{theo:structural:piecewise}
For all $\gamma > 0$, $t,d\geq 0$, the class $\classpoly[t,d]$ of $t$-piecewise degree-$d$ distributions on $[n]$ is $(\gamma,L)$-decomposable for $L \eqdef \bigO{\frac{t(d+1)}{\gamma}\log^2 n}$. (Moreover, for the class of $t$-histograms $\classhist$ ($d=0$) one can take $L = t$.)
\end{theorem}
\begin{proof}
The last part of the statement is obvious, so we focus on the first claim. Observing that each of the $t$ pieces of a distribution $\D\in\classpoly[t,d]$ can be subdivided in at most $d+1$ intervals on which $\D$ is monotone (being degree-$d$ polynomial on each such pieces), we obtain a partition of $[n]$ into at most $t(d+1)$ intervals. $\D$ being monotone on each of them, we can apply an argument almost identical to that of~\autoref{theo:structural:monotone} to argue that each interval can be further split into $O(\log^2 n/\gamma)$ subintervals, yielding a good decomposition with $O( t(d+1)\log^2 n/\gamma )$ pieces.
\end{proof}

\subsection{Projection Step: computing the distances} 
This section contains details of the distance estimation procedures for these classes, required in the last stage of \autoref{algo:test:splittable}. (Note that some of these results are phrased in terms of distance approximation, as estimating the distance $\lp[1](\D,\class)$ to sufficient accuracy in particular yields an algorithm for this stage.)

We focus in this section on achieving the sample complexities stated in~\autoref{coro:main:testing}, \autoref{coro:main:testing:tmod}, and~\autoref{coro:main:testing:piecewise}. While almost all the distance estimation procedures we give in this section are efficient, running in time polynomial in all the parameters or even with only a polylogarithmic dependence on $n$, there are two exceptions -- namely, the procedures for monotone hazard rate (\autoref{lemma:distance:mhr}) and log-concave (\autoref{lemma:distance:log}) distributions. We \emph{do} describe computationally efficient procedures for these two cases as well in~\autoref{ssec:efficient:logconcave:mhr}, at a modest additive cost in the sample complexity.

\begin{lemma}[Monotonicity {\cite[Lemma 8]{BKR:04}}]\label{lemma:distance:mon}
There exists a procedure $\estimdist[\classmon]$ that, on input $n$ as well as the full (succinct) specification of a $\ell$-histogram $\D$ on $[n]$, computes the (exact) distance $\lp[1](\D,\classmon)$ in time $\poly(\ell)$.
\end{lemma}

A straightforward modification of the algorithm above (e.g., by adapting the underlying linear program to take as input the location $m\in[\ell]$ of the mode of the distribution; then trying all $\ell$ possibilities, running the subroutine $\ell$ times and picking the minimum value) results in a similar claim for unimodal distributions:
\begin{lemma}[Unimodality]\label{lemma:distance:uni}
There exists a procedure $\estimdist[\classuni]$ that, on input $n$ as well as the full (succinct) specification of a $\ell$-histogram $\D$ on $[n]$, computes the (exact) distance $\lp[1](\D,\classuni)$ in time $\poly(\ell)$.
\end{lemma}
A similar result can easily be obtained for the class of $t$-modal distributions as well, with a $\poly(\ell,t)$-time algorithm based on a combination of dynamic and linear programming. 
 Analogous statements hold for the classes of concave and convex distributions $\classcvx, \classcve$, also based on linear programming (specifically, on running $\bigO{n^2}$ different linear programs -- one for each possible support $[a,b]\subseteq[n]$ -- and taking the minimum over them).

\begin{lemma}[MHR]\label{lemma:distance:mhr}
There exists a (non-efficient) procedure $\estimdist[\classmhr]$ that, on input $n$, $\eps$, as well as the full specification of a distribution $\D$ on $[n]$, distinguishes between $\lp[1](\D,\classmhr) \leq \eps$ and $\lp[1](\D,\classmhr)>2\eps$ in time~$2^{\tilde{O}_\eps(n)}$.
\end{lemma}
\begin{lemma}[Log-concavity]\label{lemma:distance:log}
There exists a (non-efficient) procedure $\estimdist[\classlog]$ that, on input $n$, $\eps$, as well as the full specification of a distribution $\D$ on $[n]$, distinguishes between $\lp[1](\D,\classlog) \leq \eps$ and $\lp[1](\D,\classlog)>2\eps$ in time~$2^{\tilde{O}_\eps(n)}$.
\end{lemma}
\begin{proof}[\autoref{lemma:distance:mhr} and \autoref{lemma:distance:log}]
We here give a naive algorithm for these two problems, based on an exhaustive search over a (huge) $\eps$-cover $\mathcal{S}$ of distributions over $[n]$. Essentially, $\mathcal{S}$ contains all possible distributions whose probabilities $p_1,\dots,p_n$ are of the form $j\eps/n$, for $j\in\{0,\dots,n/\eps\}$ (so that $\abs{\mathcal{S}} = \bigO{(n/\eps)^{n}}$). It is not hard to see that this indeed defines an \eps-cover of the set of all distributions, and moreover that it can be computed in time $\poly(\abs{\mathcal{S}})$. To approximate the distance from an explicit distribution $\D$ to the class $\class$ (either $\classmhr$ or $\classlog$), it is enough to go over every element $S$ of $\mathcal{S}$, checking (this time, efficiently) if $\normone{S-\D}\leq \eps$ and if there is a distribution $P\in\class$ close to $S$ (this time, pointwise, that is $\abs{P(i)-S(i)} \leq \eps/n$ for all $i$) -- which also implies $\normone{S-P}\leq \eps$ and thus $\normone{P-\D}\leq 2\eps$. The test for pointwise closeness can be done by checking feasibility of a linear program with variables corresponding to the logarithm of probabilities, i.e. $x_i \equiv \ln P(i)$. Indeed, this formulation allows to rephrase the log-concave and MHR constraints as linear constraints, and pointwise approximation is simply enforcing that $\ln(S(i)-\eps/n) \leq x_i \leq \ln(S(i)+\eps/n)$ for all $i$. At the end of this enumeration, the procedure accepts if and only if for some $S$ both $\normone{S-\D}\leq \eps$ and the corresponding linear program was feasible.
\end{proof}

\begin{lemma}[Piecewise Polynomials]\label{lemma:distance:piecewise}
There exists a procedure $\estimdist[\classpoly]$ that, on input $n$ as well as the full specification of an $\ell$-histogram $\D$ on $[n]$, computes an approximation $\Delta$ of the distance $\lp[1](\D,\classpoly)$ such that $\lp[1](\D,\classpoly) \leq \Delta \leq 3\lp[1](\D,\classpoly)+\eps$, and runs in time $\bigO{n^3}\cdot\poly(\ell,t,d,\frac{1}{\eps})$.

Moreover, for the special case of $t$-histograms ($d=0$) there exists a procedure $\estimdist[\classhist]$, which, given inputs as above, computes an approximation $\Delta$ of the distance $\lp[1](\D,\classhist)$ such that  $\lp[1](\D,\classhist) \leq \Delta \leq 4\lp[1](\D,\classhist)+\eps$, and runs in time $\poly(\ell,t,\frac{1}{\eps})$, independent of $n$.
\end{lemma}
\begin{proof}
We begin with $\estimdist[\classhist]$. Fix any distribution $\D$ on $[n]$. Given any explicit partition of $[n]$ into intervals $\mathcal{I}=(I_1,\dots,I_t)$, one can easily show that $\normone{\D - \Phi(\D,\mathcal{I})} \leq 2\opt_{\mathcal{I}}$, where $\opt_{\mathcal{I}}$ is the optimal distance of $\D$ to any histogram on $\mathcal{I}$. To get a $2$-approximation of $\lp[1](\D,\classhist)$, it thus suffices to find the minimum, over all possible partitionings $\mathcal{I}$ of $[n]$ into $t$ intervals, of the quantity $\normone{\D - \Phi(\D,\mathcal{I})}$ (which itself can be computed in time $T=O(\min(t\ell,n))$). By a simple dynamic programming approach, this can be performed in time $\bigO{t n^2 \cdot T}$. The quadratic dependence on $n$, which follows from allowing the endpoints of the $t$ intervals to be at any point of the domain, is however far from optimal and can be reduced to $(t/\eps)^2$, as we show below.

For $\eta > 0$, define an \emph{$\eta$-granular decomposition} of a distribution $\D$ over $[n]$ to be a partition of $[n]$ into $s=\bigO{1/\eta}$ intervals $J_1,\dots,J_s$ such that each interval $J_i$ is either a singleton or satisfies $\D(J_i) \leq \eta$. (Note that if $\D$ is a known $\ell$-histogram, one can compute an $\eta$-granular decomposition of $\D$ in time $\bigO{\ell/\eta}$ in a greedy fashion.)

\begin{claim}\label{claim:granularity:piecewise:projection}
Let $\D$ be a distribution over $[n]$, and $\mathcal{J} = (J_1,\dots,J_s)$ be an $\eta$-granular decomposition of $\D$ (with $s\geq t$). Then, there exists a partition of $[n]$ into $t$ intervals $\mathcal{I}=(I_1,\dots,I_t)$ and a $t$-histogram $H$ on $\mathcal{I}$ such that $\normone{\D - H} \leq 2\lp[1](\D,\classhist[t])+2t\eta$, and $\mathcal{I}$ is a coarsening of $\mathcal{J}$.
\end{claim}
Before proving it, we describe how this will enable us to get the desired time complexity for $\estimdist[\classhist]$. Phrased differently, the claim above allows us to run our dynamic program using the $\bigO{1/\eta}$ endpoints of the $\bigO{1/\eta}$ instead of the $n$ points of the domain, paying only an additive error $O(t\eta)$. Setting $\eta=\frac{\eps}{4t}$, the guarantee for $\estimdist[\classhist]$ follows.

\begin{proofof}{\autoref{claim:granularity:piecewise:projection}} Let $\mathcal{J} =(J_1,\dots,J_s)$ be an $\eta$-granular decomposition of $\D$, and $H^\ast\in\classhist[t]$ be a histogram achieving $\opt=\lp[1](\D,\classhist[t])$. Denote further by $\mathcal{I^\ast} = (I^\ast_1,\dots,I^\ast_t)$ the partition of $[n]$ corresponding to $H^\ast$. Consider now the $r\leq t$ endpoints of the $I^\ast_i$'s that do not fall on one of the endpoints of the $J_i$'s: let $J_{i_1},\dots,J_{i_r}$ be the respective intervals in which they fall (in particular, these cannot be singleton intervals), and $S=\cup_{j=1}^r J_{i_j}$ their union. By definition of $\eta$-granularity, $\D(S) \leq t\eta$, and it follows that $H^\ast(S)\leq t\eta + \frac{1}{2}\opt$. We define $H$ from $H^\ast$ in two stages: first, we obtain a (sub)distribution $H^\prime$ by modifying $H^\ast$ on $S$, setting for each $x\in J_{i_j}$ the value of $H$ to be the minimum value (among the two options) that $H^\ast$ takes on $J_{i_j}$. $H^\prime$ is thus a $t$-histogram, and the endpoints of its intervals are endpoints of $\mathcal{J}$ as wished; but it may not sum to one. However, by construction we have that $H^\prime([n]) \geq 1-H^\ast(S) \geq 1-t\eta - \frac{1}{2}\opt$. Using this, we can finally define our $t$-histogram distribution $H$ as the renormalization of $H^\prime$. It is easy to check that $H$ is a valid $t$-histogram on a coarsening of $\mathcal{J}$, and
\[
    \normone{\D-H} \leq \normone{\D-H^\prime} + (1-H^\prime([n])) \leq \normone{\D-H^\ast} + \normone{H^\ast-H^\prime} + t\eta + \frac{1}{2}\opt \leq 2\opt  + 2t\eta
\]
as stated.
\end{proofof}

Turning now to $\estimdist[\classpoly]$, we apply the same initial dynamic programming approach, which will result on a running time of $\bigO{n^2t\cdot T}$, where $T$ is the time required to estimate (to sufficient accuracy) the distance of a given (sub)distribution over an interval $I$ onto the space $\classpoly[d]$ of degree-$d$ polynomials. Specifically, we will invoke the following result, adapted from~\cite{CDSS:14} to our setting:
\begin{theorem} \label{theo:degreed:poly-projection}
Let $p$ be a $\ell$-histogram over $[-1,1)$. There is an algorithm $\textsc{ProjectSinglePoly}(d,\eta)$
which runs in time $\poly(\ell, d+1,1/\eta)$, and outputs a degree-$d$ polynomial $q$ which defines a pdf over $[-1,1) $
such that $\normone{p-q} \leq 3 \lp[1](p,\classpoly[d]) + O(\eta)$.
\end{theorem}
The proof of this modification of~\cite[Theorem 9]{CDSS:14} is deferred to~\autoref{app:structural:projection:proofs}. Applying it as a blackbox with $\eta$ set to $\bigO{\eps/t}$ and noting that computing the $\lp[1]$ distance to our explicit distribution on a given interval of the degree-$d$ polynomial returned incurs an additional $\bigO{n}$ factor, we obtain the claimed guarantee and running time.
\todonoteinline{Todo: Piecewise Polynomials projection without an $n^2$ or $n^3$ dependence.}
\end{proof}

\subsubsection{Computationally Efficient Procedures for Log-concave and MHR Distributions}\label{ssec:efficient:logconcave:mhr}

We now describe how to obtain \emph{efficient} testing for the classes \classlog and \classmhr{} -- that is, how to obtain polynomial-time distance estimation procedures for these two classes, unlike the ones described in the previous section. At a very high-level, the idea is in both case to write down a linear program on variables related \emph{logarithmically} to the probabilities we are searching, as enforcing the log-concave and MHR constraints on these new variables can be done linearly. The catch now becomes the $\lp[1]$ objective function (and, to a lesser extent, the fact that the probabilities must sum to one), now highly non-linear.

The first insight is to leverage the structure of log-concave (resp. monotone hazard rate) distributions to express this objective as slightly stronger constraints, specifically pointwise $(1\pm\eps)$-multiplicative closeness, much easier to enforce in our ``logarithmic formulation.'' Even so, doing this naively fails, essentially because of a too weak distance guarantee between our explicit histogram $\hat{\D}$ and the unknown distribution we are trying to find: in the completeness case, we are only promised $\eps$-closeness in $\lp[1]$, while we would also require good additive pointwise closeness of the order $\eps^2$ or~$\eps^3$.

The second insight is thus to observe that we ``almost'' have this for free: indeed, if we do not reject in the first stage of the testing algorithm, we do obtain an explicit $k$-histogram $\hat{\D}$ with the guarantee that $\D$ is $\eps$-close to the distribution $P$ to test. However, we \emph{also} implicitly have another distribution $\hat{\D}^\prime$ that is $\sqrt{\eps/k}$-close to $P$ \emph{in Kolmogorov distance}: as in the recursive descent we take enough samples to use the DKW inequality (\autoref{theo:dkw:ineq}) with this parameter, i.e. an additive overhead of $\bigO{k/\eps}$ samples (on top of the $\tilde{O}(\sqrt{kn}/\eps^{7/2})$). If we are willing to increase this overhead by just a small amount, that is to take $\tildeO{ \max(k/\eps, 1/\eps^4) }$, we can guarantee that $\hat{\D}^\prime$ be also $\tildeO{\eps^2}$-close to $P$ in Kolmogorov distance.\smallskip

\noindent Combining these ideas yield the following distance estimation lemmas:
\begin{restatable}[Monotone Hazard Rate]{lemma}{lemmaefficientdistancemhr}\label{lemma:distance:mhr:eff}
  There exists a procedure $\estimdist[\classmhr]^\ast$ that, on input $n$ as well as the full specification of a $k$-histogram distribution $\D$ on $[n]$ and of a $\ell$-histogram distribution $\D^{\prime}$ on $[n]$, runs in time $\poly(n,1/\eps)$, and satisfies the following.
  \begin{itemize}
    \item If there is $P\in\classmhr$ such that $\normone{\D-P} \leq \eps$  and $\kolmogorov{\D^\prime}{P} \leq \eps^3$, then the procedure returns~\yes;
    \item If $\lp[1](\D,\classmhr) > 100\eps$, then the procedure returns \no.
  \end{itemize}
\end{restatable}

\begin{restatable}[Log-concavity]{lemma}{lemmaefficientdistancelog}\label{lemma:distance:log:eff}
  There exists a procedure $\estimdist[\classlog]^\ast$ that, on input $n$ as well as the full specifications of a $k$-histogram distribution $\D$ on $[n]$ and a $\ell$-histogram distribution $\D^\prime$ on $[n]$, runs in time $\poly(n,k,\ell,1/\eps)$, and satisfies the following.
  \begin{itemize}
    \item If there is $P\in\classlog$ such that $\normone{\D-P}\leq \eps$ \emph{and} $\kolmogorov{\D^\prime}{P}\leq \frac{\eps^2}{\log^2(1/\eps)}$, then the procedure returns~\yes;
    \item If $\lp[1](\D,\classlog) \geq 100\eps$, then the procedure returns \no.
  \end{itemize}
\end{restatable}

The proofs of these two lemmas are quite technical and deferred to~\autoref{app:structural:projection:proofs}. With these in hand, a simple modification of our main algorithm (specifically, setting $m = \tilde{O}( \max({\sqrt{\abs{I}}}/{\eps^3} L, {L^2}/{\eps^2}, {1}/{\eps^c} ) )$ for $c$ either $4$ or $6$ instead of $\tilde{O}( \max ( {\sqrt{\abs{I}}}/{\eps^3} L, {L^2}/{\eps^2} ) )$, to get the desired Kolmogorov distance guarantee; and providing the empirical histogram defined by these $m$ samples along to the distance estimation procedure) suffices to obtain the following counterpart to~\autoref{coro:main:testing}:
\begin{corollary}\label{coro:main:testing:eff}
The algorithm \textsc{TestSplittable}, after this modification, can \emph{efficiently} test the classes of log-concave and monotone hazard rate (MHR) distributions, with respectively $\tilde{O}\big({\sqrt{n}/\eps^{7/2} + 1/\eps^4}\big)$ and $\tilde{O}\big({\sqrt{n}/\eps^{7/2} + 1/\eps^6}\big)$ samples.
\end{corollary}
 
\section{Going Further: Reducing the Support Size}\label{sec:effectivesupport}
  \makeatletter{}The general approach we have been following so far gives, out-of-the-box, an efficient testing algorithm with sample complexity $\tildeO{\sqrt{n}}$ for a large range of properties. However, this sample complexity can for some classes \property be brought down a lot more, by taking advantage in a preprocessing step of good concentration guarantees of distributions in \property.\medskip

\noindent As a motivating example, consider the class of Poisson Binomial Distributions (PBD). It is well-known (see e.g.~\cite[Section 2]{KG:71}) that PBDs are unimodal, and more specifically that $\classpbd_n\subseteq\classlog\subseteq\classuni$. Therefore, using our generic framework we can test Poisson Binomial Distributions with $\tildeO{\sqrt{n}}$ samples. This is, however, far from optimal: as shown in~\cite{AD:15}, a sample complexity of $\bigTheta{n^{1/4}}$ is both necessary and sufficient. The reason our general algorithm ends up making quadratically too many queries can be explained as follows. PBDs are tightly concentrated around their expectation, so that they ``morally'' live on a support of size $m=\bigO{\sqrt{n}}$. Yet, instead of testing them on this very small support, in the above we still consider the entire range $[n]$, and thus end up paying a dependence $\sqrt{n}$ -- instead of $\sqrt{m}$.

If we could use that observation to first reduce the domain to the \emph{effective support} of the distribution, then we could call our testing algorithm on this reduced domain of size $\bigO{\sqrt{n}}$. In the rest of this section, we formalize and develop this idea, and in~\autoref{ssec:testing:pbds} will obtain as a direct application a $\tildeO{n^{1/4}}$-query testing algorithm for $\classpbd_n$.

\begin{definition}\label{def:effective:support}
Given $\eps > 0$, the \emph{\eps-effective support} of a distribution $\D$ is the smallest interval $I$ such that $\D(I) \geq 1-\eps$.
\end{definition}

The last definition we shall require is of the \emph{conditioned distributions} of a class \class:
\begin{definition}
For any class of distributions \class over $[n]$, define the set of \emph{conditioned distributions of \class} (with respect to $\eps > 0$ and interval $I\subseteq[n]$) as $\class^{\eps,I}\eqdef \setOfSuchThat{\D_I}{\D\in\class, \D(I) \geq 1-\eps}$.
\end{definition}

Finally, we will require the following simple result:
\begin{lemma}\label{lemma:conditioned:class}
Let $\D$ be a distribution over $[n]$, and $I\subseteq[n]$ an interval such that $\D(I) \geq 1- \frac{\eps}{10}$. Then,
\begin{itemize}
  \item If $\D\in\class$, then $\D_I\in\class^{\frac{\eps}{10},I}$;
  \item If $\lp[1](\D,\class) > \eps$, then $\lp[1](\D_I, \class^{\frac{\eps}{10},I}) > \frac{7\eps}{10}$.
\end{itemize}
\end{lemma}
\begin{proof}
The first item is obvious. As for the second, let $P\in\class$ be any distribution with $P(I)\geq 1-\frac{\eps}{10}$. By assumption, $\normone{\D - P} > \eps$: but we have, writing $\alpha=1/10$,
  \begin{align*}
    \normone{\D_I - P_I} &=\sum_{i\in I}\abs{ \frac{\D(i)}{\D(I)} - \frac{P(i)}{P(I)} }
    = \frac{1}{\D(I)}\sum_{i\in I}\abs{ \D(i) - P(i) + P(i)\mleft(1- \frac{\D(I)}{P(I)}\mright) } \\
    &\geq \frac{1}{\D(I)}\big( \sum_{i\in I}\abs{ \D(i) - P(i) } - \abs{1- \frac{\D(I)}{P(I)} } \sum_{i\in I} P(i) \big)\\
    &= \frac{1}{\D(I)}\big( \sum_{i\in I}\abs{ \D(i) - P(i) } - \abs{P(I)- \D(I) } \big)
    \geq \frac{1}{\D(I)}\big( \sum_{i\in I}\abs{ \D(i) - P(i) } - \alpha\eps \big)\\
   &\geq \frac{1}{\D(I)}\big( \normone{\D-P} - \sum_{i\notin I} \abs{ \D(i) - P(i) } - \alpha\eps \big)
   \geq \frac{1}{\D(I)}\big( \normone{\D-P} - 3\alpha\eps \big) \\
   &> (1- 3\alpha)\eps  
   = \frac{7}{10}\eps.
  \end{align*}
\end{proof}

We now proceed to state and prove our result -- namely, efficient testing of \emph{structured} classes of distributions with nice \emph{concentration properties}.
\begin{theorem}\label{theo:main:testing:effective:support}
  Let $\class$ be a class of distributions over $[n]$ for which the following holds.
    \begin{enumerate}
      \item there is a function $M(\cdot,\cdot)$ such that each $\D\in\class$ has \eps-effective support of size at most $M(n,\eps)$;
      \item for every $\eps \in [0,1]$ and interval $I\subseteq[n]$, $\class^{\eps,I}$ is $(\gamma,L)$-splittable;
      \item there exists an efficient procedure $\estimdist[\class^{\eps,I}]$ which, given as input the explicit description of a distribution $\D$ over $[n]$ and interval $I\subseteq[n]$, computes the distance $\lp[1](\D_I,\class^{\eps,I})$.
    \end{enumerate}
    Then, the algorithm \textsc{TestEffectiveSplittable} (\autoref{algo:test:effective:support:splittable}) is a $\bigO{ \max\mleft(\frac{1}{\eps^3} \sqrt{m\ell} \log m, \frac{\ell^2}{\eps^2}\mright) }$-sample tester for $\class$, where $m=M(n,\frac{\eps}{60})$ and $\ell=L(\frac{\eps}{1200}, m)$.
\end{theorem}

\begin{algorithm}[H]
  \algblock[block]{Start}{Start}
  \algblockdefx[]{Start}{End}    [1]{\textsc{#1}}    [1][]{\textsc{#1}}
  \begin{algorithmic}[1]
    \Require Domain $\domain$ (interval of size $n$), sample access to $\D$ over $\domain$; subroutine $\estimdist[\class^{\eps,I}]$
    \renewcommand{\algorithmicrequire}{\textbf{Input:}}
    \Require Parameters $\eps\in(0,1]$, function $L(\cdot,\cdot)$, and upper bound function $M(\cdot,\cdot)$ for the effective support of the class \class.
    \State Set $m\eqdef \bigO{1/\eps^2}$, $\tau \eqdef M(n,\frac{\eps}{60})$.     \Start{Effective Support}
      \State\label{step:effesupp:get:samples} Compute $\hat{\D}$, an empirical estimate of $\D$, by drawing $m$ independent  samples from $\D$.
      \State\label{step:effesupp:int:J} Let $J$ be the largest interval of the form $\{1,\dots,j\}$ such that $\hat{\D}(J) \leq \frac{\eps}{30}$.
      \State\label{step:effesupp:int:K} Let $K$ be the largest interval of the form $\{k,\dots,n\}$ such that $\hat{\D}(K) \leq \frac{\eps}{30}$.
      \State\label{step:effesupp:test:support} Set $I\gets [n]\setminus(J\cup K)$.
      \If{ $\abs{I} > \tau$ } \Return \reject \EndIf
    \End
    \Start{Testing}
      \State\label{step:effesupp:main:test:call} Call $\textsc{TestSplittable}$ with $I$ (providing simulated access to $\D_I$ by rejection sampling, returning \fail if the number of samples $q$ from $\D_I$ required by the subroutine is not obtained after $\bigO{q}$ samples from $\D$), $\estimdist[\class^{\eps,I}]$, parameters $\eps^\prime\eqdef \frac{7\eps}{10}$ and $L(\cdot,\cdot)$.
      \State\label{step:effesupp:main:test} \Return \accept if $\textsc{TestSplittable}$ accepts, \reject otherwise.
    \End
  \end{algorithmic}
  \caption{\label{algo:test:effective:support:splittable}\sc TestEffectiveSplittable}
\end{algorithm}

\subsection{Proof of~\autoref{theo:main:testing:effective:support}}
By the choice of $m$ and the DKW inequality, with probability at least $23/24$ the estimate $\hat{\D}$ satisfies $\kolmogorov{\D}{\hat{\D}} \leq \frac{\eps}{60}$. Conditioning on that from now on, we get that $\D(I) \geq \hat{\D}(I) - \frac{\eps}{30} \geq 1-\frac{\eps}{10}$. Furthermore, denoting by $j$ and $k$ the two inner endpoints of $J$ and $K$ in Steps~\ref{step:effesupp:int:J} and \ref{step:effesupp:int:K}, we have $\D(J\cup\{j+1\}) \geq \hat{\D}(J\cup\{j+1\}) -  \frac{\eps}{60}  > \frac{\eps}{60}$ (similarly for $\D(K\cup\{k-1\})$), so that $I$ has size at most $\sigma+1$, where $\sigma$ is the $\frac{\eps}{60}$-effective support size of $\D$. 

Finally, note that since $\D(I) = \bigOmega{1}$ by our conditioning, the simulation of samples by rejection sampling will succeed with probability at least $23/24$ and the algorithm will not output \fail.

\paragraph{Sample complexity.}
The sample complexity is the sum of the $\bigO{1/\eps^2}$ in Step~\ref{step:effesupp:get:samples} and the $\bigO{q}$ in Step~\ref{step:effesupp:main:test:call}. From~\autoref{theo:main:testing} and the choice of $I$, this latter quantity is $\bigO{ \max\mleft(\frac{1}{\eps^3} \sqrt{M\ell} \log M, \frac{\ell^2}{\eps^2}\mright) }$ where $M = M(n,\frac{\eps}{60})$ and $\ell=L(\frac{\eps}{1200}, M(n,\frac{\eps}{60}))$.

\paragraph{Correctness.} If $\D\in\class$, then by the setting of $\tau$ (set to be an upper bound on the $\frac{\eps}{60}$-effective support size of any distribution in \class) the algorithm will go beyond Step~\ref{step:effesupp:test:support}. The call to $\textsc{TestSplittable}$ will then end up in the algorithm returning \accept in Step~\ref{step:effesupp:main:test}, with probability at least $2/3$ by~\autoref{lemma:conditioned:class},~\autoref{theo:main:testing} and our choice of parameters.

Similarly, if $\D$ is \eps-far from \class, then either its effective support is too large (and then the test on Step~\ref{step:effesupp:test:support} fails), or the main tester will detect that its conditional distribution on $I$ is $\frac{7\eps}{10}$-far from $\class$ and output \reject in Step~\ref{step:effesupp:main:test}.

Overall, in either case the algorithm is correct except with probability at most $1/24+1/24+1/3=5/12$ (by a union bound). Repeating constantly many times and outputting the majority vote brings the probability of failure down to $1/3$. \qed

\subsection{Application: Testing Poisson Binomial Distributions}\label{ssec:testing:pbds}

In this section, we illustrate the use of our generic two-stage approach to test the class of Poisson Binomial Distributions. Specifically, we prove the following result:
\begin{corollary}\label{coro:main:testing:effective:support:pbd}
The class of Poisson Binomial Distributions can be tested with $\tildeO{{n}^{1/4}/\eps^{7/2}} + \bigO{\log^4 n/\eps^4 }$ samples, using~\autoref{algo:test:effective:support:splittable}.
\end{corollary}

This is a direct consequence of~\autoref{theo:main:testing:effective:support} and the lemmas below. The first one states that, indeed, PBDs have small effective support:
\begin{fact}\label{fact:pbd:effective:support}
For any $\eps > 0$, a PBD has \eps-effective support of size $\bigO{\sqrt{n\log(1/\eps)}}$.
\end{fact}
\begin{proof}
By an additive Chernoff Bound, any random variable $X$ following a Poisson Binomial Distribution has $\probaOf{\abs{X-\shortexpect X} > \gamma n} \leq 2e^{-2\gamma^2n}$. Taking $\gamma\eqdef\sqrt{\frac{1}{2n}\ln\frac{2}{\eps}}$, we get that $\probaOf{X\in I} \geq 1-\eps$, where $I\eqdef [\shortexpect X-\sqrt{\frac{1}{2}\ln\frac{2}{\eps}}, \shortexpect X+\sqrt{\frac{1}{2}\ln\frac{2}{\eps}}]$.
\end{proof}

It is clear that if $\D\in\classpbd_n$ (and therefore is unimodal), then for any interval $I\subseteq[n]$ the conditional distribution $\D_I$ is still unimodal, and thus the class of \emph{conditioned PBDs} $\classpbd_n^{\eps,I}\eqdef \setOfSuchThat{\D_I}{\D\in\classpbd_n, \D(I) \geq 1-\eps}$ falls under~\autoref{theo:structural:unimodal}. The last piece we need to apply our generic testing framework is the existence of an algorithm to compute the distance between an (explicit) distribution and the class of conditioned PBDs. This is provided by our next lemma:
\begin{claim}\label{lemma:distance:pbd}
There exists a procedure $\estimdist[\classpbd_n^{\eps,I}]$ that, on input $n$ and $\eps, \in [0,1]$, $I\subseteq[n]$ as well as the full specification of a distribution $\D$ on $[n]$, computes a value $\tau$ such that $\tau \in [1\pm 2\eps] \cdot \lp[1](\D,\classpbd_n^{\eps,I}) \pm \frac{\eps}{100}$, in time $n^2 \left( {1/\eps} \right)^{\bigO{\log{1/\eps}}}$.
\end{claim}
\begin{proof}
The goal is to find a $\gamma = \Theta(\eps)$-approximation of the minimum value of $\sum_{i\in I}\abs{\frac{P(i)}{P(I)} - \frac{\D(i)}{\D(I)}}$, subject to $P(I)=\sum_{i\in I} P(i) \geq 1-\eps$ and $P\in\classpbd_n$. We first note that, given the parameters $n \in \N$ and $p_1,\dots,p_n\in[0,1]$ of a PBD $P$, the vector of $(n+1)$ probabilities $P(0),\dots,P(n)$ can be obtained in time $\bigO{n^2}$ by dynamic programming. 
Therefore, computing the $\lp[1]$ distance between $\D$ and any PBD with known parameters can be done efficiently. To conclude, we invoke a result of Diakonikolas, Kane, and Stewart, that guarantees the existence of a succinct (proper) cover of $\classpbd_n$:
\begin{theorem}[{\cite[Theorem 14]{DKS:15}} (rephrased)]
  For all $n, \gamma >0$, there exists a set $\mathcal{S}_{\gamma} \subseteq \classpbd_n$ such that:
  \begin{enumerate}[(i)]
  \item $\mathcal{S}_{\gamma}$ is a $\gamma$-cover of $\classpbd_n$; that is, for all $\D \in \classpbd_n$ there exists some $\D^\prime \in \mathcal{S}_{\gamma}$ such that $\normone{\D-\D^\prime} \leq \gamma$
  \item {$\abs{\mathcal{S}_{\gamma}} \leq n \left({1/\gamma}\right)^{\bigO{ \log{1/\gamma}}}$}
  \item $\mathcal{S}_{\gamma}$ can be computed in time {$n\left( {1/\gamma} \right)^{\bigO{\log{1/\gamma}}}$}
  \end{enumerate}
and each $\D\in\mathcal{S}_{\gamma}$  is explicitly described by its set of parameters.
\end{theorem}
\noindent We further observe that the factor $n$ in both the size of the cover and running time can be easily removed in our case, as we know a good approximation of the support size of the candidate PBDs. (That is, we only need to enumerate over a subset of the cover of~\cite{DKS:15}, that of the PBDs with effective support compatible with our distribution $\D$.)

Set $\gamma\eqdef\frac{\eps}{250}$. Fix $P\in\classpbd_n$ such that $P(I) \geq 1-\eps$, and $Q\in \mathcal{S}_{\gamma}$ such that $\normone{P-Q}\leq \gamma$. In particular, it is easy to see via the correspondence between $\lp[1]$ and total variation distance that $\abs{P(I)-Q(I)} \leq \gamma/2$.
By a calculation analogue as in~\autoref{lemma:conditioned:class}, we have 
\begin{align*}
  \normone{P_I-Q_I} &= \sum_{i\in I}\abs{\frac{P(i)}{P(I)} - \frac{Q(i)}{Q(I)}} 
  = \sum_{i\in I}\abs{\frac{P(i)}{P(I)} - \frac{Q(i)}{P(I)} + Q(i)\left( \frac{1}{P(I)} - \frac{1}{Q(I)} \right) } \\
  &= \sum_{i\in I}\abs{\frac{P(i)}{P(I)} - \frac{Q(i)}{P(I)} } \pm \sum_{i\in I} Q(i)\abs{ \frac{1}{P(I)} - \frac{1}{Q(I)} }
  = \frac{1}{P(I)}\left( \sum_{i\in I}\abs{ P(i) - Q(i) } \pm \abs{P(I)-Q(I)}\right) \\
  &= \frac{1}{P(I)}\left( \sum_{i\in I}\abs{ P(i) - Q(i) } \pm \frac{\gamma}{2}\right) = \frac{1}{P(I)}\left( \normone{ P - Q } \pm \frac{5\gamma}{2}\right) \\
  &\in [ \normone{ P - Q } - {5\gamma}/{2}, (1+2\eps)\left( \normone{ P - Q } + {5\gamma}/{2} \right) ]
  \end{align*}
where we used the fact that $\sum_{i\notin I}\abs{ P(i) - Q(i) } = 2\left(\sum_{i\notin I\colon P(i) > Q(i)} (P(i)-Q(i))\right) + Q(I)-P(I) \in [-2\gamma,2\gamma]$.
By the triangle inequality, this implies that the minimum of $\normone{P_I-\D_I}$ over the distributions $P$ of $\mathcal{S}_\eps$ with $P(I)\geq 1-(\eps+\gamma/2)$ will be within an additive $\bigO{\eps}$ of $\lp[1](\D,\classpbd_n^{\eps,I})$. The fact that the former can be done in time $\poly(n) \cdot \left( {1/\eps} \right)^{\bigO{\log^2{1/\eps}}}$ concludes the proof.
\end{proof}
As previously mentioned, this approximation guarantee for $\lp[1](\D,\classpbd_n^{\eps,I})$ is sufficient for the purpose of~\autoref{algo:test:splittable}.

\begin{proofof}{\autoref{coro:main:testing:effective:support:pbd}}
Combining the above, we invoke~\autoref{theo:main:testing:effective:support} with $M(n,\eps)=O( \sqrt{n\log(1/\eps)} )$ (\autoref{fact:pbd:effective:support}) and $L(m,\gamma)=O\big( \frac{\log^2 m}{\gamma} \big)$ (\autoref{theo:structural:unimodal}). This yields the claimed sample complexity; finally, the efficiency is a direct consequence of~\autoref{lemma:distance:pbd}.
\end{proofof}
 
\section{Lower Bounds}\label{sec:lowerbounds}
  \makeatletter{}\subsection{Reduction-based Lower Bound Approach} \label{ssec:lb-red}

We now turn to proving converses to our positive results -- namely, that many of the upper bounds we obtain cannot be significantly improved upon. As in our algorithmic approach, we describe for this purpose a \emph{generic framework} for obtaining lower bounds.

In order to state our results, we will require the usual definition of \emph{agnostic learning}. Recall that an algorithm is said to be a \emph{semi-agnostic learner} for a class \class if it satisfies the following. Given sample access to an arbitrary distribution $\D$ and parameter $\eps$, it outputs a hypothesis $\hat{\D}$ which (with high probability) does ``almost as well as it gets'':
\[
	\normone{\D - \hat{\D}} \leq c\cdot\opt_{\class,\D} + \bigO{\eps}
\]
where $\opt_{\class,\D}\eqdef \inf_{\D^\prime\in\class} \lp[1](\D^\prime,\D)$, and $c\geq 1$ is some absolute constant (if $c=1$, the learner is said to be agnostic).

\paragraph*{High-level idea.} The motivation for our result is the observation of~\cite{BKR:04} that ``monotonicity is at least as hard as uniformity.'' Unfortunately, their specific argument does not generalize easily to other classes of distributions, making it impossible to extend it readily. The starting point of our approach is to observe that while uniformity testing is hard in general, it becomes very easy \emph{under the promise that the distribution is monotone, or even only close to monotone} (namely, $\bigO{1/\eps^2}$ samples suffice). This can give an alternate proof of the lower bound for monotonicity testing, via a different reduction: first, test if  the unknown distribution is monotone; if it is, test whether it is uniform, now assuming closeness to monotone.

More generally, this idea applies to any class \class which \textsf{(a)} contains the uniform distribution, and \textsf{(b)} for which we have a $\littleO{\sqrt{n}}$-sample agnostic learner $\Learner$, as follows.
Assuming we have a tester $\Tester$ for $\class$ with sample complexity $\littleO{\sqrt{n}}$, define a uniformity tester as below.
\begin{itemize}
  \item test if $\D\in\class$ using $\Tester$; if not, reject (as $\uniform\in\class$, $\D$ cannot be uniform);
  \item otherwise, agnostically learn $\D$ with $\Learner$ (since $\D$ is close to $\class$), and obtain hypothesis $\hat{\D}$;
  \item check offline if $\hat{D}$ is close to uniform.
\end{itemize} 
By assumption, $\Tester$ and $\Learner$ each use $\littleO{\sqrt{n}}$ samples, so does the whole process; but this contradicts the lower bound of~\cite{BFRSW:00,Paninski:08} on uniformity testing. Hence, $\Tester$ must use $\bigOmega{\sqrt{n}}$ samples.\medskip

This ``testing-by-narrowing'' reduction argument can be further extended to other properties than to uniformity, as we show below:
\begin{theorem}\label{theo:main:testing:lb}
Let \class be a class of distributions over $[n]$ for which the following holds:
\begin{enumerate}[(i)]
  \item there exists a semi-agnostic learner $\Learner$ for $\class$, with sample complexity $q_L(n,\eps, \delta)$ and ``agnostic constant''~$c$;
  \item there exists a subclass $\class_{\rm Hard}\subseteq\class$ such that testing $\class_{\rm Hard}$ requires $q_H(n,\eps)$ samples.
\end{enumerate}
Suppose further that $q_L(n,\eps, 1/10)=\littleO{q_H(n,\eps)}$. Then, any tester for $\class$ must use $\bigOmega{q_H(n,\eps)}$ samples.
\end{theorem}
\begin{proof}
The above theorem relies on the reduction outlined above, which we rigorously detail here. Assuming $\class$, $\class_{\rm Hard}$, $\Learner$ as above (with semi-agnostic constant $c \geq 1$), and a tester $\Tester$ for \class with sample complexity $q_T(n,\eps)$, we define a tester $\Tester_{\rm Hard}$ for $\class_{\rm Hard}$. On input $\eps\in(0,1]$ and given sample access to a distribution $\D$ on $[n]$, $\Tester_{\rm Hard}$ acts as follows:
\begin{itemize}
  \item call $\Tester$ with parameters $n$, $\frac{\eps^\prime}{c}$ (where $\eps^\prime\eqdef\frac{\eps}{3}$) and failure probability $1/6$, to $\frac{\eps^\prime}{c}$-test if $\D\in\class$. If not, reject.
  \item otherwise, agnostically learn a hypothesis $\hat{\D}$ for $\D$, with $\Learner$ called with parameters $n$, $\eps^\prime$ and failure probability $1/6$;
  \item check offline if $\hat{\D}$ is $\eps^\prime$-close to $\class_{\rm Hard}$, accept if and only if this is the case.
\end{itemize} 
We condition on both calls (to $\Tester$ and $\Learner$) to be successful, which overall happens with probability at least $2/3$ by a union bound. The completeness is immediate: if $\D\in\class_{\rm Hard}\subseteq\class$, $\Tester$ accepts, and the hypothesis $\hat{\D}$ satisfies $\normone{\hat{\D}-\D} \leq \eps^\prime$. Therefore, $\lp[1](\hat{\D},\class_{\rm Hard}) \leq \eps^\prime$, and $\Tester_{\rm Hard}$ accepts.

For the soundness, we proceed by contrapositive. Suppose $\Tester_{\rm Hard}$ accepts; it means that each step was successful. In particular, $\lp[1](\hat{\D},\class)\leq {\eps^\prime}/{c}$; so that the hypothesis outputted by the agnostic learner satisfies $\normone{\hat{\D}-\D} \leq c\cdot\opt+\eps^\prime\leq 2\eps^\prime$. In turn, since the last step passed and by a triangle inequality  we get, as claimed, $\lp[1](\D, \class_{\rm Hard}) \leq 2\eps^\prime + \lp[1](\hat{\D},\class_{\rm Hard}) \leq 3\eps^\prime = \eps$.

Observing that the overall sample complexity is $q_T(n,\frac{\eps^\prime}{c})+q_L(n,\eps^\prime, \frac{1}{10}) = q_T(n,\frac{\eps^\prime}{c})+\littleO{q_H(n,\eps^\prime)}$ concludes the proof.
\end{proof}

Taking $\class_{\rm Hard}$ to be the singleton consisting of the uniform distribution, and from the semi-agnostic learners of~\cite{CDSS:13,CDSS:14} (each with sample complexity either $\poly(1/\eps)$ or $\poly(\log n,1/\eps)$), we obtain the following:\footnote{Specifically, these lower bounds hold as long as $\eps=\bigOmega{1/n^\alpha}$ for some absolute constant $\alpha > 0$ (so that the sample complexity of the agnostic learner is indeed negligible in front of $\sqrt{n}/\eps^2$).}
\corolbsqrtn*
Similarly, we can use another result of~\cite{DDS:PBD:12} which shows how to agnostically learn Poisson Binomial Distributions with $\tildeO{1/\eps^2}$ samples.\footnote{Note the quasi-quadratic dependence on $\eps$ of the learner, which allows us to get $\eps$ into our lower bound for $n\gg \poly\log(1/\eps)$.} Taking $\class_{\rm Hard}$ to be the single $\binomial{n}{1/2}$ distribution (along with the testing lower bound of~\cite{VV:14}), this yields the following:
\corolbpbd*

\noindent Finally, we derive a lower bound on testing $k$-SIIRVs from the agnostic learner of~\cite{DDDOST:13} (which has sample complexity $\poly(k,1/\eps)$ samples, independent~of~$n$):
\corolbksiirv*
\begin{proof}[\autoref{coro:lb:ksiirv}]
To prove this result, it is enough by \autoref{theo:main:testing:lb} to exhibit a particular $k$-SIIRV $S$ such that testing identity to $S$ requires this many samples. Moreover, from~\cite{VV:14} this last part amounts to proving that the (truncated) 2/3-norm $\norm{S^{-\max}_{-\eps_0}}_{2/3}$ of $S$ is $\bigOmega{k^{1/2}n^{1/4}}$ (for some small $\eps_0>0$). Our hard instance $S$ will be defined as follows: it is defined as the distribution of $X_1+\dots+X_n$, where the $X_i$'s are independent integer random variables uniform on $\{0,\dots,k-1\}$ (in particular, for $k=2$ we get a $\binomial{n}{1/2}$ distribution). It is straightforward to verify that $\shortexpect{S} = \frac{n(k-1)}{2}$ and $\sigma^2\eqdef\var S = \frac{(k^2-1)n}{12} = \bigTheta{k^2 n}$; moreover, $S$ is log-concave (as the convolution of $n$ uniform distributions). From this last point, we get that \textsf{(i)} the maximum probability of $S$, attained at its mode, is $\norminf{S}=\bigTheta{1/\sigma}$; and \textsf{(ii)} for every $j$ in an interval $I$ of length $2\sigma$ centered at this mode, $S(j) \geq \bigOmega{\norminf{S}}$. Putting this together, we get that the 2/3-norm (and similarly the truncated 2/3-norm) of $S$ is lower bounded by
\[
    \Big(\sum_{j\in I} S(j)^{2/3}\Big)^{3/2} \geq \mleft(2\sigma\cdot \bigOmega{1/\sigma}^{2/3}\mright)^{3/2} =\bigOmega{\sigma^{1/2}} = \bigOmega{k^{1/2}n^{1/4}}
\]
which concludes the proof.
\end{proof}

\subsection{Tolerant Testing}\label{sec:lowerbounds:tol}

This lower bound framework from the previous section carries to \emph{tolerant} testing as well, resulting in this analogue to~\autoref{theo:main:testing:lb}:
\begin{theorem}\label{theo:main:testing:tol:lb}
Let \class be a class of distributions over $[n]$ for which the following holds:
\begin{enumerate}[(i)]
  \item there exists a semi-agnostic learner $\Learner$ for $\class$, with sample complexity $q_L(n,\eps, \delta)$ and ``agnostic constant''~$c$;
  \item there exists a subclass $\class_{\rm Hard}\subseteq\class$ such that \emph{tolerant} testing $\class_{\rm Hard}$ requires $q_H(n,\eps_1,\eps_2)$ samples for some parameters $\eps_2 > (4c+1)\eps_1$.
\end{enumerate}
Suppose further that $q_L(n,\eps_2-\eps_1, 1/10)=\littleO{q_H(n,\eps_1,\eps_2)}$. Then, any \emph{tolerant} tester for $\class$ must use $\bigOmega{q_H(n,\eps_1,\eps_2)}$ samples (for some explicit parameters $\eps^\prime_1,\eps^\prime_2$).
\end{theorem}
\begin{proof}
The argument follows the same ideas as for \autoref{theo:main:testing:lb}, up to the details of the parameters. Assuming $\class$, $\class_{\rm Hard}$, $\Learner$ as above (with semi-agnostic constant $c \geq 1$), and a tolerant tester $\Tester$ for \class with sample complexity $q(n,\eps_1,\eps_2)$, we define a tolerant tester $\Tester_{\rm Hard}$ for $\class_{\rm Hard}$. On input $0 < \eps_1 < \eps_2 \leq 1$ with $\eps_2 > (4c+1)\eps_1$, and given sample access to a distribution $\D$ on $[n]$, $\Tester_{\rm Hard}$ acts as follows. After setting $\eps^\prime_1\eqdef\frac{\eps_2-\eps_1}{4}$, $\eps^\prime_2\eqdef\frac{\eps_2-\eps_1}{2}$, $\eps^\prime \eqdef \frac{\eps_2-\eps_1}{16}$ and $\tau\eqdef \frac{6\eps_2+10\eps_1}{16}$,
\begin{itemize}
  \item call $\Tester$ with parameters $n$, $\frac{\eps^\prime_1}{c}$, $\frac{\eps^\prime_2}{c}$ and failure probability $1/6$, to tolerantly test if $\D\in\class$. If $\lp[1](\D,\class) > \eps^\prime_2/c$, reject.
  \item otherwise, agnostically learn a hypothesis $\hat{\D}$ for $\D$, with $\Learner$ called with parameters $n$, $\eps^\prime$ and failure probability $1/6$;
  \item check offline if $\hat{\D}$ is $\tau$-close to $\class_{\rm Hard}$, accept if and only if this is the case.
\end{itemize} 
We condition on both calls (to $\Tester$ and $\Learner$) to be successful, which overall happens with probability at least $2/3$ by a union bound. We first argue completeness: 
assume $\lp[1](\D,\class_{\rm Hard}) \leq \eps_1$. This implies $\lp[1](\D,\class) \leq \eps_1$, so that $\Tester$ accepts as $\eps_1 \leq \eps^\prime_1/c$ (which is the case because $\eps_2 > (4c+1)\eps_1$). Thus, the hypothesis $\hat{\D}$ satisfies $\normone{\hat{\D}-\D} \leq c\cdot \eps^\prime_1/c + \eps^\prime = \eps^\prime_1 + \eps^\prime$. Therefore, 
  $\lp[1](\hat{\D},\class_{\rm Hard}) \leq \normone{\hat{\D}-\D} + \lp[1](\D,\class_{\rm Hard}) \leq \eps^\prime_1 + \eps^\prime + \eps_1 < \tau$, and $\Tester_{\rm Hard}$ accepts.

For the soundness, we again proceed by contrapositive. Suppose $\Tester_{\rm Hard}$ accepts; it means that each step was successful. In particular, $\lp[1](\hat{\D},\class)\leq {\eps^\prime_2}/{c}$; so that the hypothesis outputted by the agnostic learner satisfies $\normone{\hat{\D}-\D} \leq c\cdot\opt+\eps^\prime\leq \eps^\prime_2 + \eps^\prime$. In turn, since the last step passed and by a triangle inequality  we get, as claimed, 
$\lp[1](\D, \class_{\rm Hard}) \leq \eps^\prime_2 + \eps^\prime + \lp[1](\hat{\D},\class_{\rm Hard}) \leq \eps^\prime_2 + \eps^\prime + \tau < \eps_2$.

Observing that the overall sample complexity is $q_T(n,\frac{\eps^\prime_1}{c},\frac{\eps^\prime_2}{c})+q_L(n,\eps^\prime, \frac{1}{10}) = q_T(n,\frac{\eps^\prime}{c})+\littleO{q_H(n,\eps^\prime)}$ concludes the proof.
\end{proof}

As before, we instantiate the general theorem to obtain specific lower bounds for tolerant testing of the classes we covered in this paper. That is, taking $\class_{\rm Hard}$ to be the singleton consisting of the uniform distribution (combined with the tolerant testing lower bound of~\cite{ValiantValiant:10lb}), and again from the semi-agnostic learners of~\cite{CDSS:13,CDSS:14} (each with sample complexity either $\poly(1/\eps)$ or $\poly(\log n,1/\eps)$), we obtain the following:
\corolbtolnlogn*
Similarly, we again turn to the class of Poisson Binomial Distributions, for which we can invoke as before the $\tildeO{1/\eps^2}$-sample agnostic learner of~\cite{DDS:PBD:12}. As before, we would like to choose for $\class_{\rm Hard}$ the single $\binomial{n}{1/2}$ distribution; however, as no tolerant testing lower bound for this distribution exists~--~to the best of our knowledge~--~in the literature, we first need to establish the lower bound we will rely upon:

\begin{restatable}{theorem}{binomialtolerantlbtheorem}\label{theo:samp:binomial:tolerant:lb}
    There exists an absolute constant $\eps_0>0$ such that the following holds. Any algorithm which, given sampling access to an unknown distribution $\D$ on $\domain$ and parameter $\eps \in (0,\eps_0)$, distinguishes with probability at least $2/3$ between \textsf{(i)} $\normone{\D-\binomial{n}{1/2}} \leq \eps$ and \textsf{(ii)} $\normone{\D-\binomial{n}{1/2}} \geq 100\eps$ must use $\bigOmega{\frac{1}{\eps}\frac{\sqrt{n}}{\log n}}$ samples.
\end{restatable}
The proof relies on a reduction from tolerant testing of \emph{uniformity}, drawing on a result of Valiant and Valiant~\cite{ValiantValiant:10lb}; for the sake of conciseness, the details are deferred to~\autoref{app:binomial:tolerant:lb}. With~\autoref{theo:samp:binomial:tolerant:lb} in hand, we can apply~\autoref{theo:main:testing:tol:lb} to obtain the desired lower bound:
\corolbtolpbd*

\noindent We observe that both~\autoref{coro:tol:lb:nlogn} and~\autoref{coro:lb:tol:pbd} are tight (with regard to the dependence on $n$), as proven in the next section (\autoref{sec:toltesting:ub}).
 
\section{A Generic Tolerant Testing Upper Bound}\label{sec:toltesting:ub}
  \makeatletter{}To conclude this work, we address the question of tolerant testing of distribution classes. In the same spirit as before, we focus on describing a generic approach to obtain such bounds, in a clean conceptual manner. The most general statement of the result we prove in this section is stated below, which we then instantiate to match the lower bounds from~\autoref{sec:lowerbounds:tol}:

\begin{theorem}\label{theo:main:tol:testing:ub:almost}
Let \class be a class of distributions over $[n]$ for which the following holds: 
  \begin{enumerate}[(i)]
    \item there exists a semi-agnostic learner $\Learner$ for $\class$, with sample complexity $q_L(n,\eps, \delta)$ and ``agnostic constant''~$c$;
    \item for any $\eta\in[0,1]$, every distribution in \class has $\eta$-effective support of size at most $M(n,\eta)$.
  \end{enumerate}
Then, there exists an algorithm that, for any fixed $\kappa > 1$ and on input $\eps_1,\eps_2 \in (0,1)$ such that $\eps_2 \geq C \eps_1$, has the following guarantee (where $C > 2$ depends on $c$ and $\kappa$ only). The algorithm takes $\bigO{\frac{1}{(\eps_2-\eps_1)^2}\frac{m}{\log m}} + q_L(n,\frac{\eps_2-\eps_1}{\kappa}, \frac{1}{10})$ samples (where $m=M(n,\eps_1)$), and with probability at least $2/3$ distinguishes between \textsf{(a)} $\lp[1](\D,\class) \leq \eps_1$ and \textsf{(b)} $\lp[1](\D,\class) >~\eps_2$. (Moreover, one can take $C=(1+(5c+6)\frac{\kappa}{\kappa-1})$.)
\end{theorem}

\coromaintoltestingmlogm*

Applying now the theorem with $M(n,\eps)=\sqrt{n\log(1/\eps)}$ (as per~\autoref{coro:main:testing:effective:support:pbd}), we obtain an improved upper bound for Binomial and Poisson Binomial distributions: 
\coromaintoltestingpbd*

\paragraph*{High-level idea.} Somewhat similar to the lower bound framework developed in~\autoref{sec:lowerbounds}, the gist of the approach is to reduce the problem of tolerant testing membership of $\D$ to the \emph{class} \class to that of tolerant testing identity to a known \emph{distribution} -- namely, the distribution $\hat{\D}$ obtained after trying to agnostically learn $\D$. Intuitively, an agnostic learner for \class should result in a good enough hypothesis $\hat{\D}$ (i.e., $\hat{\D}$ close enough to both $\D$ and $\class$) when $\D$ is $\eps_1$-close to \class; but output a $\hat{\D}$ that is significantly far from either $\D$ or $\class$ when $\D$ is $\eps_2$-far from \class~--~sufficiently for us to be able to tell.
Besides the many technical details one has to control for the parameters to work out, one key element is the use of a tolerant testing algorithm for closeness of two distributions due to~\cite{VV:11:focs}, whose (tight) sample complexity scales as $n/\log n$ for a domain of size $n$. In order to get the right dependence on the effective support (required in particular for~\autoref{coro:main:tol:testing:pbd}), we have to perform a first test to identify the effective support of the distribution and check its size, in order to only call this tolerant closeness testing algorithm on this much smaller subset. (This additional preprocessing step itself has to be carefully done, and comes at the price of a slightly worse constant $C=C(c,\kappa)$ in the statement of the theorem.)

\subsection{Proof of~\autoref{theo:main:tol:testing:ub:almost}}

As described in the preceding section, the algorithm will rely on the ability to perform tolerant testing of equivalence between two unknown distributions (over some known domain of size $m$). This is ensured by an algorithm of Valiant and Valiant, restated below:
\begin{theorem}[{\cite[Theorem 3 and 4]{VV:11:focs}}]\label{theo:samp:closeness:tolerant}
    There exists an algorithm $\mathcal{E}$ which, given sampling access to two unknown distributions $\D_1,\D_2$ over $[m]$, satisfies the following. On input $\eps\in(0,1]$, it takes $O( \frac{1}{\eps^2}\frac{m}{\log m} )$ samples from $\D_1$ and $\D_2$, and outputs a value $\Delta$ such that $\dabs{\normone{\D_1-\D_2}-\Delta} \leq \eps$ with probability $1-1/\poly(m)$. (Furthermore, $\mathcal{E}$ runs in time $\poly(m)$.)
\end{theorem}

\noindent For the proof, we will also need this fact, similar to~\autoref{lemma:conditioned:class}, which relates the distance of two distributions to that of their conditional distributions on a subset of the domain:
\begin{fact}\label{lemma:conditioned:distr:distances}
Let $\D$ and $P$ be distributions over $[n]$, and $I\subseteq[n]$ an interval such that $\D(I) \geq 1- \alpha$ and $P(I) \geq 1- \beta$. Then,
\begin{itemize}
  \item $\normone{\D_I - P_I} \leq \frac{3}{2}\frac{\normone{\D - P}}{\D(I)} \leq 3\normone{\D - P}$ (the last inequality for $\alpha \leq \frac{1}{2}$); and 
  \item $\normone{\D_I - P_I} \geq \frac{}{}\normone{\D - P} - 2(\alpha+\beta)$.
\end{itemize}
\end{fact}

\iftrue
\begin{proof}
To establish the first item, write:
  \begin{align*}
      \normone{\D_I - P_I} &=\sum_{i\in I}\abs{ \frac{\D(i)}{\D(I)} - \frac{P(i)}{P(I)} }
      = \frac{1}{\D(I)}\sum_{i\in I}\abs{ \D(i) - P(i) + P(i)\big(1- \frac{\D(I)}{P(I)}\big) } \\
      &\leq \frac{1}{\D(I)}\big( \sum_{i\in I}\abs{ \D(i) - P(i) } + \abs{1- \frac{\D(I)}{P(I)} } \sum_{i\in I} P(i) \big)\\
      &= \frac{1}{\D(I)}\big( \sum_{i\in I}\abs{ \D(i) - P(i) } + \abs{P(I)- \D(I) } \big)
      \leq \frac{1}{\D(I)}\big( \sum_{i\in I}\abs{ \D(i) - P(i) } + \frac{1}{2}\normone{\D-P} \big)\\
     &\leq \frac{1}{\D(I)} \cdot \frac{3}{2}\normone{\D-P} 
  \end{align*}
  where we used the fact that $\abs{P(I)- \D(I) }\leq \totalvardist{\D}{P} = \frac{1}{2}\normone{\D-P}$.
 Turning now to the second item, we have:
  \begin{align*}
      \normone{\D_I - P_I} &= \frac{1}{\D(I)}\sum_{i\in I}\abs{ \D(i) - P(i) + P(i)\mleft(1- \frac{\D(I)}{P(I)}\mright) } 
      \geq \frac{1}{\D(I)}\big( \sum_{i\in I}\abs{ \D(i) - P(i) } - \abs{1- \frac{\D(I)}{P(I)} } \sum_{i\in I} P(i) \big)\\
      &= \frac{1}{\D(I)}\big( \sum_{i\in I}\abs{ \D(i) - P(i) } - \abs{ P(I)- \D(I) } \big)
      \geq \frac{1}{\D(I)}\big( \sum_{i\in I}\abs{ \D(i) - P(i) } - (\alpha+\beta) \big)\\
     &\geq \frac{1}{\D(I)}\big( \normone{\D-P} - \sum_{i\notin I} \abs{ \D(i) - P(i) } - (\alpha+\beta) \big)
     \geq \frac{1}{\D(I)}\big( \normone{\D-P} - 2(\alpha+\beta) \big) \\
     &\geq \normone{\D-P} - 2(\alpha+\beta).  
  \end{align*}
\end{proof}
\fi

\noindent With these two ingredients, we are in position to establish our theorem:
\begin{proofof}{\autoref{theo:main:tol:testing:ub:almost}}
The algorithm proceeds as follows, where we set $\eps\eqdef\frac{\eps_2-\eps_1}{17\kappa}$, $\theta\eqdef\eps_2 - ((6+c)\eps_1+11\eps)$, and $\tau \eqdef2 \frac{(3+c)\eps_1+5\eps}{2}$:
\begin{enumerate}[(1)]
  \item\label{algo:effectivesupport:toltest:step:0} using $O(\frac{1}{\eps^2})$ samples, get (with probability at least $1-1/10$, by~\autoref{theo:dkw:ineq}) a distribution $\tilde{\D}$ $\frac{\eps}{2}$-close to $\D$ in Kolmogorov distance; and let $I\subseteq[n]$ be the smallest interval such that $\tilde{\D}(I) > 1-\frac{3}{2}\eps_1-\eps$. Output \reject if $\abs{I} > M(n,\eps_1)$.
  \item\label{algo:effectivesupport:toltest:step:1} invoke $\Learner$ on $\D$ with parameters $\eps$ and failure probability $\frac{1}{10}$, to obtain a hypothesis $\hat{\D}$;
  \item\label{algo:effectivesupport:toltest:step:2}  call $\mathcal{E}$ (from~\autoref{theo:samp:closeness:tolerant}) on $\D_I$, $\hat{\D}_I$ with parameter $\frac{\eps}{6}$ to get an estimate $\hat{\Delta}$ of $\normone{\D_I-\hat{\D}_I}$;
  \item\label{algo:effectivesupport:toltest:step:2.5} output \reject if $\hat{\D}(I) < 1-\tau$; 
  \item\label{algo:effectivesupport:toltest:step:3}  compute ``offline'' (an estimate accurate within $\eps$ of) $\lp[1](\hat{\D},\class)$, denoted $\Delta$;
  \item\label{algo:effectivesupport:toltest:step:4}  output \reject is $\Delta+\hat{\Delta} > \theta$, and output \accept otherwise.
\end{enumerate}
The claimed sample complexity is immediate from Steps~\ref{algo:effectivesupport:toltest:step:1} and~\ref{algo:effectivesupport:toltest:step:2}, along with~\autoref{theo:samp:closeness:tolerant}. Turning to correctness, we condition on both subroutines meeting their guarantee (i.e., $\normone{\D-\hat{\D}} \leq c\cdot\opt+\eps$ and $\normone{\D-\hat{\D}} \in [\hat{\Delta} - \eps,\hat{\Delta} + \eps]$), which happens with probability at least $8/10-1/\poly(n)\geq 3/4$ by a union bound. 
\begin{itemize}
  \item Soundness: If $\lp[1](\D,\class) \leq \eps_1$, then $\D$ is $\eps_1$-close to some $P\in\class$, for which there exists an interval $J\subseteq[n]$ of size at most $M(n,\eps_1)$ such that $P(J) \geq 1-\eps_1$. It follows that $\D(J) \geq 1-\frac{3}{2}\eps_1$ (since $\abs{\D(J)-P(J)} \leq \frac{\eps_1}{2}$) and $\tilde{\D}(J) \geq 1-\frac{3}{2}\eps_1-2\cdot\frac{\eps}{2}\eps$; establishing existence of a good interval $I$ to be found (and Step~\ref{algo:effectivesupport:toltest:step:0} does not end with \reject). Additionally, $\normone{\D-\hat{\D}} \leq c\cdot\eps_1+\eps$ and by the triangle inequality this implies $\lp[1](\hat{\D},\class) \leq (1+c)\eps_1+\eps$.
  
  Moreover, as $\D(I) \geq \tilde{\D}(I) - 2\cdot\frac{\eps}{2} \geq 1-\frac{3}{2}\eps_1-2\eps$ and $\abs{\hat{\D}(I) - \D(I)} \leq \frac{1}{2}\normone{\D-\hat{\D}}$, we do have
  \[
    \hat{\D}(I) \geq 1-\frac{3}{2}\eps_1-2\eps - \frac{c\eps_1}{2}-\frac{\eps}{2} = 1-\tau
  \]
  and the algorithm does not reject in Step~\ref{algo:effectivesupport:toltest:step:2.5}. 
  To conclude, one has by~\autoref{lemma:conditioned:distr:distances} that
  \[
      \normone{\D_I-\hat{\D}_I} \leq \frac{3}{2}\frac{\normone{\D-\hat{\D}}}{\D(I)} \leq \frac{3}{2}\frac{(c\eps_1+\eps)}{1-\frac{3}{2}\eps_1-2\eps} \leq 3(c\eps_1+\eps) \tag{for $\eps_1 < 1/4$, as $\eps < 1/17$}
  \]
  Therefore, $\Delta+\hat{\Delta} \leq \lp[1](\hat{\D},\class)+\eps + \normone{\D_I-\hat{\D}_I}+\eps \leq (4c+1)\eps_1+6\eps \leq \eps_2 - ((6+c)\eps_1+11\eps) = \theta$ (the last inequality by the assumption on $\eps_2,\eps_1$), and the tester accepts.
  \item Completeness: If $\lp[1](\D,\class) > \eps_2$, then we must have $\normone{\D-\hat{\D}} + \lp[1](\hat{\D},\class) > \eps_2$. If the algorithm does not already reject in Step~\ref{algo:effectivesupport:toltest:step:2.5}, then $\hat{\D}(I) \geq 1-\tau$. But, by~\autoref{lemma:conditioned:distr:distances},
  \begin{align*}
      \normone{\D_I-\hat{\D}_I} &\geq \normone{ \D - \hat{\D} } - 2(\D(I^c) + \hat{\D}(I^c)) \geq \normone{\D_I-\hat{\D}_I} - 2\Big(\frac{3}{2}\eps_1 + 2\eps+\tau\Big) \\
      &= \normone{ \D - \hat{\D} } - ((6+c)\eps_1+9\eps) 
  \end{align*}
  we then have $\normone{\D_I-\hat{\D}_I} + \lp[1](\hat{\D},\class) > \eps_2  - ((6+c)\eps_1+9\eps)$. 
  This implies $\Delta+\hat{\Delta} > \eps_2 - ((6+c)\eps_1+9\eps) -2\eps = \eps_2 - ((6+c)\eps_1+11\eps)  = \theta$, and the tester rejects.
\end{itemize}
Finally, the testing algorithm defined above is computationally efficient as long as both the learning algorithm (Step~\ref{algo:effectivesupport:toltest:step:1}) and the estimation procedure (Step~\ref{algo:effectivesupport:toltest:step:3}) are.

\end{proofof}

\addcontentsline{toc}{section}{References}
\bibliographystyle{alpha}
\bibliography{testing-from-structural} 
\appendix
\makeatletter{}\section{Proof of \autoref{lemma:estimate:l2:add}}\label{app:l2:proof}

We now give the proof of~\autoref{lemma:estimate:l2:add}, restated below:
\lemmaestimateltwoadd*
\begin{proof}
To do so, we first describe an algorithm that distinguishes between $\normtwo{\D-\uniform}^2 \geq \eps^2/{n}$ and $\normtwo{\D-\uniform}^2 < \eps^2/(2n)$ with probability at least $2/3$, using $C\cdot\frac{\sqrt{n}}{\eps^2}$ samples. Boosting the success probability to $1-\delta$ at the price of a multiplicative $\log\frac{1}{\delta}$ factor can then be achieved by standard techniques.

Similarly as in the proof of Theorem 11 (whose algorithm we use, but with a threshold $\tau\eqdef \frac{3}{4}\frac{m^2\eps^2}{n}$ instead of $\frac{4m}{\sqrt{n}}$), define the quantities
\[
Z_k \eqdef \left(X_k-\frac{m}{n}\right)^2 - X_k,\qquad k\in[n]
\]
and $Z\eqdef\sum_{k=1}^n Z_k$, where the $X_k$'s (and thus the $Z_k$'s) are independent by Poissonization, and $X_k\sim\poisson{m \D(k)}$. It is not hard to see that $\shortexpect Z_k = \Delta_k^2$, where $\Delta_k\eqdef (\frac{1}{n}-\D(k))$, so that $\shortexpect Z = m^2\normtwo{\D-\uniform}^2$. Furthermore, we also get
\[
  \var Z_k = 2m^2\left(\frac{1}{n}-\Delta_k\right)^2 + 4m^3\left(\frac{1}{n}-\Delta_k\right)\Delta_k
\]
so that
\begin{equation}\label{eq:estimate:l2:add:variance}
  \var Z = 2m^2\left( \sum_{k=1}^n \Delta_k^2 + \frac{1}{n} -2 m\sum_{k=1}^n \Delta_k^3\right)
\end{equation}
(after expanding and since $\sum_{k=1}^n \Delta_k = 0$).

\paragraph{Soundness.} \textit{\emph{Almost} straight from \cite{DKN:15}, but the threshold has changed.}
Assume $\Delta^2\eqdef\normtwo{\D-\uniform}^2 \geq \eps^2/n$; we will show that $\probaOf{ Z < \tau } \leq 1/3$. By Chebyshev's inequality, it is sufficient to show that
$\tau \leq \shortexpect Z - \sqrt{3}\sqrt{\var Z}$, as \[
    \probaOf{ \shortexpect Z - Z > \sqrt{3}\sqrt{\var Z} } \leq 1/3\;.
\]
As $\tau < \frac{3}{4}\shortexpect Z$, arguing that $\sqrt{3}\sqrt{\var Z} \leq \frac{1}{4}\shortexpect Z$ is enough, i.e. that $48\var Z \leq (\shortexpect Z)^2$. From \eqref{eq:estimate:l2:add:variance}, this is equivalent to showing
\[
    \Delta^2 + \frac{1}{n} -2 m\sum_{k=1}^n \Delta_k^3 \leq \frac{m^2\Delta^4}{96}\;.
\]
We bound the LHS term by term.
\begin{itemize}
  \item As $\Delta^2 \geq \frac{\eps^2}{n}$, we get $m^2\Delta^2 \geq \frac{C^2}{\eps^2}$, and thus $\frac{m^2\Delta^4}{288} \geq \frac{C^2}{288\eps^2}\Delta^2 \geq \Delta^2$ (as $C \geq 17$ and $\eps \leq 1$).
  \item Similarly, $\frac{m^2\Delta^4}{288} \geq \frac{C^2}{288\eps^2}\cdot \frac{\eps^2}{n} \geq \frac{1}{n}$. 
  \item Finally, recalling that\footnotemark
  \[
    \sum_{k=1}^n \abs{\Delta_k}^3 \leq \left( \sum_{k=1}^n \abs{\Delta_k}^2 \right)^{3/2} = \Delta^3
  \]
  we get that $\abs{2m\sum_{k=1}^n \abs{\Delta_k}^3} \leq 2m \Delta^3 = \frac{m^2 \Delta^4}{288} \cdot \frac{2\cdot 288}{m\Delta} \leq \frac{m^2 \Delta^4}{288}$, using the fact that
  $\frac{m\Delta}{2\cdot 288} \geq \frac{C}{576\eps} \geq 1$ (by choice of $C \geq 576$).
\end{itemize}
Overall, the LHS is at most $3\cdot \frac{m^2 \Delta^4}{288} = \frac{m^2 \Delta^4}{96}$, as claimed.

\footnotetext{
For any sequence $x=(x_1,\dots,x_n)\in\R^n$, $p > 0 \mapsto \norm{x}_p$ is non-increasing. In particular, for $0 < p \leq q <\infty$,
\[
\left(\sum_i \abs{x_i}^q\right)^{1/q} = \norm{x}_q \leq \norm{x}_p = \left(\sum_i \abs{x_i}^p\right)^{1/p}\;.
\] To see why, one can easily prove that if $\norm{x}_p = 1$, then $\norm{x}_q^q \leq 1$ (bounding each term $\abs{x_i}^q \leq \abs{x_i}^p$), and therefore $\norm{x}_q \leq 1 = \norm{x}_p$. Next, for the general case, apply this to $y = x/\norm{x}_p$, which has unit $\lp[p]$ norm, and conclude by homogeneity of the norm.
}

\paragraph{Completeness.} Assume $\Delta^2=\normtwo{\D-\uniform}^2 < \eps^2/(4n)$. We need to show that $\probaOf{ Z \geq \tau } \leq 1/3$. Chebyshev's inequality implies
\[
    \probaOf{ Z - \shortexpect Z > \sqrt{3}\sqrt{\var Z} } \leq 1/3
\]
and therefore it is sufficient to show that 
\[
    \tau \geq \shortexpect Z + \sqrt{3}\sqrt{\var Z}
\]
Recalling the expressions of $\shortexpect Z$ and $\var Z$ from \eqref{eq:estimate:l2:add:variance}, this is tantamount to showing
\[
    \frac{3}{4}\frac{m^2\eps^2}{n} \geq m^2\Delta^2 + \sqrt{6}m\sqrt{\Delta^2 + \frac{1}{n} -2m \sum_{k=1}^n \Delta_k^3}
\]
or equivalently
\[
    \frac{3}{4} \frac{m}{\sqrt{n}}\eps^2 \geq m \sqrt{n} \Delta^2 + \sqrt{6} \sqrt{1 + n\Delta^2 -2nm \sum_{k=1}^n \Delta_k^3}\;.
\]
Since $\sqrt{1 + n\Delta^2 -2 n m  \sum_{k=1}^n \Delta_k^3} \leq \sqrt{1 + n\Delta^2} \leq \sqrt{1 + \eps^2/4} \leq \sqrt{5/4}$, we get that the second term is at most $\sqrt{30/4} < 3$. All that remains is to show that $m\sqrt{n}\Delta^2 \geq 3m\frac{\eps^2}{4\sqrt{n}}-3$. But as $\Delta^2 < \eps^2/(4n)$, $m\sqrt{n}\Delta^2 \leq m\frac{\eps^2}{4\sqrt{n}}$; and our choice of $m \geq C\cdot\frac{\sqrt{n}}{\eps^2}$ for some absolute constant $C \geq 6$  ensures this holds.
\end{proof}
 
\makeatletter{}\section{Proof of~\autoref{theo:structural:mhr}}\label{app:structural:proofs}

In this section, we prove our structural result for MHR distributions, ~\autoref{theo:structural:mhr}:
\theostructuralmhr*
\begin{proof} We reproduce and adapt the argument of~\cite[Section 5.1]{CDSS:13} to meet our definition of decomposability (which, albeit related, is incomparable to theirs). First, we modify the algorithm at the core of their constructive proof, in~\autoref{algo:mhr}: note that the only two changes are in Steps~\ref{algo:mhr:step:iiprime} and \ref{algo:mhr:step:iprimeprime}, where we use parameters respectively $\frac{\gamma}{n}$ and $\frac{\gamma}{n^2}$.
\begin{algorithm}
  \begin{algorithmic}[1]
    \Require explicit description of MHR distribution $\D$ over $[n]$; accuracy parameter $\gamma > 0$
    \State Set $J\gets [n]$ and $\mathcal{Q}\gets \emptyset$.
    \State Let $I\gets\textsc{Right-Interval}(\D,J,\frac{\gamma}{n})$ and $I^\prime\gets\textsc{Right-Interval}(\D,J\setminus I,\frac{\gamma}{n})$. Set $J\gets J\setminus (I\cup I^\prime)$. \label{algo:mhr:step:iiprime}
    \State Set $i\in J$ to be the smallest integer such that $\D(i)\geq \frac{\gamma}{n^2}$. 
            If no such $i$ exists, let $I^{\prime\prime}\gets J$ and go to Step~\ref{algo:mhr:laststep}. 
            Otherwise, let $I^{\prime\prime}\gets \{1,\dots,i-1\}$ and $J\gets J\setminus I^{\prime\prime}$.
            \label{algo:mhr:step:iprimeprime}
    \While{ $J\neq\emptyset$ }\label{algo:mhr:whileloop}
      \State Let $j\in J$ bet the smallest integer such that $\D(j) \notin [\frac{1}{1+\gamma}, 1+\gamma]\D(i)$.
            If no such $j$ exists, let $I^{\prime\prime\prime}\gets J$; otherwise let $I^{\prime\prime\prime}\gets \{i,\dots, j-1\}$.
      \State Add $I^{\prime\prime\prime}$ to $\mathcal{Q}$ and set $J\gets J\setminus I^{\prime\prime\prime}$.
      \State Let $i\gets j$.
    \EndWhile
    \State Return $\mathcal{Q}\cup\{I,I^{\prime},I^{\prime\prime}\}$\label{algo:mhr:laststep}
  \end{algorithmic}
  \caption{\label{algo:mhr} $\textsc{Decompose-MHR}^\prime(\D,\gamma)$}.
\end{algorithm}
Following the structure of their proof, we write $\mathcal{Q}=\{I_1,\dots,I_{\abs{\mathcal{Q}}}\}$ with $I_i=[a_i,b_i]$, and define $\mathcal{Q}^\prime=\setOfSuchThat{ I_i\in \mathcal{Q} }{ \D(a_i) > \D(a_{i+1})}$, $\mathcal{Q}^{\prime\prime}=\setOfSuchThat{ I_i\in \mathcal{Q} }{ \D(a_i) \leq \D(a_{i+1})}$.

\noindent We immediately obtain the analogues of their Lemmas~5.2 and~5.3:
\begin{lemma}\label{lemma:lemma52:cdss:13}
We have $\prod_{I_i\in\mathcal{Q}^\prime} \frac{\D(a_i)}{\D(a_{i+1})} \leq \frac{n}{\gamma}$.
\end{lemma}
\begin{lemma}
Step~\ref{algo:mhr:whileloop} of~\autoref{algo:mhr} adds at most $\bigO{\frac{1}{\eps}\log\frac{n}{\eps}}$ intervals to $\mathcal{Q}$.
\end{lemma}
\begin{proof}[Sketch]
This derives from observing that now $\D(I\cup I^\prime) \geq \gamma/n$, which as in~\cite[Lemma~5.3]{CDSS:13} in turn implies 
\[
  1 \geq \frac{\gamma}{n}(1+\gamma)^{\abs{\mathcal{Q}^{\prime}}-1}
\]
so that $\abs{\mathcal{Q}^{\prime}} = \bigO{\frac{1}{\eps}\log\frac{n}{\eps}}$.

Again following their argument, we also get 
\[
  \frac{\D(a_{\abs{\mathcal{Q}}+1})}{\D(a_1)} = \prod_{I_i\in\mathcal{Q}^{\prime\prime}} \frac{\D(a_{i+1})}{\D(a_i)}\cdot \prod_{I_i\in\mathcal{Q}^\prime} \frac{\D(a_{i+1})}{\D(a_i)}
\]
by combining~\autoref{lemma:lemma52:cdss:13} with the fact that $\D(a_{\abs{\mathcal{Q}}+1} \leq 1$ and that by construction $\D(a_i) \geq \gamma/n^2$, we get 
\[
    \prod_{I_i\in\mathcal{Q}^{\prime\prime}} \frac{\D(a_{i+1})}{\D(a_i)} \leq \frac{n}{\gamma} \cdot \frac{n^2}{\gamma} = \frac{n^3}{\gamma}\ .
\] 
But since each term in the product is at least $(1+\gamma)$ (by construction of $\mathcal{Q}$ and the definition of $\mathcal{Q}^{\prime\prime}$), this leads to
\[
  (1+\gamma)^{\abs{ \mathcal{Q}^{\prime\prime} }} \leq \frac{n^3}{\gamma}
\]
and thus $\abs{ \mathcal{Q}^{\prime\prime} } = \bigO{\frac{1}{\eps}\log\frac{n}{\eps}}$ as well.
\end{proof}

It remains to show that $\mathcal{Q}\cup\{I,I^{\prime},I^{\prime\prime}\}$ is indeed a good decomposition of $[n]$ for $\D$, as per~\autoref{def:struct:dec:split}. Since by construction every interval in $\mathcal{Q}$ satisfies~\autoref{def:struct:item:flat}, we only are left with the case of $I$, $I^{\prime}$ and $I^{\prime\prime}$. For the first two, as they were returned by $\textsc{Right-Interval}$ either \textsf{(a)} they are singletons, in which case~\autoref{def:struct:item:flat} trivially holds; or \textsf{(b)} they have at least two elements, in which case they have probability mass at most $\frac{\gamma}{n}$ (by the choice of parameters for $\textsc{Right-Interval}$) and thus~\autoref{def:struct:item:light} is satisfied. Finally, it is immediate to see that by construction $\D(I^{\prime\prime}) \leq n\cdot \gamma/n^2 = \gamma/n$, and~\autoref{def:struct:item:light} holds in this case as well.
\end{proof}
 
\makeatletter{}\section{Proofs from~\autoref{sec:structural}}\label{app:structural:projection:proofs}

This section contains the proofs omitted from~\autoref{sec:structural}, namely the distance estimation procedures for $t$-piecewise degree-$d$ (\autoref{theo:degreed:poly-projection}), monotone hazard rate (\autoref{lemma:distance:mhr:eff}), and log-concave distributions (\autoref{lemma:distance:log:eff}).

\subsection{Proof of~\autoref{theo:degreed:poly-projection}}\label{app:structural:projection:proofs:degreed}

In this section, we prove the following:
\begin{theorem} \label{thm:project:single:poly}
Let $p$ be a $\ell$-histogram over $[-1,1)$. There is an algorithm $\textsc{ProjectSinglePoly}(d,\eps)$
which runs in time $\poly(\ell, d+1,1/\eps)$, and outputs a degree-$d$ polynomial $q$ which defines a pdf over $[-1,1) $
such that $\normone{p-q} \leq 3 \lp[1](p,\classpoly[d]) + O(\eps)$.
\end{theorem}
As mentioned in~\autoref{sec:structural}, the proof of this statement is a rather straightforward adaption of the proof of~\cite[Theorem 9]{CDSS:14}, with two differences: first, in our setting there is no uncertainty nor probabilistic argument due to sampling, as we are provided with an explicit description of the histogram $p$. Second, Chan et al. require some ``well-behavedness'' assumption on the distribution $p$ (for technical reasons essentially due to the sampling access), that we remove here. Besides these two points, the proof is almost identical to theirs, and we only reproduce (our modification of) it here for the sake of completeness.
(Any error introduced in the process, however, is solely our responsibility.)
\begin{proof}
\newcommand{\rr}{{r}} \newcommand{\pp}{{p}} Some preliminary definitions will be helpful:
\begin{definition}[Uniform partition]
  Let $p$ be a subdistribution on an interval $I \subseteq [-1,1)$.
  A partition $\mathcal{I} = \{I_1, \dots, I_\ell\}$ of $I$ is
  \emph{$(p,\eta)$-uniform} if $p(I_j) \leq \eta$ for all $1\leq j\leq \ell$.
\end{definition}

We will also use the following notation:  For this subsection, let $I = {[-1,1)}$ ({$I$ will denote a subinterval of $[-1,1)$ when the results
are applied in the next subsection}).
We write $\|f\|^{(I)}_{1}$ to denote $\int_{I} |f(x)| dx$,
and we write $\dtv^{(I)}(p,q)$ to denote $\normone{p-q}^{(I)}/2$.
We write $\opt^{(I)}_{1,d}$ to denote the {infimum of the} 
distance $\normone{p-g}^{(I)}$ between $p$ and any degree-$d$
subdistribution $g$ on $I$ that satisfies $g(I) = p(I)$.

The key step of {\sc ProjectSinglePoly} is 
Step~\ref{step:projectsinglepoly:find} where it calls the {\sc FindSinglePoly} procedure.
In this procedure
$T_i(x)$ denotes the degree-$i$ Chebychev polynomial
of the first kind.
The function {\sc FindSinglePoly} should be thought of
as the CDF of a ``quasi-distribution'' $f$; we say that
$f=F'$ is a ``quasi-distribution'' 
and not a \textit{bona fide} probability distribution because it is
not guaranteed to be non-negative everywhere on $[-1,1)$.  Step~\ref{step:findsinglepoly:2}
of {\sc FindSinglePoly} processes
$f$ slightly to obtain a polynomial $q$ which is an actual distribution over
$[-1,1).$
\begin{algorithm}
  \caption{{\sc ProjectSinglePoly}}
  \begin{algorithmic}[1]
    \Require parameters $d,\eps$; and the full description of a $\ell$-histogram $p$ over $[-1,1)$.
    \Ensure a degree-$d$ distribution $q$ such that $\dtv(p,q) \leq 3 \cdot \opt_{1,d} + O(\eps)$
    \State Partition $[-1,1)$ into $z = \Theta((d+1)/\eps)$ intervals $I_0 = [i_0,i_1), \dots, I_{z-1}=[i_{z-1},i_z)$, where $i_0=-1$ and $i_z=1$, such that
  for each $j \in \{1,\dots,z\}$ we have $p(I_j) = \Theta(\eps/(d+1))$ or ($\abs{I_j}=1$ and $p(I_j)=\bigOmega{\eps/(d+1)}$).
    \State\label{step:projectsinglepoly:find} Call {\sc FindSinglePoly}($d$, $\eps$, $\eta:=\Theta(\eps/(d+1))$, $\{I_0,\dots,I_{z-1}\}$, $p$ and output the hypothesis $q$ that it returns.
  \end{algorithmic}
\end{algorithm}

\begin{algorithm}
  \caption{{\sc FindSinglePoly}}
  \begin{algorithmic}[1]
    \Require  degree parameter $d$;
error parameter $\eps$; parameter $\eta$; $(p,\eta)$-uniform partition
$\mathcal{I}_I = \{I_1, \dots, I_{{z}}\}$ of interval $I$ into ${{z}}$ intervals {such that $\sqrt{
\eps z}\cdot \eta \leq \eps/2$}; a subdistribution $p$ on $I$

    \Ensure a number $\tau$ and a degree-$d$
subdistribution $q$ on $I$ such that $q(I) = p(I)$,
\[ 
  \opt^{(I)}_{1,d} \leq \normone{p-q}^{(I)} \leq 3\opt^{(I)}_{1,d} + \sqrt{\eps z (d+1)}
  \cdot \eta + {\rm error},
\]
$0 \leq \tau \leq \opt^{(I)}_{1,d}$ {and ${\rm error} = O({(d+1)}\eta)$}.

\State\label{step:findsinglepoly:1} Let $\tau$ be the solution to the following LP:
\[
\text{minimize~}\tau~\text{subject to the following constraints:}
\]
(Below 
$F(x) = \sum_{i=0}^{d+1} c_i T_i(x)$ where $T_i(x)$ is the degree-$i$
Chebychev polynomial of the first kind, and $f(x)=F'(x) =
\sum_{i=0}^{d+1} c_i T'_i(x)$.)

\begin{enumerate}[(a)]

\item \label{item:total} $F(-1)=0$ and $F(1)=p(I)$;

\item \label{item:phat} For each $0 \leq j < k \leq z$,
\begin{equation} \label{eq:agno-phat}
  \abs{ \left(p([i_j,i_k)) + \sum_{j\leq \ell < k} w_\ell \right) - (F(i_k) - F(i_j)) } \leq \sqrt{\eps \cdot (k-j)} \cdot \eta;
\end{equation}

\item \label{item:robust}
\begin{align}
  \sum_{0\leq \ell < {z}} w_\ell &= 0, \label{item:robust:zerosum} \\
  -y_\ell \leq w_\ell &\leq y_\ell \qquad \text{for all $0\leq \ell < {z}$,} \label{item:robust:absval} \\
  \sum_{0\leq \ell < {z}} y_\ell &\leq \tau \label{item:robust:lb};
\end{align}

\item \label{item:AK}  The constraints $|c_i| \leq \sqrt{2}$ for $i=0,\dots,d+1$;

\item \label{item:AK2} The constraints
\[
0 \leq F(z)  \leq 1 \quad \text{for all~} z \in J,
\]
where $J$ is a set of {$\bigO{(d+1)^6}$} equally spaced points across $[-1,1)$;

\item \label{item:nonneg-1} The constraints
\[
\sum_{i=0}^d c_i T'_i(x) \geq 0 \quad \text{for all~}x \in K,
\]
where $K$ is a set of $O((d+1)^2/\eps)$ equally spaced points across
$[-1,1)$.

\end{enumerate}

\State\label{step:findsinglepoly:2} Define
$q(x) = { \eps f(I)/\abs{I} + (1-\eps)f(x)}.$
Output $q$ as the hypothesis pdf.

\end{algorithmic}
\end{algorithm}

The rest of this subsection gives the proof of Theorem~\ref{thm:project:single:poly}.
The claimed running time bound is obvious
 (the computation is dominated by
solving the $\poly(d,1/\eps)$-size LP in {\sc ProjectSinglePoly}, with an additional term linear in $\ell$ when partitioning $[-1,1)$ in the initial first step),
so it suffices to prove correctness.

Before launching into the proof we give some intuition for the linear
program.
Intuitively $F(x)$ represents the cdf of a degree-$d$ polynomial
distribution $f$ where $f=F'.$  Constraint~\ref{item:total} captures the endpoint
constraints that any cdf must obey {if it has the same total weight as $p$}.
Intuitively, constraint~\ref{item:phat} ensures that for each interval $[i_j,i_k)$,
the value $F(i_k)-F(i_j)$ (which we may alternately write as
$f([i_j,i_k))$) is close to the weight 
$p([i_j,i_k))$ that the distribution 
puts on the interval.  
Recall that by assumption
$p$ is $\opt_{1,d}$-close to some degree-$d$ polynomial $r$.
Intuitively the variable $w_\ell$ represents $\int_{[i_\ell, i_{\ell+1})}
(r-p)$ (note that these values sum to zero by
constraint~\ref{item:robust}\eqref{item:robust:zerosum}, and $y_\ell$ represents the absolute value of $w_\ell$
(see constraint~\ref{item:robust}\eqref{item:robust:absval}).
The value $\tau$, which by constraint~\ref{item:robust}\eqref{item:robust:lb} is at least the 
sum of the $y_\ell$'s, represents a lower bound on 
$\opt_{1,d}.$
The constraints in~\ref{item:AK} and~\ref{item:AK2} reflect the fact that
as a cdf, $F$ should be bounded between 0 and 1 (more on this below),
and the~\ref{item:nonneg-1} constraints reflect the fact that the pdf $f=F'$ should be
everywhere nonnegative (again more on this below).

\medskip

We begin by observing that 
{\sc ProjectSinglePoly} calls {\sc FindSinglePoly} with input parameters that satisfy
{\sc FindSinglePoly}'s input requirements:

\begin{enumerate}
\item [(I)] the non-singleton intervals $I_0,\dots,I_{z-1}$ are $(p,\eta)$-uniform; and
\item [(II)] the singleton intervals each have weight at least $\frac{\eta}{10}$.
\end{enumerate}

We then proceed to show that, from there, {\sc FindSinglePoly}'s LP is feasible and has a high-quality
optimal solution.

\begin{lemma} \label{lem:feasible}
Suppose $p$ is an $\ell$-histogram over $[-1,1)$, so that conditions (I) and (II) above hold;
then the LP defined in Step~\ref{step:findsinglepoly:1} of {\sc FindSinglePoly}
is feasible; and the optimal solution $\tau$ is at most $\opt_{1,d}$.
\end{lemma}

\begin{proof}
As above, let $r$ be a degree-$d$ polynomial pdf such that $\opt_{1,d}=
\normone{p-r}$ {and $r(I) = p(I)$}.We exhibit a feasible solution as follows:
take $F$ to be the cdf of {$\rr$} (a degree $d$ polynomial).
Take $w_\ell$ to be $\int_{[i_\ell,i_{\ell+1})} ({\rr-\pp})$,
and take $y_\ell$ to be $\abs{w_\ell}$.
Finally, take $\tau$ to be $\sum_{0 \leq \ell < {z}} y_\ell.$

We first argue feasibility of the above solution.  
We first take care of the easy constraints:
since $F$ is the cdf of a {sub}distribution over $I$ it is clear that
constraints~\ref{item:total} and~\ref{item:AK2} are satisfied,
and since both $r$ and $p$ are pdfs {with the same total weight} it is clear
that constraints~\ref{item:robust}\eqref{item:robust:zerosum} and~\ref{item:nonneg-1} are both satisfied. 
Constraints~\ref{item:robust}\eqref{item:robust:absval} and~\ref{item:robust}\eqref{item:robust:lb} also hold.  
So it remains to argue constraints~\ref{item:phat} and~\ref{item:AK}.

Note that constraint~\ref{item:phat} is equivalent to $p + (\rr - p) = \rr$ 
and $\rr$ satisfying $(\mathcal{I}, \eps/(d+1), \eps)$-inequalities, 
therefore this constraint is satisfied.

To see that constraint~\ref{item:AK} is satisfied we recall some of the analysis
of Arora and Khot~\cite[{Section~3}]{AK:03}.  This analysis shows that since
$F$ is a cumulative distribution function (and in particular a function bounded between 0 and 1 on $I$) each of its
Chebychev coefficients is at most $\sqrt{2}$ in magnitude.

To conclude the proof of the lemma we need to argue that 
$\tau \leq \opt_{1,d}$.
Since $w_\ell = \int_{[i_\ell,i_{\ell+1})} ({\rr-\pp})$ it
is easy to see that $\tau = \sum_{0 \leq \ell < {z}} y_\ell = \sum_{0
\leq \ell < {z}} |w_\ell| \leq \normone{\pp-\rr}$, and
hence indeed $\tau \leq \opt_{1,d}$ as required.
\end{proof}

Having established that with high probability the LP is indeed feasible,
henceforth we let $\tau$ denote the optimal solution to the LP and
$F$, $f$, $w_\ell$, $c_i$, $y_\ell$ denote the values in the optimal solution.
A simple argument (see e.g. the proof of {\cite[Theorem~8]{AK:03}}) gives that $\norminf{F}\leq 2$.
Given this bound on $\norminf{F}$, the Bernstein--Markov inequality implies that $\norminf{f} = \norminf{F^\prime}\leq O((d+1)^2)$.
Together with \ref{item:nonneg-1} this implies that
$f(z) \geq -\eps/2$ for all $z \in [-1,1)$.
Consequently $q(z) \geq 0$ for all $z \in [-1,1)$,
and
\[
\int_{-1}^1 q(x) dx = \eps + (1 - \eps) \int_{-1}^1 f(x)dx = \eps +
(1-\eps)(F(1)-F(-1)) = 1. \]
So $q(x)$ is indeed a degree-$d$ pdf.  To prove~\autoref{thm:project:single:poly} it remains to show that $\normone{p-q} \leq 3 \opt_{1,d} + O(\eps).$

We sketch the argument that we shall use to bound $\normone{p-q}$.
A key step in achieving this bound is to 
bound the $\norm{\cdot}_{\cal A}$ distance between $f$ and
$\widehat{p}_m + w$ where ${\cal A} = {\mathcal A_{d+1}}$ is the class of
all unions of $d+1$ intervals and $w$ is a function based on the $w_\ell$
values (see \eqref{eq:good} below).
\new{If we can bound $\norm{(p+w)- f}_{\cal A} \leq O(\eps)$ then it will not be difficult to show that
$\norm{r - f}_{\cal A} \leq \opt_{1,d} + O(\eps)$.}. 
Since $r$ and $f$ are both degree-$d$ polynomials we have 
$\normone{r - f} = 2\norm{r - f}_{\cal A}  \leq 2 \opt_{1,d} + O(\eps)$, 
so the triangle inequality (recalling that
$\normone{p-r} = \opt_{1,d}$) gives
$\normone{p-f} \leq 3 \opt_{1,d}+O(\eps).$
 From this point a simple argument 
(Proposition~\ref{prop:epsmixture}) gives that
$\normone{p-q} \leq \normone{p-f} + O(\eps)$, which gives the theorem.

We will use the following lemma {that translates $(\mathcal{I}, \eta,\eps)$-inequalities into a bound on $\mathcal A_{d+1}$ distance}.

\begin{lemma} \label{lem:ad-dist}
Let $\mathcal{I} = \{I_0=[i_0, i_1), \dots, I_{z-1}=[i_{z-1}, i_z)\}$ be a
$(p,\eta)$-uniform partition of $I$, \new{possibly augmented with singleton intervals}.
If $h\colon I\to \R$ and \new{$p$} satisfy the $(\mathcal{I}, \eta,
\eps)$-inequalities, then
\[ {\norm{p-h}_{\mathcal A_{{d+1}}}^{(I)} \leq \sqrt{\eps z {(d+1)}}\cdot \eta + {\rm error},} \]
{where ${\rm error} = O({(d+1)}\eta)$}.
\end{lemma}

\begin{proof}
To analyze
$\norm{p-h}_{\mathcal A_{d+1}}$,
consider any union of ${d+1}$ 
disjoint non-overlapping intervals $S = J_1 \cup\dots \cup J_{d+1}$.
We will bound $\norm{ p - h }_{\mathcal A_{d+1}}$ 
by bounding $\abs{ p(S) - h(S)}$.

We lengthen intervals in $S$ slightly to obtain $T = J'_1 \cup \dots \cup J'_{{d+1}}$ 
so that each $J'_j$ is a union of intervals of the form $[i_\ell,i_{\ell+1})$.
Formally, if $J_j = [a,b)$, then $J'_j = [a',b')$, where $a' = \max_\ell \setOfSuchThat{ i_\ell }{ i_\ell \leq a }$ and $b' = \min_\ell \setOfSuchThat{ i_\ell }{ i_\ell \geq b }$.
We claim that
\begin{equation} \label{eq:lengthen}
  \abs{ p(S) - h(S) } \leq O({(d+1)}\eta) + \abs{ p(T) - f(T) } .
\end{equation}
Indeed, consider any interval of the form $J = [i_\ell, i_{\ell+1})$ 
such that $J \cap S \neq J \cap T$ \new{(in particular, such an interval cannot be one of the singletons)}.  We have
\begin{equation} \label{eq:lengthen-single}
\abs{ p(J \cap S) - p(J \cap T) } \leq p(J) \leq {O(\eta)},
\end{equation}
where the first inequality uses non-negativity of $p$ 
and the second inequality follows from the bound
$p([i_\ell,i_{\ell + 1})) \leq \eta$.
The {$(\mathcal{I}, \eta, \eps)$-inequalities 
(between $h$ and $p$)}
implies that the inequalities in 
\eqref{eq:lengthen-single} also hold with $h$ in place of $p$.
Now \eqref{eq:lengthen} follows by 
adding \eqref{eq:lengthen-single} across all
$J = [i_\ell, i_{\ell+1})$ such that $J\cap S\neq J\cap T$
(there are at most $2{(d+1)}$ such intervals $J$), 
since each interval $J_j$ in $S$ can change at most two such
$J$'s when lengthened.

Now rewrite $T$ as a 
disjoint union of $s \leq {d+1}$ intervals
$[i_{L_1}, i_{R_1}) \cup \dots \cup [i_{L_s}, i_{R_s})$.
We have
\[ \abs{ p(T) - h(T) } \leq \sum_{j=1}^s \sqrt{R_j - L_j} \cdot \sqrt
\eps\eta \]
by {$(\mathcal{I}, \eta, \eps)$-inequalities between $p$ and $h$}.
Now observing that 
that $0 \leq L_1 \leq R_1 \cdots \leq L_s \leq R_s \leq t =
O((d+1)/\eps)$, we get that the largest possible value of $\sum_{j=1}^s
\sqrt{R_j - L_j}$ is $\sqrt{sz} \leq {\sqrt{{(d+1)}z}}$, 
so the RHS of
(\ref{eq:lengthen}) is at most $O({(d+1)}\eta) + {\sqrt{
{(d+1)}z\eps}\eta}$, as
desired.
\end{proof}

Recall from above that $F$, $f$, $w_\ell$, $c_i$, $y_\ell$, $\tau$
denote the values in the optimal solution.
We claim that 
\begin{equation}
\label{eq:good}
 \norm{ (p+ w) - f }_{\cal A} = O(\eps) ,
\end{equation}
where $w$ is the subdistribution 
which is constant on each $[i_\ell, i_{\ell+1})$
and has weight $w_\ell$ there, so in particular \new{$\normone{w} \leq \tau \leq \opt_{1,d}$.} Indeed, this equality follows by applying~\autoref{lem:ad-dist} with ${h
= f-w}$.
{The lemma requires $h$ and $p$ to satisfy $(\mathcal{I}, \eta,
\eps)$-inequalities, which follows from constraint~\ref{item:phat} ($(\mathcal{I}, \eta,
\eps)$-inequalities between $p+w$ and $f$) and observing that
$(p+ w) - f = p- (f - w)$.
We have also used $\eta = \Theta(\eps/{(d+1)})$ 
to bound the {\rm error} term of the lemma by $O(\eps)$.}

Next, by the triangle inequality we have
{(writing ${\cal A}$ for ${\cal A}_{d+1}$)}
\[
\norm{ r - f }_{\cal A} \leq \norm{ r - (p+w) }_{\cal A} 
+ \norm{ (p+w) - f }_{\cal A}.
\]
The last term on the RHS has just been shown to be $O(\eps)$.
The first term is bounded by
\[ \| r-(p+w)\|_{\cal A} \leq \frac{1}{2}\normone{ r-(p+w) } 
\leq \frac{1}{2}(\normone{r-p} + \normone{w}) \leq \opt_{1,d}. \]
Altogether, we get that $\norm{ r - f }_{\cal A} \leq \opt_{1,d}+ O(\eps)$.

Since $r$ and $f$ are degree $d$ polynomials, $\normone{ r - f } = 2\norm{ r - f }_{\cal A} \leq 2\opt_{1,d}+ O(\eps)$.
This implies $\normone{ p - f } \leq \normone{p-r} + \normone{ r - f } \leq 3\opt_{1,d} +  O(\eps)$.
{Finally, we turn our quasidistribution $f$ which has value $\geq -\eps/2$
everywhere into a distribution $q$ (which is nonnegative), by redistributing
the weight.}
The following simple proposition {bounds the error incurred}.

\begin{proposition} \label{prop:epsmixture}
{Let $f$ and $p$ be any sub-quasidistribution on $I$.}
If $q = {\eps f(I)/\abs{I} + (1- \eps)f}$, then $\norm{q - p}_1 \leq \norm{f
- p}_1 + {\eps(f(I)+p(I))}$.
\end{proposition}

\begin{proof}
  We have
  \[ q - p = {\eps(f(I)/\abs{I} - p) + (1-\eps)(f - p)}. \]
  Therefore
  \[ \normone{ q - p } \leq { \eps \norm{f(I)/|I| - p}_1 + (1-\eps) \norm{ f
    - p }_1 \leq \eps(f(I)+p(I)) + \norm{ f - p }_1 } .
\qedhere \]
\end{proof}

We now have $\normone{ p - q } \leq \normone{ p-f } + O(\eps)$ by~\autoref{prop:epsmixture},
concluding the proof of~\autoref{thm:project:single:poly}.
\end{proof}

\subsection{Proof of~\autoref{lemma:distance:mhr:eff}}\label{app:structural:projection:proofs:mhr}

\lemmaefficientdistancemhr*

\begin{proof}
For convenience, let $\alpha \eqdef \eps^3$; we also write $[i,j]$ instead of $\{i,\dots,j\}$.

First, we note that it is easy to reduce our problem to the case where, in the completeness case, we have $P\in\classmhr$ such that $\normone{\D-P} \leq 2\eps$  and $\kolmogorov{\D}{P} \leq 2\alpha$; while in the soundness case $\lp[1](\D,\classmhr) \geq 99\eps$. Indeed, this can be done with a linear program on $\poly(k,\ell)$ variables, asking to find a $(k+\ell)$-histogram $\D^{\prime\prime}$ on a refinement of $\D$ and $\D^\prime$ minimizing the $\lp[1]$ distance to $\D$, under the constraint that the Kolmogorov distance to $\D^\prime$ be bounded by $\eps$. (In the completeness case, clearly a feasible solution exists, as $P$ is one.) We therefore follow with this new formulation: either
  \begin{enumerate}[\sf(a)]
    \item $\D$ is $\eps$-close to a monotone hazard rate distribution $P$ (in $\lp[1]$ distance) \emph{and} $\D$ is $\alpha$-close to $P$ (in Kolmogorov distance); and
    \item $\D$ is $32\eps$-far from monotone hazard rate
  \end{enumerate} 
where $\D$ is a $(k+\ell)$-histogram.\medskip

We then proceed by observing the following easy fact: suppose $P$ is a MHR distribution on $[n]$, i.e. such that the quantity $h_i \eqdef \frac{P(i)}{\sum_{j=i}^n P(i)}$, $i\in[n]$ is non-increasing. Then, we have
\begin{equation}\label{eq:mhr:parameterization}
  P(i) = h_i \prod_{j=1}^{i-1} (1-h_j), \qquad i\in[n].
\end{equation}
and there is a bijective correspondence between $P$ and $(h_i)_{i\in[n]}$.\medskip

We will write a linear program with variables $y_1,\dots,y_n$, with the correspondence $y_i\eqdef\ln(1-h_i)$. Note that with this parameterization, we get that if the $(y_i)_{i\in[n]}$ correspond to a MHR distribution $P$, then for $i\in[n]$
\[
  P([i,n]) = \prod_{j=1}^{i-1} e^{y_j} = e^{\sum_{j=1}^{i-1} y_j}
\]
and asking that $\ln(1-\eps) \leq \sum_{j=1}^{i-1} y_{j} - \ln \D([i,n]) \leq  \ln(1+\eps)$ amounts to requiring
\[
   P([i,n]) \in [1\pm\eps] \D([i,n]).
\]

We focus first on the completeness case, to provide intuition for the linear program. Suppose there exists $P\in\classmhr$ such $P\in\classmhr$ such that $\normone{\D-P} \leq \eps$  and $\kolmogorov{\D^\prime}{P} \leq \alpha$. This implies that for all $i\in[n]$, $\abs{ P([i,n]) - \D([i,n]) } \leq 2\alpha$. Define $I=\{b+1,\dots,n\}$ to be the longest interval such that $\D(\{b+1,\dots,n\})\leq \frac{\eps}{2}$. It follows that for every $i\in [n]\setminus I$,
\begin{equation}\label{eq:mhr:mult:approx:cdf}
    \frac{P([i,n])}{\D([i,n])} \leq \frac{\D([i,n])+2\alpha}{\D([i,n])} \leq 1+\frac{2\alpha}{\eps/2} = 1+4\eps^2 \leq 1+\eps
\end{equation}
and similarly
$
    \frac{P([i,n])}{\D([i,n])} \geq \frac{\D([i,n])-2\alpha}{\D([i,n]} \geq 1-\eps
$.
This means that for the points $i$ in $[n]\setminus I$, we can write constraints asking for multiplicative closeness (within $1\pm \eps)$ between $e^{\sum_{j=1}^{i-1} y_j}$ and $\D([i,n])$, which is very easy to write down as linear constraints on the $y_i$'s.

\paragraph{The linear program.} 
Let $T$ and $S$ be respectively the sets of ``light'' and ``heavy'' points, defined as $T=\setOfSuchThat{ i\in \{1,\dots,b\} }{ \D(i) \leq \eps^2 }$ and  $S=\setOfSuchThat{ i\in \{1,\dots,b\} }{ \D(i) > \eps^2 }$, where $b$ is as above. (In particular, $\abs{S} \leq 1/\eps^2$.)

\begin{algorithm}
\caption{\label{algo:lp:mhr}Linear Program}
  \begin{align}
  \text{Find }\qquad  & y_1,\dots,y_b & \notag\\
  \text{s.t.}\qquad & \hfill& \notag\\
   &y_i \leq 0 &     \label{lp:mhr:01}\\
   & y_{i+1} \leq y_{i} & \forall i \in \{1,\dots,b-1\}     \label{lp:mhr}\\
   & \ln\!\left(1-\eps\right)  \leq \sum_{j=1}^{i-1} y_{j} - \ln \D([i,n]) \leq  \ln\!\left(1+\eps\right) & \forall i \in \{1,\dots,b\}     \label{lp:mhr:mult:close}\\
   &\frac{\D(i)-\eps_i}{(1+\eps)\D[i,n]} \leq -y_i \leq (1+4\eps)\frac{\D(i)+\eps_i}{(1-\eps)\D[i,n]}      & \forall i\in T  \label{lp:mhr:bound:yi}\\
    &\sum_{i\in T} \eps_i \leq \eps \label{lp:mhr:bound:sum:epsi}\\
    & 0 \leq \eps_i \leq 2\alpha & \forall i\in T \label{lp:mhr:noneg:epsi}\\
    & \ln\left( 1-\frac{\D(i)+2\alpha}{(1-\eps)\D[i,n]} \right) \leq y_i \leq \ln\left( 1-\frac{\D(i)-2\alpha}{(1+\eps)\D[i,n]} \right) & \forall i\in S \label{lp:mhr:heavypoints}
  \end{align}
\end{algorithm}

Given a solution to the linear program above, define $\tilde{P}$ (a non-normalized probability distribution) by setting $\tilde{P}(i) = (1-e^{y_i})e^{\sum_{j=1}^{i-1} y_j}$ for $i\in \{1,\dots,b\}$, and $\tilde{P}(i) = 0$ for $i\in I = \{b+1,\dots, n\}$. A MHR distribution is then obtained by normalizing $\tilde{P}$. 

\paragraph{Completeness.} Suppose $P\in\classmhr$ is as promised. In particular, by the Kolmogorov distance assumption we know that every $i\in T$ has $P(i) \leq \eps^2+2\alpha < 2\eps^2$.
\begin{itemize}
  \item For any $i\in T$, we have that $\frac{P(i)}{P[i,n]} \leq \frac{2\eps^2}{(1-\eps)\eps} \leq 4\eps$, and 
\begin{equation}
  \frac{\D(i)-\eps_i}{(1+\eps)\D[i,n]} \leq \frac{P(i)}{P[i,n]} \leq \underbrace{-\ln(1-\frac{P(i)}{P[i,n]})}_{-y_i}
  \leq (1+4\eps)\frac{P(i)}{P[i,n]} = (1+4\eps)\frac{\D(i)+\eps_i}{P[i,n]} \leq \frac{1+4\eps}{1-\eps}\frac{\D(i)+\eps_i}{\D[i,n]}
\end{equation}
where we used~\autoref{eq:mhr:mult:approx:cdf} for the two outer inequalities; and so~\eqref{lp:mhr:bound:yi},~\eqref{lp:mhr:bound:sum:epsi}, and~\eqref{lp:mhr:noneg:epsi} would follow from setting $\eps_i \eqdef \abs{P(i)-\D(i)}$ (along with the guarantees on $\lp[1]$ and Kolmogorov distances between $P$ and $\D$).
  \item For $i\in S$, Constraint~\eqref{lp:mhr:heavypoints} is also met, as 
  $\frac{P(i)}{P([i,n])} \in \left[\frac{\D(i)-2\alpha}{P([i,n])},\frac{\D(i)+2\alpha}{P([i,n])}\right] 
    \subseteq \left[\frac{\D(i)-2\alpha}{(1+\eps)\D([i,n])},\frac{\D(i)+2\alpha}{(1-\eps)\D([i,n])}\right]$.
\end{itemize}

\paragraph{Soundness.}
\noindent Assume a feasible solution to the linear program is found. We argue that this implies $\D$ is $\bigO{\eps}$-close to some MHR distribution, namely to the distribution obtained by renormalizing $\tilde{P}$.

In order to do so, we bound separately the $\lp[1]$ distance between $\D$ and $\tilde{P}$, from $I$, $S$, and $T$. 
First, $\sum_{i\in I} \abs{\D(i) - \tilde{P}(i)}  = \sum_{i\in I} \D(i) \leq \frac{\eps}{2}$ by construction.
For $i\in T$, we have $\frac{\D(i)}{\D[i,n]} \leq \eps$, and
\begin{align*}
  \tilde{P}(i) = (1-e^{y_i)}) e^{\sum_{j=1}^{i-1} y_j} \in \left[1\pm \eps\right] (1-e^{y_i}) \D([i,n]).
\end{align*}
Now,
\begin{align*}
  1-(1-\eps)\frac{\D(i)-\eps_i}{(1+\eps)\D[i,n]} \geq e^{-\frac{\D(i)-\eps_i}{(1+\eps)\D[i,n]}} 
  \geq e^{y_i} 
  \geq e^{-(1+4\eps)\frac{\D(i)+\eps_i}{(1-\eps)\D[i,n]}} 
  \geq 1-(1+4\eps)\frac{\D(i)+\eps_i}{(1-\eps)\D[i,n]}
\end{align*}
so that
\[
  (1-\eps)\frac{(1-\eps)}{(1+\eps)}(\D(i)-\eps_i)
  \leq
  \tilde{P}(i) 
  \leq
  (1+4\eps)\frac{(1+\eps)}{(1-\eps)}(\D(i)+\eps_i)
\]
which implies
\[
  (1-10\eps)(\D(i)-\eps_i)
  \leq
  \tilde{P}(i) 
  \leq
  (1+10\eps)(\D(i)+\eps_i)
\]
so that $\sum_{i\in T} \abs{\D(i) - \tilde{P}(i)} \leq 10\eps \sum_{i\in T} \D(i) + (1+10\eps)\sum_{i\in T} \eps_i \leq 10\eps + (1+10\eps)\eps \leq 20\eps$
where the last inequality follows from Constraint~\eqref{lp:mhr:bound:sum:epsi}.

To analyze the contribution from $S$, we observe that Constraint~\eqref{lp:mhr:heavypoints} implies that, for any $i\in S$,
\[
    \frac{\D(i)-2\alpha}{(1+\eps)\D([i,n])} \leq \frac{\tilde{P}(i)}{\tilde{P}([i,n])} \leq \frac{\D(i)+2\alpha}{(1-\eps)\D([i,n])}
\]
which combined with Constraint~\eqref{lp:mhr:mult:close} guarantees
\[
    \frac{\D(i)-2\alpha}{(1+\eps)^2\tilde{P}([i,n])} \leq \frac{\tilde{P}(i)}{\tilde{P}([i,n])} \leq \frac{\D(i)+2\alpha}{(1-\eps)^2\tilde{P}([i,n])}
\]
which in turn implies that $\abs{\tilde{P}(i) - \D(i) } \leq 3\eps\tilde{P}(i) + 2\alpha$. Recalling that $\abs{S} \leq \frac{1}{\eps^2}$ and $\alpha=\eps^3$, this yields
$\sum_{i\in S} \abs{\D(i) - \tilde{P}(i)} \leq 3\eps \sum_{i\in S}  \tilde{P}(i)  + 2\eps \leq 3\eps(1+\eps) + 2\eps \leq 8\eps$. Summing up, we get
$
  \sum_{i=1}^n \abs{\D(i) - \tilde{P}(i)} \leq 30\eps
$
which finally implies by the triangle inequality that the $\lp[1]$ distance between $\D$ and the normalized version of $\tilde{P}$ (a valid MHR distribution) is at most $32\eps$.

\paragraph{Running time.} The running time is immediate, from executing the two linear programs on $\poly(n,1/\eps)$ variables and constraints.
\end{proof}

\subsection{Proof of~\autoref{lemma:distance:log:eff}}\label{app:structural:projection:proofs:log}

\lemmaefficientdistancelog*
\begin{proof}
We set $\alpha \eqdef \frac{\eps^2}{\log^2(1/\eps)}$, $\beta \eqdef \frac{\eps^2}{\log(1/\eps)}$, and $\gamma \eqdef \frac{\eps^2}{10}$ (so that $\alpha \ll \beta \ll \gamma \ll \eps$),

Given the explicit description of a distribution $\D$ on $[n]$, which a $k$-histogram over a partition $\mathcal{I}=(I_1,\dots, I_k)$ of $[n]$ with $k=\poly(\log n, 1/\eps)$ \newest{and the explicit description of a distribution $\D^\prime$ on $[n]$}, one must \emph{efficiently} distinguish between:
\begin{enumerate}[\sf(a)]
  \item $\D$ is $\eps$-close to a log-concave $P$ (in $\lp[1]$ distance) \emph{and} $\D^\prime$ is $\alpha$-close to $P$ (in Kolmogorov distance); and
  \item $\D$ is $100\eps$-far from log-concave.
\end{enumerate} 
If we are willing to pay an extra factor of $\bigO{n}$, we can assume without loss of generality that we know the mode of the closest log-concave distribution (which is implicitly assumed in the following: the final algorithm will simply try all possible modes).

\paragraph{Outline.} First, we argue that we can simplify to the case where $\D$ is unimodal. Then, reduce to the case where where $\D$ and $\D^\prime$ are only one distribution, satisfying both requirements from the completeness case. Both can be done efficiently (\autoref{stage:1}), and make the rest much easier.
Then, perform some \emph{ad hoc} partitioning of $[n]$, using our knowledge of $\D$, into $\tildeO{1/\eps^2}$ pieces such that each piece is either a ``heavy'' singleton, or an interval $I$ with weight very close (multiplicatively) to $\D(I)$ \emph{under the target log-concave distribution, if it exists} (\autoref{stage:2}). This in particular simplifies the type of log-concave distribution we are looking for: it is sufficient to look for distributions putting that very specific weight on each piece, up to a $(1+o(1))$ factor. Then, in \autoref{stage:3}, we write and solve a linear program to try and find such a ``simplified'' log-concave distribution, and reject if no feasible solution exists.

Note that the first two sections allow us to argue that instead of additive (in $\lp[1]$) closeness, we can enforce constraints on \emph{multiplicative} (within a $(1+\eps)$ factor) closeness between $\D$ and the target log-concave distribution. This is what enables a linear program with variables being the logarithm of the probabilities, which plays very nicely with the log-concavity constraints. \medskip

\noindent We will require the following result of Chan, Diakonikolas, Servedio, and Sun:
\begin{theorem}[{\cite[Lemma 4.1]{CDSS:13}}]\label{lemma:cdss13:41}
Let $\D$ be a distribution over $[n]$, log-concave and non-decreasing over $\{1,\dots,b\} \subseteq [n]$. Let $a\leq b$ such that
 $\sigma = D(\{1,\dots,a-1\}) > 0$, and write $\tau=\D(\{a,\dots,b\})$. Then 
 		$\frac{\D(b)}{\D(a)} \leq 1+\frac{\tau}{\sigma}$.
\end{theorem}

\subsubsection{Step 1}\label{stage:1}

\paragraph{Reducing to $\D$ unimodal.}
Using a linear program, find a closest \emph{unimodal} distribution $\tilde{\D}$ to $\D$ (also a $k$-histogram on $\mathcal{I}$) under the constraint that $\kolmogorov{\D}{P} \leq \alpha$: this can be done in time $\poly(k)$. If $\normone{\D-\tilde{\D}} > \eps$, output \reject.

\begin{itemize}
  \item If $\D$ is $\eps$-close to a log-concave distribution $P$ as above, then it is in particular $\eps$-close to unimodal and we do not reject. Moreover, by the triangle inequality $\normone{\tilde{\D} - P} \leq 2\eps$ and $\kolmogorov{\tilde{\D}}{P} \leq 2\alpha$.
  \item If $\D$ is $100\eps$-far from log-concave and we do not reject, then $\lp[1](\tilde{\D},\classlog) \geq 99\eps$.
\end{itemize}

\paragraph{Reducing to $\D=\D^\prime$.}
First, we note that it is easy to reduce our problem to the case where, in the completeness case, we have $P\in\classlog$ such that $\normone{\D-P} \leq 4\eps$  and $\kolmogorov{\D}{P} \leq 4\alpha$; while in the soundness case $\lp[1](\D,\classlog) \geq 97\eps$. Indeed, this can be done with a linear program on $\poly(k,\ell)$ variables and constraints, asking to find a $(k+\ell)$-histogram $\D^{\prime\prime}$ on a refinement of $\D$ and $\D^\prime$ minimizing the $\lp[1]$ distance to $\D$, under the constraint that the Kolmogorov distance to $\D^\prime$ be bounded by $2\alpha$. (In the completeness case, clearly a feasible solution exists, as (the flattening on this $(k+\ell)$-interval partition) of $P$ is one.) We therefore follow with this new formulation: either
  \begin{enumerate}[\sf(a)]
    \item $\D$ is $4\eps$-close to a log-concave $P$ (in $\lp[1]$ distance) \emph{and} $\D$ is $4\alpha$-close to $P$ (in Kolmogorov distance); and
    \item $\D$ is $97\eps$-far from log-concave;
  \end{enumerate} 
where $\D$ is a $(k+\ell)$-histogram.\medskip

\noindent This way, we have reduced the problem to a slightly more convenient one, that of~\autoref{stage:2}.

\paragraph{Reducing to knowing the support $[a,b]$.}
The next step is to compute a good approximation of the support of any target log-concave distribution. This is easily obtained in time $O(k)$ as the interval $\{a,\cdots,b\}$ such that
\begin{itemize}
  \item $\D(\{1,\dots,a-1\}) \leq \alpha$ but $\D(\{1,\dots,a\}) > \alpha$; and
  \item $\D(\{b+1,\dots,\}n) \leq \alpha$ but $\D(\{b,\dots,n\}) > \alpha$.
\end{itemize} 
Any log-concave distribution that is $\alpha$-close to $\D$ must include  $\{a,\cdots,b\}$ in its support, since otherwise the $\lp[1]$ distance between $\D$ and $P$ is already greater than $\alpha$. Conversely, if $P$ is a log-concave distribution $\alpha$-close to $\D$, it is easy to see that the distribution obtained by setting $P$ to be zero outside $\{a,\cdots,b\}$ and renormalizing the result is still log-concave, and $O(\alpha)$-close to $\D$.

\subsubsection{Step 2}\label{stage:2}
Given the explicit description of a \emph{unimodal} distribution $\D$ on $[n]$, which a $k$-histogram over a partition $\mathcal{I}=(I_1,\dots, I_k)$ of $[n]$ with $k=\poly(\log n, 1/\eps)$, one must \emph{efficiently} distinguish between:
  \begin{enumerate}[\sf(a)]
    \item $\D$ is $\eps$-close to a log-concave $P$ (in $\lp[1]$ distance) and $\alpha$-close to $P$ (in Kolmogorov distance); and
    \item $\D$ is $24\eps$-far from log-concave,
  \end{enumerate} 
  assuming we know the mode of the closest log-concave distribution, which has support $[n]$.

In this stage, we compute a partition $\mathcal{J}$ of $[n]$ into $\tildeO{1/\eps^2}$ intervals \newer{(here, we implicitly use the knowledge of the mode of the closest log-concave distribution, in order to apply~\autoref{lemma:cdss13:41} differently on two intervals of the support, corresponding to the non-decreasing and non-increasing parts of the target log-concave distribution)}.

As $\D$ is unimodal, we can efficiently ($\bigO{\log k}$) find the interval $S$ of heavy points, that is 
\[
  S\eqdef \setOfSuchThat{ x \in [n] }{  \D(x) \geq \beta }.
\]
Each point in $S$ will form a singleton interval in our partition. 
Let $T\eqdef [n]\setminus S$ be its complement ($T$ is the union of at most two intervals $T_1,T_2$ on which $\D$ is monotone, the head and tail of the distribution). For convenience, we focus on only one of these two intervals, without loss of generality the ``head'' $T_1$ (on which $\D$ is non-decreasing).

\begin{enumerate}
  \item Greedily find $J=\{1,\dots,a\}$, the smallest prefix of the distribution satisfying $\D(J)\in\left[\frac{\eps}{10}-\beta, \frac{\eps}{10}\right]$.
  \item Similarly, partition $T_1\setminus J$ into intervals $I^\prime_1,\dots,I^\prime_s$ (with $s=\bigO{1/\gamma}=\bigO{1/\eps^2}$) such that
    $
      \frac{\gamma}{10} \leq \D(I^\prime_j) \leq \frac{9}{10}\gamma
    $
    for all $1\leq j \leq s-1$, \new{and $\frac{\gamma}{10} \leq \D(I^\prime_s) \leq \gamma$. This is possible as all points not in $S$ have weight less than $\beta$, and $\beta \ll \gamma$.}
\end{enumerate}

\paragraph*{Discussion: why doing this?}\label{ssec:logconcave:completeness}
We focus on the completeness case: let $P\in\classlog$ be a log-concave distribution such that $\normone{\D-P} \leq \eps$ and $\kolmogorov{\D}{P} \leq \alpha$.
Applying~\autoref{lemma:cdss13:41} on $J$ and the $I^\prime_j$'s, we obtain (using the fact that $\abs{P(I^\prime_j) - \D(I^\prime_j)} \leq 2\alpha$) that:
\[
    \frac{\max_{x\in I^\prime_j} P(x)}{\min_{x\in I^\prime_j} P(x)} 
    \leq 1+\frac{\D(I^\prime_j)+2\alpha}{\D(J)-2\alpha} 
    \leq 1 + \frac{\gamma+2\alpha}{\frac{\eps}{10}-2\alpha}
    = 1+ \eps + \bigO{\frac{\eps^2}{\log^2(1/\eps)}} \eqdef 1+\kappa.
\]
Moreover, we also get that each resulting interval $I^\prime_j$ will satisfy
\[
      \D(I^\prime_j)(1-\kappa_j) = \D(I^\prime_j)-2\alpha \leq P(I^\prime_j) \leq \D(I^\prime_j)+2\alpha = \D(I^\prime_j)(1+\kappa_j)
\]
with $\kappa_j \eqdef \frac{2\alpha}{\D(I^\prime_j)} = \bigTheta{1/\log^2(1/\eps)}$.\medskip

Summing up, we have a partition of $[n]$ into $\abs{S}+2 = \tildeO{1/\eps^2}$ intervals such that:
\begin{itemize}
  \item The (at most) two end intervals have $\D(J)\in\left[\frac{\eps}{10}-\beta, \frac{\eps}{10}\right]$, and thus $P(J)\in\left[\frac{\eps}{10}-\beta-2\alpha, \frac{\eps}{10}+2\alpha\right]$;
  \item the $\tildeO{1/\eps^2}$ singleton-intervals from $S$ are points $x$ with $\D(x) \geq \beta$, so that $P(x) \geq \beta -2\alpha \geq \frac{\beta}{2}$;
  \item each other interval $I=I^\prime_j$ satisfies 
  \begin{equation}\label{eq:logconcave:completeness:1}
    (1-\kappa_j) \D(I) \leq P(I) \leq (1+\kappa_j) \D(I)
  \end{equation}
  with $\kappa_j=\bigO{1/\log^2(1/\eps)}$; and
  \begin{equation}\label{eq:logconcave:completeness:2}
  \frac{\max_{x\in I}P(x)}{\min_{x\in I}P(x)} \leq 1+\kappa < 1+\frac{3}{2}\eps.
  \end{equation}
\end{itemize}
We will use in the constraints of the linear program the fact that $(1+\frac{3}{2}\eps)(1+\kappa_j) \leq 1+2\eps$, and $\frac{1-\kappa_j}{1+\frac{3}{2}\eps} \geq \frac{1}{1+2\eps}$.

\subsubsection{Step 3}\label{stage:3}

We start by computing the partition $\mathcal{J}=(J_1,\dots,J_{\ell})$ as in~\autoref{stage:2}; with $\ell=\tildeO{1/\eps^2}$; and write $J_j=\{a_j,\dots,b_j\}$ for all $j\in[\ell]$. We further denote by $S$ and $T$ the set of heavy and light points, following the notations from~\autoref{stage:2}; and let $T^\prime \eqdef T_1\cup T_2$ be the set obtained by removing the two ``end intervals'' (called $J$ in the previous section) from $T$.

\begin{algorithm}
\caption{\label{algo:lp:logconcave}Linear Program}
\begin{align}
\text{Find }\qquad  &x_1,\dots,x_n, \eps_1,\dots,\eps_{\abs{S}} \notag\\
\text{s.t.}\qquad & \hfill& \notag\\
 &x_i \leq 0      \label{lp:01}\\
 &x_{i}-x_{i-1} \geq x_{i+1}-x_{i} & \forall i \in [n]     \label{lp:logconcave}\\
 &-\ln(1+2\eps) \leq x_i - \mu_j \leq \ln(1+2\eps), & \forall j\in T^\prime, \forall i \in J_j     \label{lp:right:weight:light:js}\\
 &-2\frac{\eps_i}{\D(i)} \leq x_i - \ln \D(i) \leq \frac{\eps_i}{\D(i)}, & \forall i\in S    \label{lp:right:weight:heavy:js}\\
 &\sum_{i\in S} \eps_i \leq \eps \label{lp:bound:sum:epsi}\\
 & 0 \leq \eps_i \leq 2\alpha & \forall i\in S \label{lp:noneg:epsi}\\
\end{align}
where $\mu_j\eqdef \ln\frac{\D(J_j)}{\abs{J_j}}$ for $j\in T^\prime$.\medskip
\end{algorithm}

\begin{lemma}[Soundness]\label{lemma:lp:logconcave:soundness}
If the linear program (\autoref{algo:lp:logconcave}) has a feasible solution, then $\lp[1](\D, \classlog)\leq \bigO{\eps}$.
\end{lemma}
\begin{proof}
\noindent A feasible solution to this linear program will define (setting $p_i=e^{x_i}$) a sequence $p=(p_1,\dots,p_n) \in (0,1]^n$ such that
\begin{itemize}
  \item $p$ takes values in $(0,1]$ (from \eqref{lp:01});
  \item $p$ is log-concave (from \eqref{lp:logconcave});
  \item $p$ is ``$(1+O(\eps))$-multiplicatively constant'' on each interval $J_j$ (from \eqref{lp:right:weight:light:js});
  \item $p$ puts roughly the right amount of weight on each $J_i$:
    \begin{itemize}
      \item weight $(1\pm O(\eps))\D(J)$ on every $J$ from $T$ (from~\eqref{lp:right:weight:light:js}), so that the $\lp[1]$ distance between $\D$ and $p$ coming from $T^\prime$ is at most $O(\eps)$;
      \item it puts weight approximately $\D(J)$ on every singleton $J$ from $S$, i.e. such that $\D(J) \geq \beta$. To see why, observe that each $\eps_i$ is in $[0,2\alpha]$ by constraints~\eqref{lp:noneg:epsi}. In particular, this means that $\frac{\eps_i}{\D(i)} \leq 2\frac{\alpha}{\beta} \ll 1$, and 
      we have
      \[
           \D(i) - 4\eps_i \leq \D(i)\cdot e^{-4\frac{\eps_i}{\D(i)}} \leq p_i = e^{x_i} \leq \D(i)\cdot e^{2\frac{\eps_i}{\D(i)}} \leq \D(i)+4\eps_i
      \]
      
       and together with~\eqref{lp:bound:sum:epsi} this guarantees that the $\lp[1]$ distance between $\D$ and $p$ coming from $S$ is at most $\eps$.
    \end{itemize}
\end{itemize}
Note that the solution obtained this way may not sum to one -- i.e., is not necessarily a probability distribution. However, it is easy to renormalize $p$ to obtain a \emph{bona fide} probability distribution $\tilde{P}$ as follows: set $\tilde{P} = \frac{p(i)}{\sum_{i\in S\cup T^\prime} p(i)}$ for all $i\in S\cup T^\prime$, and $p(i) =0$ for $i\in T\setminus T^\prime$.

Since by the above discussion we know that $p(S\cup T^\prime)$ is within $\bigO{\eps}$ of $\D(S\cup T^\prime)$ (itself in $[1-\frac{9\eps}{5}, 1+\frac{9\eps}{5}]$ by construction of $T^\prime$), $\tilde{P}$ is a log-concave distribution such that $\normone{\tilde{P}-\D} = \bigO{\eps}$.
\end{proof}

\begin{lemma}[Completeness]\label{lemma:lp:logconcave:completeness}
If There is $P$ in $\classlog$ such that $\normone{\D-P}\leq \eps$ and $\kolmogorov{\D}{P}\leq \alpha$, then the linear program (\autoref{algo:lp:logconcave}) has a feasible solution.
\end{lemma}
\begin{proof}
 Let $P\in\classlog$ such that $\normone{\D - P}\leq \eps$ and $\kolmogorov{\D}{P}\leq \alpha$.
 Define $x_i\eqdef \ln P(i)$ for all $i\in[n]$. Constraints~\eqref{lp:01} and~\eqref{lp:logconcave} are immediately satisfied, since $P$ is log-concave. By the discussion from~\autoref{ssec:logconcave:completeness} (more specifically, Eq.~\eqref{eq:logconcave:completeness:1} and~\eqref{eq:logconcave:completeness:2}), constraint~\eqref{lp:right:weight:light:js} holds as well.
 
 Letting $\eps_i\eqdef \abs{P(i)-\D(i)}$ for $i\in S$, we also immediately have~\eqref{lp:bound:sum:epsi} and~\eqref{lp:noneg:epsi} (since $\normone{P-\D} \leq \eps$ and $\kolmogorov{\D}{P}\leq \alpha$ by assumption). Finally, to see why~\eqref{lp:right:weight:heavy:js} is satisfied, we rewrite
 \[
      x_i - \ln\D(i) = \ln\frac{P(i)}{\D(i)} =  \ln\frac{\D(i)\pm\eps_i}{\D(i)} = \ln(1\pm \frac{\eps_i}{\D(i)})
 \] 
 and use the fact that $\ln(1+x) \leq x$ and $\ln(1-x) \geq -2x$ (the latter for $x < \frac{1}{2}$, along with $\frac{\eps_i}{\D(i)} \leq \frac{2\alpha}{\beta} \ll 1$).
\end{proof}

\subsubsection{Putting it all together: Proof of~\autoref{lemma:distance:log}}

The algorithm is as follows (keeping the notations from~\autoref{stage:1} to~\autoref{stage:3}):
\begin{itemize}
  \item Set $\alpha,\beta,\gamma$ as above.
  \item Follow~\autoref{stage:1} to reduce it to the case where $\D$ is unimodal and satisfies the conditions for Kolmogorov and $\lp[1]$ distance; and a good $[a,b]$ approximation of the support is known
  \item For each of the $\bigO{n}$ possible modes $c\in[a,b]$:
  \begin{itemize}
    \item Run the linear program~\autoref{algo:lp:logconcave}, return \accept if a feasible solution is found
  \end{itemize}
  \item None of the linear programs was feasible: return \reject.
\end{itemize}

The correctness comes from~\autoref{lemma:lp:logconcave:soundness} and~\autoref{lemma:lp:logconcave:completeness} and the discussions in~\autoref{stage:1} to~\autoref{stage:3}; as for the claimed running time, it is immediate from the algorithm and the fact that the linear program executed each step has $\poly(n,1/\eps)$ constraints and variables.

\end{proof}
 
\makeatletter{}\section{Proof of~\autoref{theo:samp:binomial:tolerant:lb}}\label{app:binomial:tolerant:lb}

In this section, we establish our lower bound for tolerant testing of the Binomial distribution, restated below:

\binomialtolerantlbtheorem*
\noindent The theorem will be a consequence of the (slightly) more general result below:

\begin{theorem}\label{theo:samp:binomial:tolerant}
  There exist absolute constants $\eps_0> 0$ and $\lambda > 0$ such that the following holds. Any algorithm which, given $\SAMP$ access to an unknown distribution $\D$ on $\domain$ and parameter $\eps \in (0,\eps_0)$, distinguishes with probability at least $2/3$ between \textsf{(i)} $\normone{\D-\binomial{n}{\frac{1}{2}}} \leq \eps$ and \textsf{(ii)} $\normone{\D-\binomial{n}{\frac{1}{2}}} \geq \lambda\eps^{1/3}-\eps$ must use $\bigOmega{\eps\frac{\sqrt{n}}{\log(\eps^{2/3} n)}}$ samples.
\end{theorem}
\noindent By choosing a suitable $\phi$ and working out the corresponding parameters, this for instance enables us to derive the following:
\begin{corollary}\label{coro:samp:binomial:tolerant}
  There exists an absolute constant $\eps_0 \in(0,1/1000)$ such that the following holds. Any algorithm which, given $\SAMP$ access to an unknown distribution $\D$ on $\domain$, distinguishes with probability at least $2/3$ between \textsf{(i)} $\normone{\D-\binomial{n}{\frac{1}{2}}} \leq \eps_0$ and \textsf{(ii)} $\normone{\D-\binomial{n}{\frac{1}{2}}} \geq 100\eps_0$ must use $\bigOmega{\frac{\sqrt{n}}{\log n}}$ samples.
\end{corollary}
\noindent By standard techniques, this will in turn imply~\autoref{theo:samp:binomial:tolerant:lb}.

\begin{proofof}{\autoref{theo:samp:binomial:tolerant}}
Hereafter, we write for convenience $B_n\eqdef \binomial{n}{\frac{1}{2}}$. To prove this lower bound, we will rely on the following:
\begin{theorem}[{\cite[Theorem 1]{ValiantValiant:10lb}}]\label{theo:samp:uniformity:tolerant}
  For any constant $\phi\in(0,1/4)$, following holds. Any algorithm which, given $\SAMP$ access to an unknown distribution $\D$ on $\domain$, distinguishes with probability at least $2/3$ between \textsf{(i)} $\normone{\D - \uniform_n} \leq \phi$ and \textsf{(ii)} $\normone{\D - \uniform_n} \geq \frac{1}{2}-\phi$, must have sample complexity at least $\frac{\phi}{32}\frac{n}{\log n}$.
\end{theorem}

Without loss of generality, assume $n$ is even (so that $B_n$ has only one mode located at $\frac{n}{2}$). For $c>0$, we write $I_{n,c}$ for the interval $\{\frac{n}{2}-c\sqrt{n},\dots,\frac{n}{2}+c\sqrt{n}\}$ and $J_{n,c}\eqdef\domain\setminus I_{n,c}$.
\begin{fact}\label{fact:easy:binomial}
For any $c > 0$, \[
\frac{B_n(\frac{n}{2} + c\sqrt{n})}{B_n({n}/{2})}, \frac{B_n(\frac{n}{2} - c\sqrt{n})}{B_n({n}/{2})} \operatorname*{\sim}_{n\to\infty} e^{-2c^2} \]
and 
\[
B_n(I_{n,c}) \in (1\pm o(1))\cdot[e^{-2c^2},1]\cdot 2c\sqrt{\frac{2}{\pi}} = \bigTheta{c}\,.
\]
\end{fact}

The reduction proceeds as follows: given sampling access to $\D$ on $[n]$, we can simulate sampling access to a distribution $\D^\prime$ on $[N]$ (where $N=\bigTheta{n^2}$) such that
\begin{itemize}
  \item if $\normone{\D - \uniform_n} \leq \phi$, then $\normone{\D^\prime - B_N} < \eps$;
  \item if $\normone{\D - \uniform_n} \geq \frac{1}{2}-\phi$, then $\normone{\D^\prime - B_N} > \eps^\prime - \eps$
\end{itemize}
for $\eps \eqdef \Theta(\phi^{3/2})$ and $\eps^\prime \eqdef \Theta(\phi^{\frac{1}{2}})$;
in a way that preserves the sample complexity.\medskip

More precisely, define $c\eqdef \sqrt{2\ln\frac{1}{1-\phi}} =\bigTheta{\sqrt{\phi}}$ (so that $\phi = 1 - e^{-2c^2}$) and $N$ such that $\abs{I_{N,c}}= n$ (that is, $N=(n/(2c))^2 = \bigTheta{n^2/\phi}$). From now on, we can therefore identify $[n]$ to $I_{N,c}$ in the obvious way, and see a draw from $\D$ as an element in $I_{N,c}$.

Let $p\eqdef B_N(I_{N,c}) = \bigTheta{\sqrt{\phi}}$, and $B_{N,c}$, $\bar{B}_{N,c}$ respectively denote the conditional distributions induced by $B_N$ on $I_{N,c}$ and $J_{N,c}$. Intuitively, we want $\D$ to be mapped to the conditional distribution of $\D^\prime$ on $I_{N,c}$, and the conditional distribution of $\D^\prime$ on $J_{N,c}$ to be exactly $\bar{B}_{N,c}$. This is done as by defining $\D^\prime$ by the process below:
\begin{itemize}
  \item with probability $p$, we draw a sample from $\D$ (seen as an element of $I_{N,c}$;
  \item with probability $1-p$, we draw a sample from $\bar{B}_{N,c}$.
\end{itemize}

Let $\tilde{B}_N$ be defined as the distribution which exactly matches $B_N$ on $J_{n,c}$, but is uniform on $I_{n,c}$:
\begin{align*}
  \tilde{B}_N(i) = \begin{cases}
    \frac{p}{\abs{I_{n,c}}} & i\in I_{n,c}\\
    B_N(i) & i\in J_{n,c}\\
     \end{cases}
\end{align*}
From the above, we have that $\normone{\D^\prime - \tilde{B}_N} = p\cdot \normone{\D - \uniform_n}$. Furthermore, by \autoref{fact:easy:binomial}, \autoref{lemma:small:l2:close:uniform:l1} and the definition of $I_{n,c}$, we get that $\normone{B_N - \tilde{B}_N} = p\cdot \normone{(B_N)_{I_{n,c}} - \uniform_{I_{n,c}}} \leq p\cdot \phi$. Putting it all together,

\begin{itemize}
  \item If $\normone{\D - \uniform_n} \leq \phi$, then by the triangle inequality $\normone{\D^\prime - B_N} \leq p(\phi + \phi) = 2p\phi$;
  \item If $\normone{\D - \uniform_n} \geq \frac{1}{2}-\phi$, then similarly $\normone{\D^\prime - B_N} \geq p(\frac{1}{2}-\phi -\phi) = \frac{p}{4}-2p\phi$.
\end{itemize}
Recalling that $p= \bigTheta{\sqrt{\phi}}$ and setting $\eps \eqdef 2p\phi$ concludes the reduction. From \autoref{theo:samp:uniformity:tolerant}, we conclude that
\[
\frac{\phi}{32}\frac{n}{\log n}
  = \bigOmega{\phi\frac{\sqrt{\phi N}}{\log(\phi N)}} 
  = \bigOmega{\eps\frac{\sqrt{N}}{\log(\eps^{2/3} N)}}
\]
samples are necessary.

\end{proofof}
\begin{proofof}{\autoref{coro:samp:binomial:tolerant}}
The corollary follows from the proof of~\autoref{coro:samp:binomial:tolerant}, by taking $\phi=1/1000$ and computing the corresponding $\eps$ and $\eps^\prime-\eps$\xspace to check that indeed $\lim_{n\to\infty}\frac{\eps^\prime - \eps}{\eps} > 100$.
\end{proofof}
 
\end{document}